%% file: nnfbnpgraphv5.tex
\documentclass[preprint]{imsart}
\pdfoutput=1
\usepackage{amsfonts}
\usepackage{amsmath}
\usepackage{amssymb}
\usepackage{graphicx}
\usepackage{tikz}
\usepackage{pgfplots}
\usepackage{subfigure}
\usepackage[round]{natbib}
\usepackage[tmargin=2.5cm,bmargin=2.5cm,lmargin=2.5cm,rmargin=2.5cm]{geometry}
\usepackage{algorithm}
\usepackage{epstopdf}

\providecommand{\U}[1]{\protect\rule{.1in}{.1in}}
\usetikzlibrary{arrows}
\usetikzlibrary{positioning}
\usetikzlibrary{fit}
\usetikzlibrary{calc}
\providecommand{\U}[1]{\protect\rule{.1in}{.1in}}
\newtheorem{theorem}{Theorem}

\newtheorem{lemma}[theorem]{Lemma}

\newtheorem{proposition}[theorem]{Proposition}
\newtheorem{remark}[theorem]{Remark}

\newenvironment{proof}[1][Proof]{\noindent\textbf{#1.} }{\ \rule{0.5em}{0.5em}}
\graphicspath{{./}{./figures/}}

\usepackage{dsfont} % \mathds{1}
\usepackage[group-separator={,}]{siunitx} % \num{100000} numbers with comma thousands separator
\usepackage{adjustbox}
\usepackage{url}
\usepackage{hyperref}

\usepackage{float}
\usepackage{array}
\providecommand{\subfloat}{\subfigure}

\DeclareMathOperator{\CRM}{CRM}

\DeclareMathOperator{\Poi}{Poisson}
\DeclareMathOperator{\tPoi}{tPoisson}
\DeclareMathOperator{\Gam}{Gamma}
\DeclareMathOperator{\Ber}{Ber}
\DeclareMathOperator{\Mult}{Multinomial}

\newcommand{\1}[1]{\mathds{1}_{#1}}

\begin{document}
\title{Exchangeable Random Measures for Sparse and Modular Graphs with Overlapping Communities}
\begin{aug}
\author{\fnms{Adrien} \snm{Todeschini}\ead[label=e2]{Adrien.Todeschini@inria.fr}},
\author{\fnms{Xenia} \snm{Miscouridou}\ead[label=e3]{xenia.miscouridou@spc.ox.ac.uk}}
\and
\author{\fnms{Fran\c{c}ois} \snm{Caron}\ead[label=e1]{caron@stats.ox.ac.uk}}

\runauthor{A. Todeschini, X. Miscouridou and F. Caron}
\affiliation{INRIA and University of Oxford}

\address{INRIA \& Institut de Math\'ematiques de Bordeaux, France\\
%200 avenue de la vieille tour\\
%33\\
\printead{e2}}
\address{Department of Statistics, University of Oxford, UK\\
%1 South Parks Road\\
%Oxford, OX1 3TG United Kingdom\\
\printead{e3,e1}}
\runtitle{Exchangeable Random Measures for Sparse and Modular Graphs}
\end{aug}
\begin{abstract}
We propose a novel statistical model for sparse networks with overlapping community structure. The model is based on representing the graph as an exchangeable point process, and naturally generalizes existing probabilistic models with overlapping block-structure to the sparse regime. Our construction builds on vectors of completely random measures, and has interpretable parameters, each node being assigned a vector representing its level of affiliation to some latent communities. We develop methods for simulating this class of random graphs, as well as to perform posterior inference. We show that the proposed approach can recover interpretable structure from two real-world networks and can handle graphs with thousands of nodes and tens of thousands of edges.
\end{abstract}

\maketitle

\section{Introduction}
\label{sec:intro}

There has been a growing interest in the analysis, understanding and modeling
of network data over the recent years. A network is composed of a set of
nodes, or vertices, with connections between them. Network data arise in a
wide range of fields, and include social networks, collaboration networks,
communication networks, biological networks, food webs and are a useful way of representing interactions between sets of objects. Of particular
importance is the elaboration of random graph models, which can capture the
salient properties of real-world graphs. Following the seminal work of
\cite{Erdos1959}, various network models have been
proposed; see the overviews of
\cite{Newman2003,Newman2009}, \cite{Kolaczyk2009}, \cite{Bollobas2001}, \cite{Goldenberg2010}, \cite{Fienberg2012} or \cite{Jacobs2014}.
In particular, a large body of the literature has concentrated on models that can capture some modular or community structure within the network. The first statistical network model in this line of research is the popular stochastic block-model~\citep{Holland1983,Snijders1997,Nowicki2001}. The stochastic block-model assumes that each node belongs to one of $p$ latent communities, and the probability of connection between two nodes is given by a $p\times p$ connectivity matrix. This model has been extended in various directions, by introducing degree-correction parameters~\citep{Karrer2011}, by allowing the number of communities to grow with the size of the network~\citep{Kemp2006}, or by considering overlapping communities~\citep{Airoldi2008,Miller2009,Latouche2011,Palla2012,Yang2013}. Stochastic block-models and their extensions have shown to offer a very flexible modeling framework, with interpretable parameters, and have been successfully used for the analysis of numerous real-world networks. However, as outlined by~\cite{Orbanz2015}, when one makes the usual assumption that the ordering of the nodes is irrelevant in the definition of the statistical network model, the Bayesian probabilistic versions of those models lead to dense networks\footnote{We refer to graphs  whose number of edges scales quadratically with the number of nodes as dense, and sparse if it scales sub-quadratically.}: that means that the number of edges grows quadratically with the number of nodes. This property is rather undesirable, as many real-world networks are believed to be sparse. \medskip

Recently, \cite{Caron2014} proposed an alternative framework for statistical network modeling. The framework is based on representing the graph as an exchangeable
random measure on the plane. More precisely, the nodes are embedded at
some location $\theta_{i}\in\mathbb{R}_{+}$ and, for simple graphs, a connection exists between
two nodes $i$ and $j$ if there is a point at locations $(\theta_{i},\theta_{j})$ and
$(\theta_{j},\theta_{i})$. An undirected simple graph is therefore represented by a symmetric
point process $Z$ on the plane
\begin{equation}
Z=\sum_{i,j}z_{ij}\delta_{(\theta_{i},\theta_{j})}\label{eq:pointprocessZ}%
\end{equation}
where $z_{ij}=z_{ji}=1$ if $i$ and $j$ are connected, 0 otherwise; see Figure~\ref{fig:pointprocess} for an illustration. \cite{Caron2014} noted that jointly exchangeable random measures, a notion to be
 defined in Eq.~\eqref{eq:kallenbergexchangeability}, admit a representation theorem due to
\cite{Kallenberg1990}, providing a general construction for exchangeable random measures hence random graphs represented by such objects. This connection is further explored by \cite{Veitch2015} and \cite{Borgs2016}, who provide a detailed description and extensive theoretical analysis of the associated class of random graphs, which they name \textit{Kallenberg exchangeable graphs} or \textit{graphon processes}. Within this class of models, \cite{Caron2014} consider in particular the following simple generative model, where two nodes $i\neq j$ connect
with probability
\begin{equation}
\Pr(z_{ij}=1|(w_\ell)_{\ell=1,2,\ldots})=1-e^{-2w_{i}w_{j}} \label{eq:link1}%
\end{equation}
where the $(w_i,\theta_i)_{i=1,2,\ldots}$ are the points of a Poisson
point process on $\mathbb{R}_{+}^{2}$. The parameters $w_i>0$ can be interpreted as sociability parameters.  Depending on the properties of the mean measure of the Poisson process, the authors show that it is possible to generate both dense and sparse graphs, with potentially heavy-tailed degree distributions, within this framework. The construction
\eqref{eq:link1} is however rather limited in terms of capturing structure
in the network. \cite{Herlau2015} proposed an extension of  \eqref{eq:link1}, which can accommodate a community structure. More precisely, introducing latent community membership variables $c_i\in\{1,\ldots,p\}$, two nodes $i\neq j$ connect with probability
\begin{equation}
\Pr(z_{ij}=1|(w_\ell,c_\ell)_{\ell=1,2,\ldots},(\eta_{k\ell})_{1\leq k,\ell\leq p})=1-e^{-2\eta_{c_{i}c_j} w_{i}w_{j}} \label{eq:linkHerlau}%
\end{equation}
where the $(w_i,{c_i},\theta_i)_{i=1,2,\ldots}$ are the points of a (marked) Poisson
point process on $\mathbb{R}_{+}\times \{1,\ldots,p\}\times\mathbb{R}_{+}$ and $\eta_{k\ell}$ are positive random variables parameterizing the strength of interaction between nodes in community $k$ and nodes in community $\ell$. The model is similar in spirit to the degree-corrected stochastic block-model~\citep{Karrer2011}, but within the point process framework~\eqref{eq:pointprocessZ}, and can thus accommodate both sparse and dense networks with community structure. The model of \cite{Herlau2015} however shares the limitations of the (degree-corrected) stochastic block-model, in the sense that it cannot model overlapping community structures, each node being assigned to a single community; see~\cite{Latouche2011} and \cite{Yang2013} for more discussion along these lines. Other extensions with block structure or mixed membership block structure are also suggested by \cite{Borgs2016}.%, but without providing details on how to develop scalable algorithms in order to simulate or perform inference for such models.
\bigskip

In this paper, we consider that each node $i$ is assigned a set of latent non-negative
parameters $w_{ik}$, $k=1,\ldots,p$, and that the probability that two nodes
$i\neq j$ connect is given by
\begin{equation}
\Pr(z_{ij}=1|(w_{\ell 1},\ldots,w_{\ell p})_{\ell=1,2,\ldots})=1-e^{-2\sum_{k=1}^{p}w_{ik}w_{jk}}.\label{eq:link2}%
\end{equation}
These non-negative weights can be interpreted as measuring the level of affiliation of node $i$ to the latent communities $k=1,\ldots,p$. For example, in a friendship network, these communities can correspond to colleagues, family, or sport partners, and the weights measure the level of affiliation of an individual to each community. Note that as individuals can have high weights in different communities, the model can capture overlapping communities. The link
probability (\ref{eq:link2}) builds on a non-negative factorization; it has been used by other authors for network modeling \citep{Yang2013,Zhou2015} and is also closely related to the model for multigraphs of~\cite{Ball2011}.
The main contribution of this paper is to use the link probability  \eqref{eq:link2}
within the point process framework of \cite{Caron2014}. To this aim, we consider
that the node locations and weights $(w_{i1},\ldots,w_{ip},\theta_i)_{i=1,2,\ldots}$ are drawn from a
Poisson point process on $\mathbb R_+^{p+1}$ with a given mean measure $\nu$. The construction of such multivariate point process relies on vectors of completely random measures (or equivalently multivariate subordinators). In particular, we build on the flexible though tractable construction recently introduced by \cite{Griffin2014}.

The proposed model generalizes that of \cite{Caron2014} by allowing the model to capture more structure in the network, while retaining its main features, and is shown to have the following properties:
\begin{itemize}
\item Interpretability: each node is assigned a set of positive parameters, which can be interpreted as measuring the levels of affiliation of a node to latent communities; once those parameters are learned, they can be used to undercover the latent structure in the network.
\item Sparsity: we can generate graphs whose number of edges grows subquadratically with the number of nodes.
\item Exchangeability: in the sense of \cite{Kallenberg1990}.
\end{itemize}
Additionally, we develop a Markov chain Monte Carlo (MCMC) algorithm for posterior inference with this model, and show experiments on two real-world networks with a thousand of nodes and tens of thousands of edges.\medskip

The article is organized as follows. The class of random graph models is introduced in Section~\ref{sec:model}. Properties of the class of graphs and simulation are described in Section~\ref{sec:properties}. We derive a scalable MCMC algorithm for posterior inference in Section~\ref{sec:inference}. In Section~\ref{sec:experiments} we provide illustrations of the proposed method on simulated data and on two networks: a network of citations between political blogs and a network of connections between US airports. We show that the approach is able to discover interpretable structure in the data and performs well compared to alternatives.

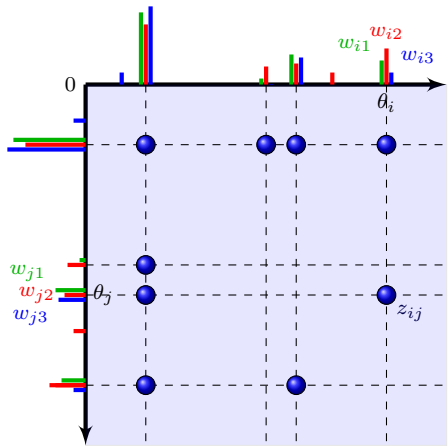
\begin{figure}[t]
\begin{center}
\begin{tikzpicture}[node distance=1.4cm, auto,>=latex',scale=.8]
\draw[->,ultra thick] node[left] {0} (0,0) -- (6,0) node[right] {};
\draw[->,ultra thick] (0,0) -- (0,-6) node[right] {};
\draw[fill=blue,opacity=0.1] (0,0) -- (0,-6) -- (6,-6) -- (6,0)  ;
% Draw the grid
\draw[dashed] (1,0) -- (1,-6) node[right] {};
\draw[dashed] (3,0) -- (3,-6) node[right] {};
\draw[dashed] (3.5,0) -- (3.5,-6) node[right] {};
\draw[dashed] (5,0) -- (5,-6) node[right] {};
\draw[dashed] (0,-1) -- (6,-1) node[right] {};
\draw[dashed] (0,-3) -- (6,-3) node[right] {};
\draw[dashed] (0,-3.5) -- (6,-3.5) node[right] {};
\draw[dashed] (0,-5) -- (6,-5) node[right] {};
% draw the point process
\node[draw,circle,inner sep=2.5pt,fill=blue,shading=ball] at (1,-1) {};
\node[draw,circle,inner sep=2.5pt,fill=blue,shading=ball] at (1,-5) {};
\node[draw,circle,inner sep=2.5pt,fill=blue,shading=ball] at (3,-1) {};
\node[draw,circle,inner sep=2.5pt,fill=blue,shading=ball] at (3.5,-1) {};
\node[draw,circle,inner sep=2.5pt,fill=blue,shading=ball] at (5,-1) {};
\node[draw,circle,inner sep=2.5pt,fill=blue,shading=ball] at (1,-3) {};
\node[draw,circle,inner sep=2.5pt,fill=black,shading=ball] at (1,-3.5) {};
\node[draw,circle,inner sep=2.5pt,fill=black,shading=ball] at (5,-3.5) {};
\node[draw,circle,inner sep=2.5pt,outer color=black,shading=ball] at (3.5,-5) {};
\node[] at (5.4,-3.8) {$\color{blue!30!black}z_{ij}$};
\node[red] at (5,-.3) {$\color{black}\theta_{i}$};
\node[red] at (.3,-3.5) {$\color{black}\theta_{j}$};
% Draw the increments of the Lévy process
\draw[-,ultra thick,green!70!black] (.92,0) -- (.92,1.2) node[right] {};
\draw[-,ultra thick,red] (1,0) -- (1,1) node[right] {};
\draw[-,ultra thick,blue] (1.08,0) -- (1.08,1.3) node[right] {};
\draw[-,ultra thick,green!70!black] (2.92,0) -- (2.92,.1) node[right] {};
\draw[-,ultra thick,red] (3,0) -- (3,.3) node[right] {};
\draw[-,ultra thick,blue] (3.08,0) -- (3.08,.01) node[right] {};
\draw[-,ultra thick,green!70!black] (3.42,0) -- (3.42,.5) node[right] {};
\draw[-,ultra thick,red] (3.5,0) -- (3.5,.35) node[right] {};
\draw[-,ultra thick,blue] (3.58,0) -- (3.58,.45) node[right] {};
\draw[-,ultra thick,green!70!black] (4.92,0) -- (4.92,.4) node[above left] {$\color{green!70!black}w_{i1}$};
\draw[-,ultra thick,red] (5,0) -- (5,.6) node[above,distance=1cm] {$\color{red}w_{i2}$};
\draw[-,ultra thick,blue] (5.08,0) -- (5.08,.2) node[above right] {$\color{blue}w_{i3}$};
\draw[-,ultra thick,blue] (.6,0) -- (.6,.2) node[right] {};
%\draw[-,ultra thick,red] (1.2,0) -- (1.2,.1) node[right] {};
\draw[-,ultra thick,red] (4.1,0) -- (4.1,.2) node[above] {};
% Symmetric increments
\draw[-,ultra thick,green!70!black] (0,-.92) -- (-1.2,-.92) node[right] {};
\draw[-,ultra thick,red] (0, -1) -- (-1,-1) node[right] {};
\draw[-,ultra thick,blue] (0,-1.08) -- (-1.3,-1.08) node[right] {};
\draw[-,ultra thick,green!70!black] (0,-2.92) -- (-.1,-2.92) node[right] {};
\draw[-,ultra thick,red] (0,-3) -- (-.3,-3) node[right] {};
\draw[-,ultra thick,green!70!black] (0,-4.92) -- (-.4,-4.92);
\draw[-,ultra thick,red] (0,-5) -- (-.6,-5);
\draw[-,ultra thick,blue] (0,-5.08) -- (-.2,-5.08);
\draw[-,ultra thick,green!70!black] (0,-3.42) -- (-.5,-3.42) node[above left] {$\color{green!70!black}w_{j1}$};
\draw[-,ultra thick,red] (0,-3.5) -- (-.35,-3.5) node[left,distance=2cm] {$\color{red}w_{j2}$};
\draw[-,ultra thick,blue] (0,-3.58) -- (-.45,-3.58) node[below left] {$\color{blue}w_{j3}$};
\draw[-,ultra thick,red] (0,-5) -- (-.5,-5) node[right] {};
\draw[-,ultra thick,blue] (0,-.6) -- (-.2,-.6) node[right] {};
%\draw[-,ultra thick,red] (0,-1.2) -- (-.1,-1.2) node[right] {};
\draw[-,ultra thick,red] (0,-4.1) -- (-.2,-4.1) node[right] {};
\end{tikzpicture}
\end{center}
\caption{Representation of a undirected graph via a point process $Z$. Each node $i$ is
embedded in $\mathbb{R}_{+}$ at some location $\theta_{i}$ and is associated
with a set of positive attributes $(w_{i1},\ldots,w_{ip})$. An edge between
nodes $\theta_{i}$ and $\theta_{j}$ is represented by a point at locations
$(\theta_{i},\theta_{j})$ and $(\theta_{j},\theta_{i})$ in $\mathbb{R}_{+}%
^{2}$. }%
\label{fig:pointprocess}%
\end{figure}

\section{Sparse graph models with overlapping communities}
\label{sec:model}

In this section, we present the statistical model for simple graphs. The construction builds on vectors of completely random measures \citep[CRM,][]{Kingman1967}. We only provide here the necessary material for the definition of the network model; please refer to Appendix~\ref{sec:background} for additional background on vectors of CRMs. The model described in this section can also be extended to bipartite graphs; see Appendix~\ref{sec:app:bipartite}.

\subsection{General construction using vectors of CRMs}
\label{subsec:general}

We consider that each node $i$ is embedded at some location $\theta_i\in\mathbb R_+$, and has some set of positive weights $(w_{i1},\ldots,w_{ip})\in \mathbb R_+^p$. The points $(w_{i1},\ldots,w_{ip},\theta_i)_{i=1,\ldots,\infty}$  are assumed to be drawn from a Poisson process with mean measure
\begin{equation}
\nu(dw_{1},\ldots,dw_{p},d\theta)=\rho(dw_{1},\ldots,dw_{p})\lambda(d\theta)\label{eq:nu}
\end{equation}
where $\lambda$ is the Lebesgue measure and $\rho$ is a $\sigma$-finite measure on $\mathbb R^p_+$, concentrated on $\mathbb R^p_+\backslash \{\mathbf 0\} $, which satisfies
\begin{equation}
\int_{\mathbb R_+^p} \min\left (1,\sum_{k=1}^p w_k\right )\rho(dw_1,\ldots,dw_p)<\infty.
\label{eq:conditionLevydef}
\end{equation}
Under this condition~\citep{Skorohod1991,Barndorff-Nielsen2001}, we can describe the set of weights and locations using a vector of completely random measures $(W_1,\ldots,W_p)$ on $\mathbb R_+$:
\begin{equation}\label{eq:Wk}
W_k=\sum_{i=1}^\infty w_{ik}\delta_{\theta_i}, \text{ for }k=1,\ldots,p.
\end{equation}
We simply write
\begin{equation}
(W_1,\ldots,W_p)\sim \text{CRM}(\rho,\lambda).\label{eq:vCRM}
\end{equation}

\begin{figure}[ptb]
\begin{center}%
\begin{tabular}
[c]{ccc}%
\begin{tabular}[c]{c}
\includegraphics[width=.3\textwidth]{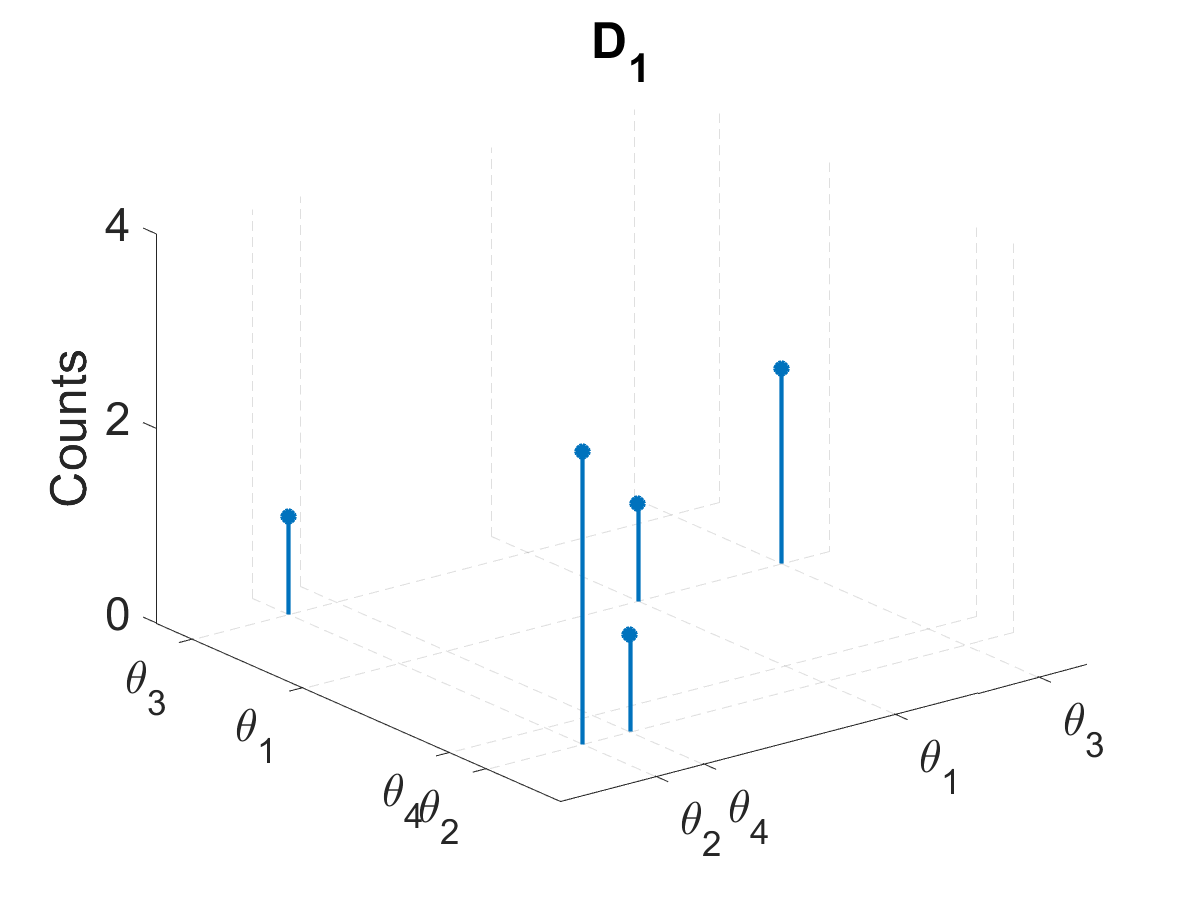}\\
\includegraphics[width=.3\textwidth]{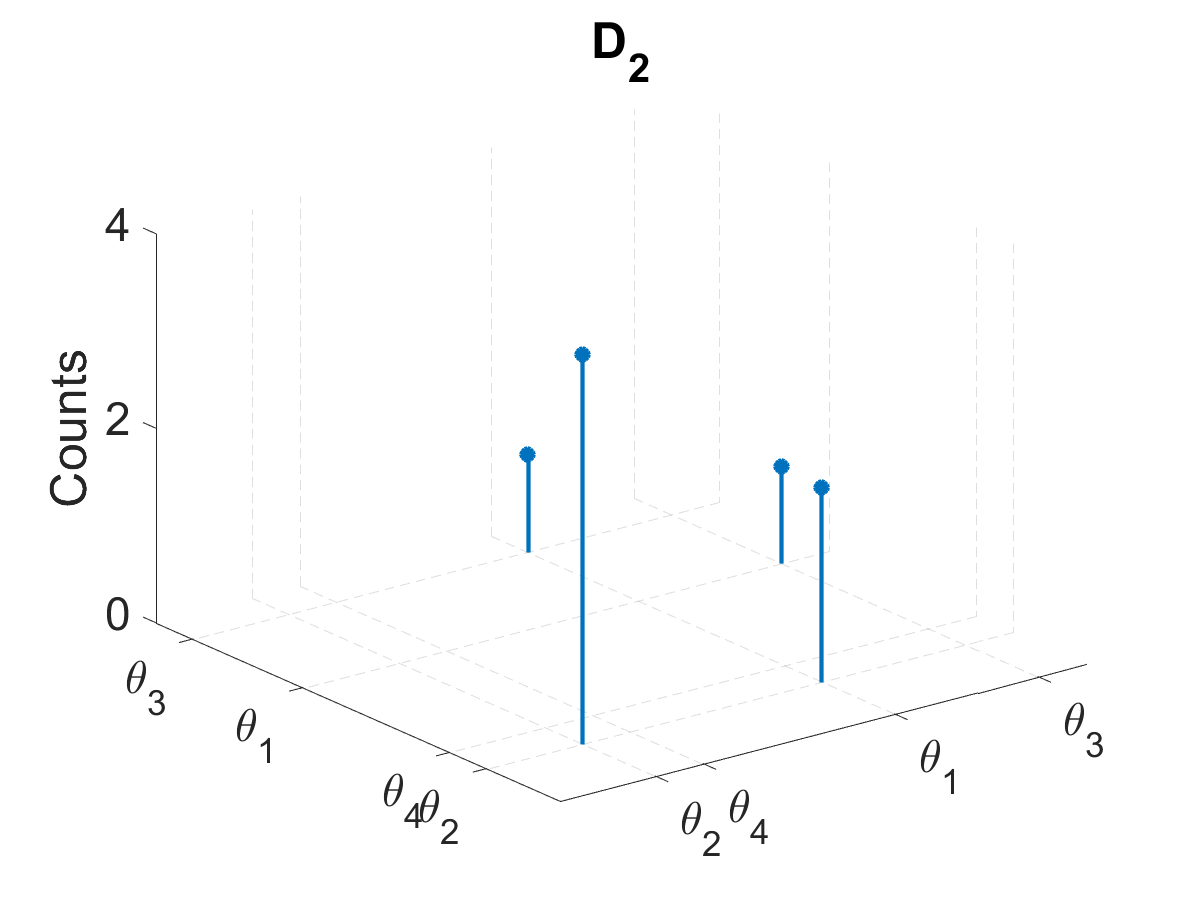}\end{tabular} &
\begin{tabular}[c]{c}
\resizebox{0.3\textwidth}{!}{\input{./figures/tikz1b}}\\
\resizebox{0.3\textwidth}{!}{\input{./figures/tikz1}}\end{tabular}
 &\begin{tabular}[c]{c}
~~\\\resizebox{0.3\textwidth}{!}{\input{./figures/tikz2}}\end{tabular}\\
(a) & ~~(b) & ~~(c)
\end{tabular}
\end{center}
\caption{An example of (a) the restriction on $[0,1]^2$ of the two atomic measures $D_1$ and $D_2$, (b) the corresponding
multiview directed multigraphs (top: view 1; bottom: view 2) and (c) corresponding undirected graph.}%
\label{fig:multiview}%
\end{figure}

\begin{figure}[ptb]
\begin{center}
\subfigure[$W_k \times W_k$]{\includegraphics[width=.25\textwidth]{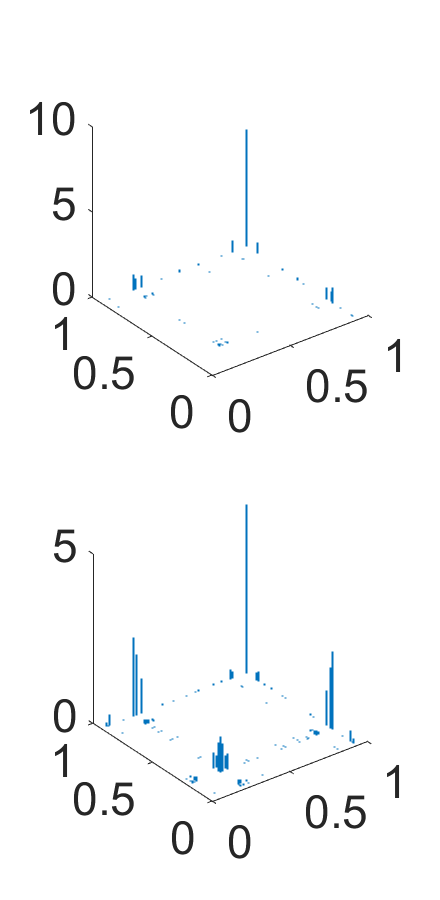}}
\subfigure[Integer point processes $D_k$]{\includegraphics[width=.25\textwidth]{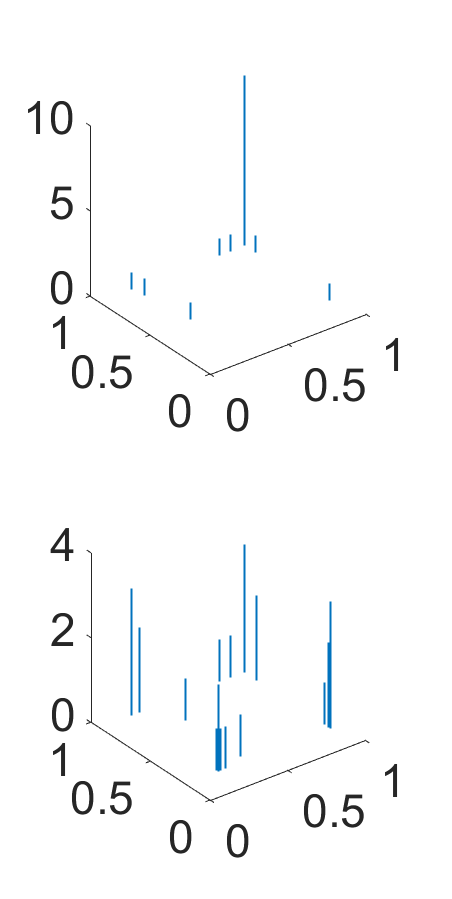}}
\subfigure[Point process $Z$]{\includegraphics[width=.3\textwidth]{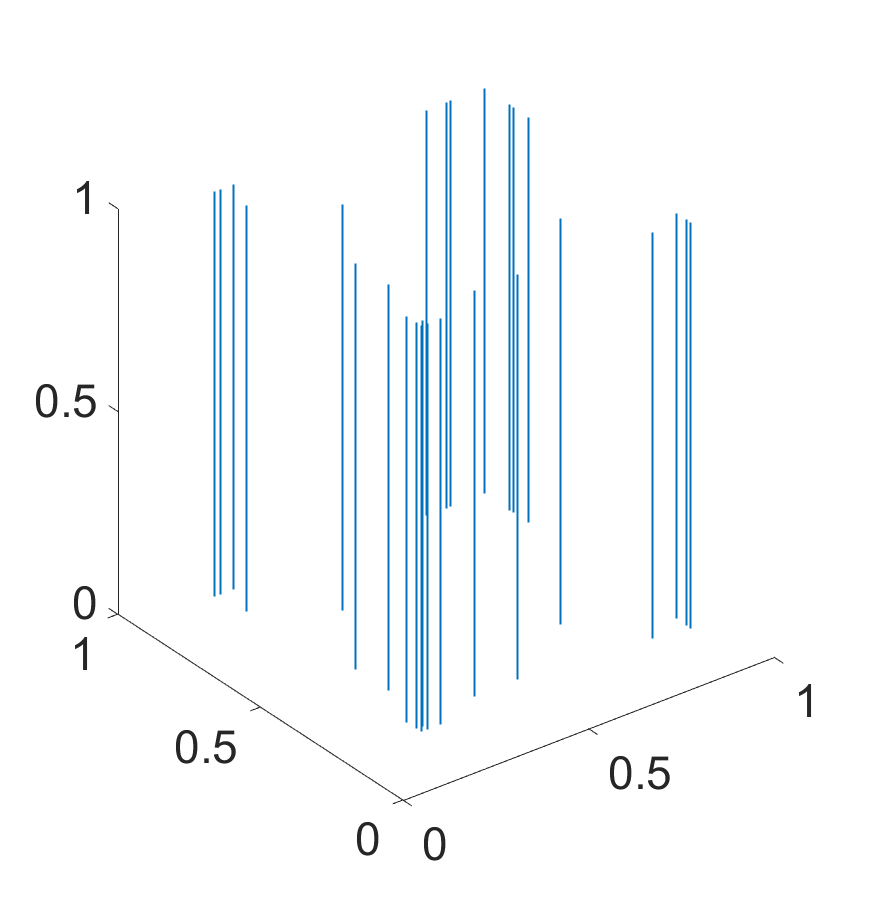}}
\end{center}
\vspace{-0.15in}
\caption{An example, for $p=2$, of (a) the product measures $W_k\times
W_k$, (b) a draw of the directed multigraph measures $D_k\mid W_k\sim {\Poi}(W_k\times W_k)$ {and}
(c) corresponding undirected measure $Z=\sum_{i=1}^{\infty}\sum_{j=1}^{\infty
}\min(1,\sum_{k=1}^p n_{ijk}+n_{jik})\delta_{(\theta_{i},\theta_{j})}$.} % and (d) the sociability measure $W$ used above.}
\label{fig:undirectedHierarchy}%
\end{figure}

Mimicking the hierarchical construction of~\cite{Caron2014}, we introduce integer-valued random measures $D_k$ on $\mathbb R^2_+$, $k=1,\ldots,p$,
\begin{equation}
D_k    =\sum_{i=1}^{\infty}\sum_{j=1}^{\infty}n_{ijk}\delta_{(\theta_{i}%
,\theta_{j})}
\end{equation}
where the $n_{ijk}$ are natural integers. The vector of random measures $(D_1,\ldots,D_p)$ can be interpreted as representing a multiview (a.k.a. multiplex or multi-relational) directed multigraph~\citep{Verbrugge1979,Salter-Townshend2013}, where $n_{ijk}$ represents the number of interactions from node $i$ to node $j$ in the view $k$; see Figure~\ref{fig:multiview} for an illustration. Conditionally on the vector of CRMs, the measures $D_k$ are independently drawn from a Poisson process\footnote{Note that we consider a generalized definition of a Poisson process, where the mean measure is allowed to have atoms; see e.g.~\citet[Section~2.4]{Daley2008}.} with mean measure $W_{k}\times W_k$
\begin{align}
D_k|(W_{1},\ldots,W_{p})  &  \sim\Poi\left(  W_{k}\times
W_{k}\right)
\end{align}
that is, the $n_{ijk}$ are independently Poisson distributed with rate $w_{ik} w_{jk}$.

Finally, the point process $Z$ representing the graph (Eq. \eqref{eq:pointprocessZ}) is deterministically obtained from $(D_1,\ldots,D_p)$ by setting $z_{ij}=1$ if there is at least one directed connection between $i$ and $j$ in any view, and 0 otherwise, therefore $z_{ij}=\min(1,\sum_{k=1}^p n_{ijk}+n_{jik})$. To sum up, the graph model is described as follows:
\begin{align}
\begin{aligned} \begin{array}{ll} W_k =\sum_{i=1}^{\infty}w_{ik}\delta_{\theta_{i}} & (W_1,\ldots,W_p) \sim \mbox{CRM}(\rho,\lambda)\\ D_k =\sum_{i=1}^{\infty}\sum_{j=1}^{\infty}n_{ijk}\delta_{(\theta_{i},\theta_{j})} & D_k\mid W_k\sim\text{Poisson}\left(W_k \times W_k\right) \\ Z =\sum_{i=1}^{\infty}\sum_{j=1}^{\infty}\min(1,\sum_{k=1}^p n_{ijk}+n_{jik})\delta_{(\theta _{i},\theta_{j})}. & \end{array} \end{aligned} \label{eq:Zhierarchy}%
\end{align}
The model construction is illustrated in Figure~\ref{fig:undirectedHierarchy}. Integrating out the measures $D_k$, $k=1,\ldots,p$, the construction can be expressed as, for $i\leq j$
\begin{align}
z_{ij}|(w_{\ell 1},\ldots,w_{\ell p})_{\ell=1,2,\ldots}  &  \sim\left\{
\begin{array}
[c]{ll}%
\text{Ber}(1-\exp(-2\sum_{k=1}^{p}w_{ik}w_{jk})) & i\neq j\\
\text{Ber}(1-\exp(-\sum_{k=1}^{p}w_{ik}^{2})) & i=j
\end{array}
\right.\label{eq:zij}
\end{align}
and $z_{ji}=z_{ij}$; see Figure~\ref{fig:pointprocess}.

\paragraph{Graph Restrictions.} Except in trivial cases, we have $W_k(\mathbb R_+)=\infty$ a.s. and therefore $Z(\mathbb R^2_+)=\infty$ a.s., so the number of points over the plane is infinite a.s. For $\alpha>0$, we consider restrictions of the measures $W_k$, $k=1,\ldots,p$, to the interval $[0,\alpha]$ and of the measures $D_k$ and $Z$ to the box $[0,\alpha]^2$, and write respectively $W_{k\alpha}$, $D_{k\alpha}$ and $Z_\alpha$ these restrictions. Note that condition \eqref{eq:conditionLevydef} ensures that $W_{k\alpha}([0,\alpha])<\infty$ a.s. hence $D_{k\alpha}([0,\alpha]^2)<\infty$ and $Z_{\alpha}([0,\alpha]^2)<\infty$ a.s. As a consequence, for a given $\alpha>0$, the model yields a finite number of edges a.s., even though there may be an infinite number of points $(w_i,\theta_i)\in\mathbb R_+\times[0,\alpha]$; see Section~\ref{sec:properties}.

\begin{remark}
The model defined above can also be used for random multigraphs, where $n_{ij}=\sum_{k=1}^p n_{ijk}$ is the number of directed interactions between $i$ and $j$. Then we have $$n_{ij}|(w_{\ell 1},\ldots,w_{\ell p})_{\ell=1,2,\ldots}\sim \Poi \left (\sum_{k=1}^p w_{ik}w_{jk} \right )$$
which is a Poisson non-negative factorization~\citep{Lee1999,Cemgil2009,Psorakis2011,Ball2011,Gopalan2015}.
\end{remark}

\begin{remark}
The model defined by Eq.~\eqref{eq:zij} allows to model networks which exhibit assortativity~\citep{Newman2003a}, meaning that two nodes with similar characteristics (here similar set of weights) are more likely to connect than nodes with dissimilar characteristics.
 The link function can be generalized to \citep[see e.g.][]{Zhou2015}
 $$z_{ij} \sim \Ber\left(1-\exp\left(-\sum_{k=1}^p\sum_{\ell=1}^p \eta_{k\ell}w_{ik}w_{j\ell}\right )\right )$$
where $\eta_{k\ell}\geq 0$, in order to be able to capture both assortative and dissortative mixing in the network. In particular, setting larger values off-diagonal than on the diagonal of the matrix $(\eta_{k\ell})_{1\leq k,\ell\leq p}$ allows to capture dissortative mixing. The properties and algorithms for simulation and posterior inference can trivially be extended to this more general case. In order to keep the notations as simple as possible, we focus here on the simpler link function \eqref{eq:zij}.
\end{remark}

\subsection{Particular model based on compound CRMs}
\label{subsec:ccrm}

\begin{figure}[ptb]
\begin{center}
\includegraphics[width=10cm]{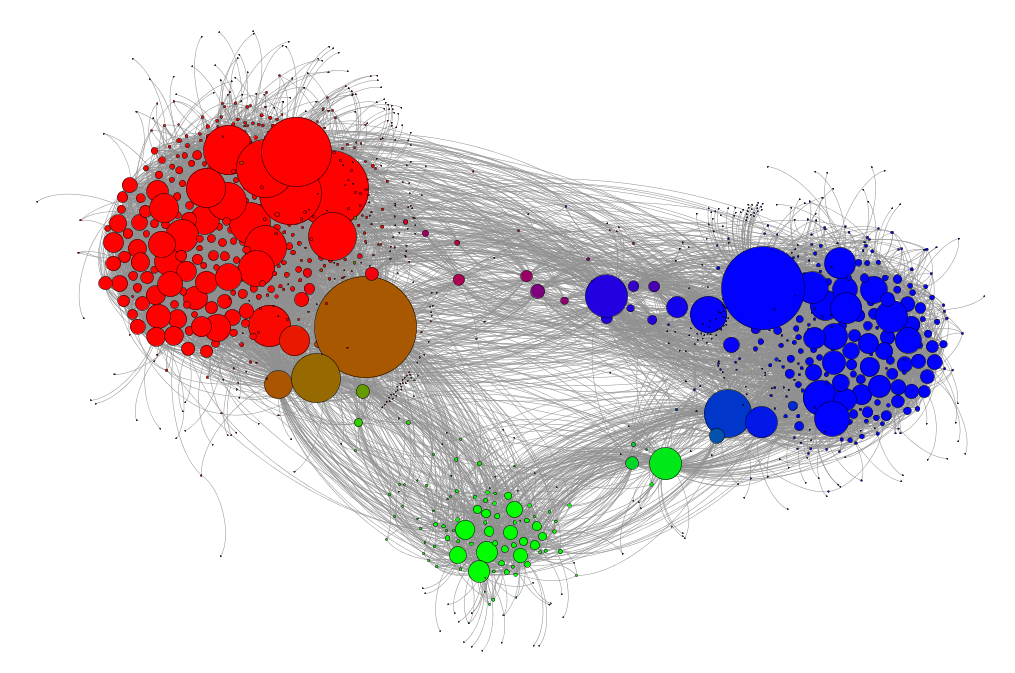}
\end{center}
\caption{Graph sampled from the model with three latent communities, identified by colors red, green, blue. For each
node, the intensity of each color is proportional to the value of the associated weight in that community. Pure red/green/blue color indicates the node is only strongly affiliated to a single community. A mixture of those colors indicates balanced affiliations to different communities. Graph generated with the software Gephi~\citep{Bastian2009}.}%
\label{fig:samplegraph3}%
\end{figure}

The key component in our statistical network model is the multivariate L\'evy measure $\rho$ in \eqref{eq:vCRM}. Various approaches have been developed for constructing multivariate L\'evy measures~\citep{Tankov2003,Cont2003,Kallsen2006,Barndorff-Nielsen2001,Skorohod1991}, or more specifically vectors of completely random measures~\citep{Epifani2010,Leisen2011,Leisen2013,Griffin2013,Lijoi2014}.  We will in this paper consider the following particular form:
\begin{equation}
\rho(dw_{1},\ldots,dw_{p})=e^{-\sum_{k=1}^{p}\gamma_{k}w_{k}}\int_{0}^{\infty
}w_{0}^{-p}F\left(  \frac{dw_{1}}{w_{0}},\ldots,\frac{dw_{p}}{w_{0}}\right)
\rho_{0}(dw_{0})\label{eq:rho}
\end{equation}
where $F(d\beta_{1},\ldots d\beta_{p})$ is some \textit{score} probability  distribution on
$\mathbb{R}_{+}^{d}$, with moment generating function $M(t_1,\ldots,t_p)$, $\rho_0$ is
a \textit{base} L\'{e}vy measure on $\mathbb{R}_{+}$ and $\gamma_{k}\geq0$ are \textit{exponentially tilting parameters} for
$k=1,\ldots,p$. The model defined by \eqref{eq:nu} and \eqref{eq:rho} is a special case of the compound completely random measure (CCRM) model proposed by \cite{Griffin2014}. It admits the following hierarchical construction, which makes interpretability, characterization of the conditionals and analysis of this class of models particularly easy. Let
\begin{equation}
W_{0}=\sum_{i=1}^{\infty}w_{i0}\delta_{\theta_{i}}\sim\text{CRM}%
(\widetilde{\rho}_{0},\lambda) \label{eq:W0}%
\end{equation}
where $\widetilde{\rho}_{0}$ is a measure on $\mathbb R_+$ defined by $\widetilde{\rho}_{0}(dw_0)=M(-w_{0}\gamma_{1},\ldots,-w_{0}\gamma_{p})\rho_{0}(dw_0)$, and for
$k=1,\ldots,p$ and  $i=1,2,\ldots$
\[
w_{ik}=\beta_{ik}w_{i0}%
\]
where the scores $\beta_{ik}$ have the following joint distribution%
\begin{equation}
(\beta_{i1},\ldots,\beta_{ip})|w_{i0}\overset{\text{ind}}{\sim}H(\cdot|w_{i0})\label{eq:condbeta}
\end{equation}
with $H$ is an exponentially tilted version of $F$:
\begin{equation}
H(d\beta_{1},\ldots,d\beta_{p}|w_{0})=\frac{e^{-w_{0}\sum_{k=1}^{p}\gamma
_{k}\beta_{k}}F\left(  d\beta_{1},\ldots,d\beta_{p}\right)  }{\int%
_{\mathbb{R}_{+}^{p}}e^{-w_{0}\sum_{k=1}^{p}\gamma_{k}\widetilde{\beta}_{k}%
}F\left(  d\widetilde{\beta}_{1},\ldots,d\widetilde{\beta}_{p}\right)}.
\end{equation}
Additionally, the set of points $(w_{i0},\beta_{i1},\ldots,\beta_{ip})_{i=1,2,\ldots}$ is a
Poisson point process with mean measure%
\begin{equation}
e^{-w_{0}\sum_{k=1}^{p}\gamma_{k}\beta_{k}}F(d\beta_{1},\ldots,d\beta_{p}%
)\rho_{0}(dw_{0}).\label{eq:Levyw0beta}
\end{equation}

Dependence between the different CRMs is both tuned by the shared scaling parameter $w_{i0}$ and potential dependency between the scores $(\beta_{i1},\ldots,\beta_{ip})$. The hierarchical construction has the following interpretation:
\begin{itemize}
\item The weight $w_{i0}$ is an individual scaling parameter for node $i$ whose distribution is tuned by the base L\'evy measure $\rho_0$. It can be considered as a degree correction, as often used in network models~\citep{Karrer2011,Zhao2012,Herlau2015}. As shown in Section~\ref{sec:properties}, $\rho_0$ tunes the overall sparsity properties of the network.
\item The community-related scores $\beta_{ik}$ tune the level of affiliation of node $i$ to community $k$; this is controlled by both the score distribution $F$ and the tilting coefficients $\gamma_k$. These parameters tune the overlapping block-structure of the network.
\end{itemize}
An example of such a graph with three communities is displayed in Figure~\ref{fig:samplegraph3}.

\paragraph{Specific choices for $F$ and $\rho_0$.}
We now give here specific choices of score distribution $F$ and base L\'evy measure $\rho_0$, which lead to scalable inference algorithms. As in \cite{Griffin2014}, we consider that $F$ is a product of independent gamma distributions%
\begin{equation}
F(d\beta_{1},\ldots,d\beta_{p})=\prod_{k=1}^{p}\beta_{k}^{a_{k}-1}%
e^{-b_{k}\beta_{k}}\frac{b_{k}^{a_{k}}}{\Gamma(a_{k})}d\beta_{k}\label{eq:Fproductgamma}
\end{equation}
where $a_{k}>0$,$b_{k}>0$, $k=1,\ldots,p$, which leads to
\[
H(dw_{1},\ldots,dw_{p}|w_{0})\propto\prod_{k=1}^{p}w_{k}^{a_{k}-1}%
e^{-\frac{b_{k}w_{k}}{w_{0}}-\gamma_{k}w_{k}}dw_{k}%
\]
which is also a product of gamma distributions.

$\rho_0$ is set to be the mean measure of the jump part of a generalized gamma process \citep{Hougaard1986,Brix1999}, which
has been extensively used in BNP models due to its generality, the
interpretability of its parameters and its attractive conjugacy
properties~\citep{James2002,Lijoi2007,Saeedi2011,Caron2012,Caron2014a}. The L\'{e}vy measure in this case is%
\begin{equation}
\rho_0(dw_0)=\frac{1}{\Gamma(1-\sigma)}w_0^{-1-\sigma}\exp(-w_0\tau)dw_0
\label{eq:levyGGP}%
\end{equation}
where the parameters $(\sigma,\tau)$ verify

\begin{align}
\label{eq:conditionsGGP}%
\sigma \in(0,1),\tau\geq0~~\text{ or }~~ \sigma   \in(-\infty,0],\tau>0.
\end{align}

 The gamma
process ($\sigma=0$), the inverse Gaussian process ($\sigma=\frac{1}{2}$) and
the stable process ($\sigma\in(0,1)$, $\tau=0$) are special cases. Using \eqref{eq:Fproductgamma} and \eqref{eq:levyGGP}, the multivariate L\'evy measure has the following analytic form
\begin{align*}
\rho(dw_{1},\ldots,dw_{p})
&  =\frac{2e^{-\sum_{k=1}^{p}\gamma_{k}w_{k}}  }{\Gamma(1-\sigma)}\left[  \prod_{k=1}^{p}\frac
{w_{k}^{a_{k}-1}b_{k}^{a_{k}}}{\Gamma(a_{k})}\right]\left(  \frac{\tau}{\sum
_{k=1}^{p}b_{k}w_{k}}\right)  ^{{-}\frac{\kappa}{2}}%
K_{\kappa}\left(  2\sqrt{\tau\sum_{k}b_{k}w_{k}}\right)dw_1\ldots dw_p
\end{align*}
where $\kappa=\sigma+\sum_{k=1}^{p}a_{k}$ and $K$ is the modified Bessel function of the second kind.

\section{Properties and Simulation}
\label{sec:properties}

The first theorem provides expressions for the expected number of edges in the multigraph and simple graph, and for the expected number of nodes. The proof is given in Appendix \ref{sec:proofs2}.
\begin{theorem}
\label{sec::expected_num}
The expected number of edges in the multigraph $D^*_\alpha$, edges in the undirected graph $N^{(e)}_\alpha$ and observed nodes $N_\alpha$ are given as follows:
\begin{align*}
\mathbb{E}[D^*_\alpha]&= \alpha^2 \mu^T\mu + \alpha \text{tr}(\Sigma)\\
\mathbb{E}[N^{(e)}_\alpha]&=
\alpha \int_{\mathbb{R}^p_+}   \left (1-e^{-w^Tw}\right )  \rho(dw_{1},\ldots,dw_p) +  \frac{\alpha^2}{2} \int_{\mathbb{R}^p_+} \psi(2w_1,\ldots,2w_p)  \rho(dw_{1},\ldots,dw_p)\\
\mathbb{E}[N_\alpha]&=
\alpha  \int_{\mathbb{R}^p_+} \left (1-e^{-w^T w - \alpha \psi(2w)}\right )\rho(dw_1,\ldots,dw_p)
\end{align*}
where $\mu=\int_{\mathbb{R}^p_+} w\rho(dw_1,\ldots,dw_p)$, $\Sigma = \int_{\mathbb{R}^p_+} w w^T \rho(dw_1,\ldots,dw_p)$ and $\psi(t_1,\ldots,t_p) = \int_{\mathbb{R}^p_+}(1-e^{-\sum_{k=1}^p t_iw_i})\rho(dw_{1},\ldots,dw_p)$ is the multivariable Laplace exponent.
\end{theorem}

\subsection{Exchangeability}

The point process $Z$ defined by \eqref{eq:Zhierarchy} is jointly exchangeable in the sense of \cite{Kallenberg1990,Kallenberg2005}. For any $h>0$ and any permutation $\pi$ of $\mathbb N$
\begin{equation}
(Z(A_i\times A_j))\overset{d}{=}(Z(A_{\pi(i)}\times A_{\pi(j)}))\text{ for }(i,j)\in \mathbb N^2
\label{eq:kallenbergexchangeability}
\end{equation}
where $A_i=[h(i-1),hi]$. This follows directly from the fact that the vector of CRMs $(W_1,\ldots,W_p)$ has independent and identically distributed increments, hence
\begin{equation}
(W_1(A_i),\ldots,W_p(A_i))\overset{d}{=}(W_1(A_{\pi(i)}),\ldots,W_p(A_{\pi(i)})).
\end{equation}
  The model thus falls into the general representation theorem for exchangeable point processes~\citep{Kallenberg1990}.

\subsection{Sparsity}

In this section, following the asymptotic notations of \cite{Janson2011}, we derive the sparsity properties of our graph model, first for the general construction of Section~\ref{subsec:general}, then for the specific construction on compound CRMs of Section~\ref{subsec:ccrm}. Similarly to the notations in \cite{Caron2014}, let $Z_\alpha$ be the restriction of $Z$ to the box $[0,\alpha]^2$. Let  $(N_\alpha)_{\alpha\geq 0}$ and $(N_\alpha^{(e)})_{\alpha\geq 0}$ be counting processes respectively corresponding to the number of nodes and edges in $Z_\alpha$:
\begin{align*}
N_\alpha &=\sum_i \1{\theta_i\leq \alpha}\1{(\sum_j z_{ij}\1{\theta_j\leq \alpha})\geq 1}\\% \text{card}(\{ \theta_i\in[0,\alpha]|Z(\{\theta_i\}\times[0,\alpha])>0\})\\
N_\alpha^{(e)} &=\sum_{i\leq j} z_{ij}\1{\theta_i\leq \alpha}\1{\theta_j\leq \alpha}.
\end{align*}

Note that in the propositions below, we discard the trivial case $\int_{\mathbb R^p_+} \rho(dw_1,\ldots,dw_p)=0$ which implies $N_\alpha^{(e)}=N_\alpha=0$ a.s.

\paragraph{General construction.} The next proposition characterizes the sparsity properties of the random graph depending on the properties of the L\'evy measure $\rho$. In particular, if
\begin{equation}
\int_{\mathbb R^p_+} \rho(dw_1,\ldots,dw_p)=\infty
\end{equation}
then, for any $\alpha>0$, there is a.s. an infinite number of $\theta_i\in[0,\alpha]$ for which $\sum_k w_{ik}>0$ and the vector of CRMs is called infinite-activity. Otherwise, it is finite-activity.  %All unspecified limits are as $\alpha\rightarrow\infty$.

\begin{proposition}
Assume that, for any $k=1,\ldots,p$,
\begin{equation}
\int_{{\mathbb R^p_+}} w_k\rho(dw_1,\ldots,dw_p)<\infty\label{eq:momentcondition}
\end{equation}
Then

\[
N_\alpha^{(e)} =\left\{\begin{array}{ll}
                   \Theta(N_\alpha^2) & \text{if }(W_1,\ldots,W_p)\text{ is finite-activity} \\
                    o(N_\alpha^2)& \text{otherwise}
                 \end{array}\right .
\]
a.s. as $\alpha$ tends to $\infty$. \label{th:sparsitygeneral}
\end{proposition}

The proof is given in Appendix~\ref{sec:proofs}.

\paragraph{Construction based on CCRMs.} For the CCRM L\'evy measure \eqref{eq:rho}, the sparsity properties are solely tuned by the base L\'evy measure $\rho_0$. Ignoring trivial degenerate cases for the score distribution $F$, it is easily shown that the CCRM model defined by \eqref{eq:nu} and \eqref{eq:rho} is infinite-activity iff the L\'evy measure $\rho_0$ verifies
\begin{equation}
\int_{0}^{\infty}\rho_0(dw)=\infty.\label{eq:condinfiniteCRM}
\end{equation}
In this case all CRMs $W_0,W_1,\ldots,W_p$ are infinite-activity. Otherwise they are all finite-activity and the vector of CRMs is finite-activity. In the particular case of a CCRM with independent gamma distributed scores \eqref{eq:Fproductgamma} and generalized gamma process base measure \eqref{eq:levyGGP}, the condition \eqref{eq:condinfiniteCRM} is satisfied whenever $\sigma\geq 0$. The next proposition characterizes the sparsity of the network depending on the properties of the base L\'evy measure $\rho_0$.

\begin{proposition}
\label{th:sparsityCCRM}
Assume that
\begin{equation}
\int_0^\infty w_0\rho_0(dw_0)<\infty\label{eq:momentcondition2}
\end{equation}
and $F$ is not degenerated at 0. Then

\[
N_\alpha^{(e)} =\left\{\begin{array}{ll}
                   \Theta(N_\alpha^2) & \text{if }\int_0^\infty \rho_0(dw)<\infty \\
                    o(N_\alpha^2)& \text{otherwise}
                 \end{array}\right .
\]
a.s. as $\alpha$ tends to $\infty$. Furthermore, if the tail L\'evy intensity $\overline{\rho}_0$ defined by
\begin{equation}
\overline{\rho}_0(x)=\int_x^\infty \rho_0(dw),\label{eq:taillevy}
\end{equation}
is a regularly varying function, i.e.
%\[
%\overline{\rho}_0(x)\overset{x\downarrow0}{\sim} x^{-\sigma}\ell(1/x)
%\]
\[
\frac{\overline{\rho}_0(x)}{x^{-\sigma}\ell(1/x)}\longrightarrow 1\text{ as }x\rightarrow 0
\]
for some $\sigma\in(0,1)$ where $\ell$ is a {slowly varying} function verifying $\lim_{t\rightarrow\infty}\ell(at)/\ell(t)=1$ for any $a>0$ and $\lim_{t\rightarrow\infty}\ell(t)>0$,
then
\[
N_\alpha^{(e)} =  O(N_\alpha^{2/(1+\sigma)})
\]
a.s. as $\alpha$ tends to $\infty$. In the particular case of a CCRM with independent gamma distributed scores \eqref{eq:Fproductgamma} and generalized gamma process base measure \eqref{eq:levyGGP}, condition \eqref{eq:momentcondition2} is equivalent to having $\tau>0$. In this case, we therefore have
\[
N_\alpha^{(e)} =\left\{\begin{array}{ll}
                   \Theta(N_\alpha^2) & \text{if }\sigma<0 \\
                    o(N_\alpha^2)& \text{if }\sigma\geq 0\\
                    O(N_\alpha^{2/(1+\sigma)})& \text{if }\sigma\in (0,1).
                 \end{array}\right .
\]

\end{proposition}

\begin{figure}[ptb]
\begin{center}
\subfigure[]{\includegraphics[width=0.33\columnwidth]{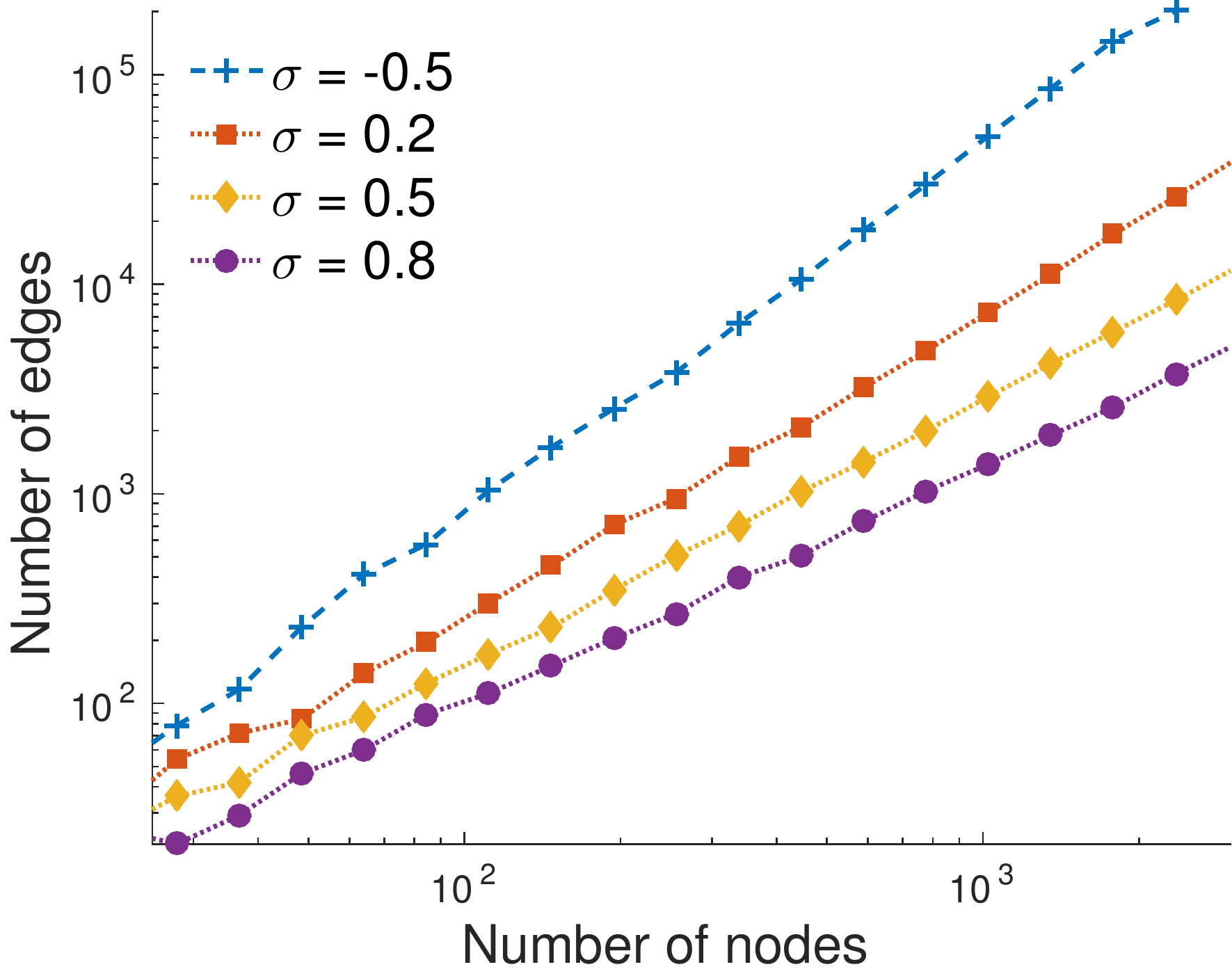}}
\hspace*{.1\textwidth}
\subfigure[]{\includegraphics[width=0.33\columnwidth]{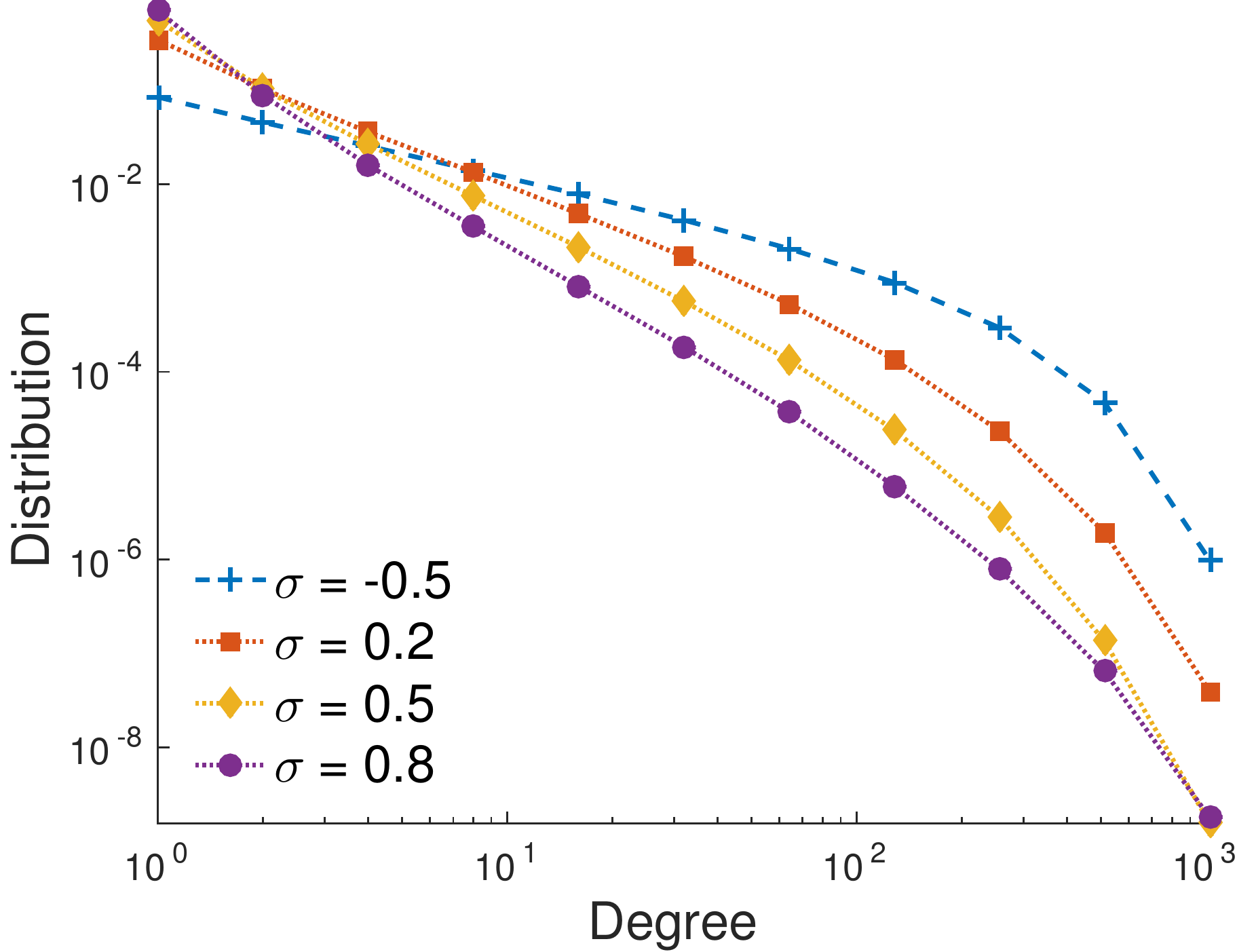}}
\end{center}

\caption{ Empirical analysis of the properties of CCRM based graphs generated with parameters $p=2$, $\tau=1$, $a_k=0.2$, $b_k=\frac{1}{p}$ and averaging over various $\alpha$. (a) Number of edges versus the number of nodes and (b) degree distributions on a log-log scale for various $\sigma$: one finite-activity CCRM ($\sigma=-0.5$) and three infinite-activity CCRMs ($\sigma=0.2$, $\sigma=0.5$ and $\sigma=0.8$). In (a) we note growth at a rate $\Theta(N_\alpha ^2)$ for $\sigma=-0.5$ and $O(N_\alpha^{2/(1+\sigma)})$ for $\sigma\in(0,1)$. }
\label{fig:degree_edgesvsnodes}%
\end{figure}

The proof is given in Appendix~\ref{sec:proofs}. Figure~\ref{fig:degree_edgesvsnodes}(a) provides an empirical illustration of Proposition~\ref{th:sparsityCCRM} for a CCRM with independent gamma scores and generalized gamma based L\'evy measure. Figure~\ref{fig:degree_edgesvsnodes}(b) shows empirically that the degree distribution also exhibits a power-law behaviour when $\sigma\in(0,1)$.

\subsection{Simulation}
\label{sec:simulation}

The point process $Z$ is defined on the plane. We describe in this section how to sample realizations of restrictions $Z_\alpha$ of $Z$ to the box $[0,\alpha]^2$.

\paragraph{General construction.} The hierarchical construction given by {Eq.~\eqref{eq:Zhierarchy}} suggests a direct way to sample from the model:
 \begin{enumerate}
 \item Sample $(w_{i1},\ldots,w_{ip},\theta_i)_{i=1,2,\ldots}$ from a Poisson process with mean measure $\nu(dw_1,\ldots,dw_p, d\theta){\mathds 1}_{\theta\in[0,\alpha]}$.
 \item For each pair of points, sample $z_{ij}$ from \eqref{eq:zij}.
  \end{enumerate}
  There are two caveats to this strategy. First, for infinite-activity CRMs, the number of points in $\mathbb R_+^p\times [0,\alpha]$ is almost surely infinite; even for finite-activity CRMs, it may be so large that it is not practically feasible. We need therefore to resort to an approximation, by sampling from a Poisson process with an approximate mean measure $\nu^{\varepsilon}(dw_1,\ldots, dw_p, d\theta){\mathds 1}_{\theta\in[0,\alpha]}=\rho^\varepsilon(dw_1,\ldots, dw_p)\lambda(d\theta){\mathds 1}_{\theta\in[0,\alpha]}$ where
\begin{equation*}
\int_{\mathbb R_+^p}\rho^{\varepsilon}(dw_1,\ldots, dw_p )<\infty
\end{equation*}
with  $\varepsilon>0$ controlling the level of approximation. The approximation is specific to the choice of the mean measure, and we describe such an approximation for CCRMs below.

The second caveat is that, for applying Eq.~\eqref{eq:zij}, we need to consider all pairs $i\leq j$, which can be computationally problematic. We can instead, similarly to~\cite{Caron2014}, use the hierarchical Poisson construction as follows:
\begin{enumerate}
\item Sample $(w_{i1},\ldots,w_{ip},\theta_i)_{i=1,2,\ldots,K}$ from a Poisson process with mean measure $\nu^{\varepsilon}(dw_1,\ldots,dw_p,d\theta){\mathds 1}_{\theta\in[0,\alpha]}$. Let $W_{k,\alpha}^\varepsilon=\sum_{i=1}^K w_{ik}\delta_{\theta_i}$ be the associated truncated CRMs and $W_{k,\alpha}^{\varepsilon \ast  }=\sum_{i=1}^K w_{ik}$ their total masses.
\item For $k=1,\ldots,p$, sample $D^\ast _{k,\alpha}|W_{k,\alpha}^{\varepsilon \ast  }\sim\text{Poisson}((W_{k,\alpha}^{\varepsilon \ast  })^2)$.
\item For $k=1,\ldots,p$, $\ell=1,\ldots,D^\ast_{k,\alpha}$, $j=1,2$, sample $U_{k\ell j}|W_{k,\alpha}^\varepsilon\overset{\text{ind}}{\sim}\frac{W_{k,\alpha}^\varepsilon}{W_{k,\alpha}^{\varepsilon \ast}}$.
\item Set $D_{k,\alpha}^\varepsilon=\sum_{\ell = 1}^{D^\ast_{k,\alpha}} \delta_{U_{k\ell 1,k\ell 2}}$.
\item Obtain $Z$ from $(D_1,\ldots,D_p)$ as in \eqref{eq:Zhierarchy}.
\end{enumerate}

\paragraph{Construction based on CCRMs.} The hierarchical construction of compound CRMs suggests an algorithm to simulate a vector of CRMS. We consider the following (truncated) mean measure
\begin{equation}
\rho^{\varepsilon}(dw_{1},\ldots,dw_{p})=e^{-\sum_{k=1}^{p}\gamma_{k}w_{k}}%
\int_{\varepsilon}^{\infty}w_{0}^{-p}F\left(  \frac{dw_{1}}{w_{0}}%
,\ldots,\frac{dw_{p}}{w_{0}}\right)  \rho_{0}(dw_{0})
\end{equation}
with $\varepsilon\geq 0$. We can sample from the (truncated) CCRM as follows
\begin{enumerate}
\item[1.] (a) Sample $(w_{i0},\theta_i)_{i=1,\ldots,K}$ from a Poisson point process with mean measure $\widetilde{\rho}_{0}(dw_0)\lambda(d\theta){\mathds 1}_{\{w_0>\varepsilon,\theta\in[0,\alpha]\}}$.
\item[~~] (b) For $i=1,\ldots,K$ and $k=1,\ldots,p$, set $w_{ik}=\beta_{ik}w_{i0}$ where $(\beta_{i1},\ldots,\beta_{ip})|w_{i0}$ is drawn from \eqref{eq:condbeta}.
\end{enumerate}
The truncation level $\varepsilon$ is set to 0 for finite-activity CCRMs, and $\varepsilon>0$ otherwise. We explain in Appendix~\ref{sec:app:simujumps} how to perform step 1.(a) in the case of a tilted generalized gamma process.

\section{Posterior inference}
\label{sec:inference}

In this section, we describe a MCMC algorithm for posterior inference of the model parameters and hyperparameters in the statistical network model defined in Section~\ref{sec:model}. We first describe the data augmentation scheme and characterization of conditionals. We then describe the sampler for a general L\'evy measure $\rho$, and finally derive the sampler for compound CRMs.

\subsection{Characterization of conditionals and data augmentation}

Assume that we have observed a set of connections $(z_{ij})_{1\leq i,j\leq N_\alpha}$, where $N_\alpha$ is the number of nodes with at least one connection. We aim at inferring the positive parameters $(w_{i1},\ldots,w_{ip})_{i=1,\ldots,N_\alpha}$ associated to the nodes with at least one connection. We also want to estimate the positive parameters associated to the other nodes with no connection. The number of such nodes may be large, and even infinite for infinite-activity CRMs; but under our model, these parameters are only identifiable through their sum, denoted $(w_{\ast 1},\ldots,w_{\ast p})$. Note that the node locations $\theta_i$ are not likelihood identifiable, and we will not try to infer them. We assume that there is a set of unknown hyperparameters $\phi$ of the mean intensity $\rho$, with prior $p(\phi)$. We assume that the L\'evy measure $\rho$ is absolutely continuous with respect to the Lebesgue measure on $\mathbb R^d$, and write simply $\rho(dw_1,\ldots,dw_p;\phi)=\rho(w_1,\ldots,w_p;\phi)dw_1\ldots dw_p$.  The parameter $\alpha$ is also assumed to be unknown, with some prior $\alpha\sim \Gam(a_\alpha, b_\alpha)$ with $a_\alpha>0,b_\alpha>0$.  We therefore aim at approximating $p((w_{1k},\ldots,w_{N_\alpha k},w_{\ast k})_{k=1,\ldots,p},\phi,\alpha|(z_{ij})_{1\leq i,j\leq N_\alpha})$.\medskip

As a first step, we characterize the conditional distribution of the restricted vector of CRMs $(W_{1\alpha},\ldots,W_{p\alpha})$ given the restricted measures $(D_{1\alpha},\ldots,D_{p\alpha})$. Proposition~\ref{th:characterization} below extends Theorem 12 in \cite{Caron2014} to the multivariate setting.

\begin{proposition}\label{th:characterization}
Let $(\theta_1,\ldots,\theta_{N_\alpha})$, $N_\alpha\geq 0$ be the support points of $(D_{1\alpha},\ldots,D_{p\alpha})$, with $$D_{k\alpha}=\sum_{1\leq i,j\leq N_\alpha}n_{ijk} \delta_{(\theta_i,\theta_j)}.$$ The conditional distribution of $(W_{1\alpha},\ldots,W_{p\alpha})$ given $(D_{1\alpha},\ldots,D_{p\alpha})$ is equivalent to the distribution of
\begin{equation}
\left (\widetilde{W}_{1} +\sum_{i=1}^{N_\alpha} w_{i1}\delta_{\theta_i} ,\ldots,\widetilde{W}_p +\sum_{i=1}^{N_\alpha} w_{ip}\delta_{\theta_i} \right )
\end{equation}
where $(\widetilde{W}_{1},\ldots,\widetilde{W}_{p})$ is a vector of discrete random measures, which depends on $(D_{1\alpha},\ldots,D_{p\alpha})$ only through the total masses $w_{\ast k}=\widetilde{W}_{k}([0,\alpha])$.

The set of weights $(w_{ik})_{i=1,\ldots,N_\alpha;k=1,\ldots,p}$ and $(w_{\ast k})_{k=1,\ldots,p}$ are dependent, with joint conditional distribution
\begin{align}
&p((w_{1k},\ldots,w_{N_\alpha k},w_{\ast  k})_{k=1,\ldots,p}|(n_{ijk})_{1\leq i,j\leq N_\alpha{;}k=1,\ldots,p},\phi,\alpha)\nonumber\\
&~~~~\propto \left [\prod_{i=1}^{N_\alpha} \prod_{k=1}^{p} w_{ik}^{m_{ik}}\right ] e^{-\sum_{k=1}^p( w_{\ast k}+\sum_{i=1}^{N_\alpha} w_{ik} )^2}\left [\prod_{i=1}^{N_\alpha} \rho(w_{i1},\ldots,w_{ip};\phi)\right ]\alpha^{N_\alpha}g_{\ast\alpha}(w_{\ast 1},\ldots,w_{\ast p};\phi)\label{eq:jointposterior}
\end{align}
where $m_{ik}=\sum_{j=1}^{N_\alpha} n_{ijk}+n_{jik}$ and $g_{\ast\alpha}(w_{\ast 1},\ldots,w_{\ast p};\phi)$ is the probability density function of the random vector $(W_{1}([0,\alpha]),\ldots,W_{p}([0,\alpha]))$.

\end{proposition}
The proof can be straightforwardly adapted from that of \cite{Caron2014}, or from Proposition 5.2 of \cite{James2014} and is omitted here. It  builds on other posterior characterizations in Bayesian nonparametric models \citep{Prunster2002,James2002,James2005,James2009}.

\paragraph{Data augmentation.} Similarly to \cite{Caron2014}, we introduce latent count variables $\widetilde{n}%
_{ijk}=n_{ijk}+n_{jik}$ with
\begin{align}
(\widetilde{n}_{ij1},\ldots,\widetilde{n}_{ijp})|w,z&\sim\left\{
\begin{array}
[c]{ll}%
\delta_{(0,\ldots,0)} & \text{if }z_{ij}=0\\
\tPoi(2w_{i1}w_{j1},\ldots,2w_{ip}w_{jp}) & \text{if }%
z_{ij}=1\text{, }i\neq j
\end{array}
\right.\nonumber\\
\left(  \frac{\widetilde{n}_{ij1}}{2},\ldots,\frac{\widetilde{n}_{ijp}}%
{2}\right)  |w,z&\sim%
\begin{array}
[c]{ll}%
\tPoi(w_{i1}^{2},\ldots,w_{ip}^{2}) & \text{if }z_{ij}=1\text{, }i=j
\end{array}\label{eq:condlatent}
\end{align}
where $\tPoi(\lambda_{1},\ldots,\lambda_{p})$ is the multivariate
Poisson distribution truncated at zero, whose {pmf} is
\[
\tPoi(x_{1},\ldots x_{p};\lambda_{1},\ldots,\lambda_{p})=\frac
{\prod_{k=1}^{p}\text{Poisson}(x_{k};\lambda_{k})}{1-\exp(-\sum_{k=1}^{p}%
x_{k}\lambda_{k})}{\mathds 1}_{\left\{  \sum_{k=1}^{p}x_{k}>0\right\}  }.%
\]
One can sample from this distribution by first sampling $x=\sum_{k=1}^{p}%
x_{k}$ from a zero-truncated Poisson distribution with rate $\sum_{k=1}%
^{p}\lambda_{k}$ , and then $(x_{1},\ldots,x_{p})|(\lambda_{1},\ldots,\lambda_p),x\sim\Mult\left(  x,\left( \frac{\lambda_{1}}{\sum\lambda_{k}},\ldots
\frac{\lambda_{p}}{\sum\lambda_{k}}\right) \right) $.

\subsection{Markov chain Monte Carlo algorithm: General construction}

Using the data augmentation scheme together with the posterior characterization \eqref{eq:jointposterior}, we can derive the following MCMC sampler, which uses Metropolis-Hastings (MH) and Hamiltonian Monte Carlo (HMC) updates within a Gibbs sampler, and iterates as described in Algorithm~\ref{algo:MCMC}.

\begin{algorithm}
\caption{Markov chain Monte Carlo sampler for posterior inference{.}}
\label{algo:MCMC}

\flushleft At each iteration
\begin{enumerate}
\item Update $(w_{i1},\ldots,w_{ip})$, $i=1,\ldots {,}N_\alpha$ given the rest using MH or HMC{.}

\item Update hyperparameters $(\phi,\alpha)$ and total masses
$(w_{\ast 1},\ldots,w_{\ast p})$ given the rest using MH{.}

\item Update the latent variables given the rest using \eqref{eq:condlatent}{.}
\end{enumerate}
\end{algorithm}
In general, if the L\'evy intensity $\rho$ can be evaluated pointwise, one can use a MH update for step 1, but it would scale poorly with the number of nodes. Alternatively, if the L\'evy intensity $\rho$ is differentiable, one can use a Hamiltonian Monte Carlo update~\citep{Duane1987,Neal2011}.

The challenging part of the Algorithm~\ref{algo:MCMC} is Step 2. From Eq.~\eqref{eq:jointposterior} we have
\begin{align*}
&p((w_{\ast  k})_{k=1,\ldots,p},\phi,\alpha|\text{rest})\propto  p(\phi)p(\alpha)e^{-\sum_{k=1}^p( w_{\ast k}+\sum_{i=1}^{N_\alpha} w_{ik} )^2}\left [\prod_{i=1}^{N_\alpha} \rho(w_{i1},\ldots,w_{ip};\phi)\right ]\alpha^{N_\alpha}g_{\ast\alpha}(w_{\ast 1},\ldots,w_{\ast p};\phi).
\end{align*}
This conditional distribution is not of standard form and involves the multivariate pdf $g_{\ast\alpha}(w_{\ast 1},\ldots,w_{\ast p})$ of the random vector $(W_{1}([0,\alpha]),\ldots,W_{p}([0,\alpha]))$ for which there is typically no analytical expression. All is available is its Laplace transform, which is given by
\begin{align}
\mathbb{E}\left[  e^{-\sum_{k=1}^{p}%
t_{k}W_{k}([0,\alpha])}\right] = e^{-\alpha\psi(t_{1},\ldots,t_{p};\phi)}
\end{align}
where
\begin{align}
\psi(t_{1},\ldots,t_{p};\phi)   =\int_{{\mathbb R^p_+}}\left(  1-e^{-\sum_{k=1}^{p}t_{k}w_{k}}\right)
\rho(dw_{1},\ldots,dw_{p};\phi)\label{eq:Laplaceexponent}
\end{align}
is the multivariate Laplace exponent, which involves a $p$-dimensional integral. We propose to use a Metropolis-Hastings step, with proposal
\begin{align}
q(\widetilde w_{\ast 1:p},\widetilde \phi,\widetilde \alpha|w_{\ast 1:p}, \phi,\alpha)&=q(\widetilde w_{\ast 1:p}|w_{\ast 1:p}, \widetilde \phi,\widetilde \alpha)\times q(\widetilde \phi| \phi)\times q(\widetilde \alpha| \alpha,\widetilde\phi,w_{\ast 1:p})
\end{align}
where
\begin{align}
q(\widetilde \alpha| \alpha,\widetilde \phi,w_{\ast 1:p})&=\Gam(\widetilde \alpha;a_\alpha+N_\alpha, b_\alpha+\psi(\lambda_1,\ldots,\lambda_p;\widetilde\phi))
\end{align}
and the proposal for $w_{\ast 1:p}$ is an exponentially tilted version of $g_{\ast \alpha}$
\begin{align}
q((\widetilde w_{\ast k})_{k=1,\ldots,p}|(w_{\ast k})_{k=1,\ldots,p}, \widetilde \phi)&=\frac{e^{-\sum_{k=1}^p \lambda_k \widetilde w_{\ast k}}g_{\ast\widetilde\alpha}(\widetilde w_1,\ldots,\widetilde w_p;\widetilde \phi) }{e^{-{\widetilde\alpha \psi(\lambda_1,\ldots,\lambda_p;\widetilde \phi)}}}\label{eq:proposalwstar}
\end{align}
where $\lambda_k=w_{\ast k}+2\sum_{i=1}^{N_\alpha} w_{ik}$ and $q(\widetilde \phi| \phi)$ can be freely specified by the user. This leads to the following acceptance ratio
\begin{align*}
r&=\frac{p(\widetilde\phi)q(\phi|\widetilde\phi)}{p(\phi)q(\widetilde\phi|\phi)}\left[\prod_{i=1}^{N_\alpha}\frac{\rho(w_{i1},\ldots,w_{ip};\widetilde\phi)}{\rho(w_{i1},\ldots,w_{ip};\phi)}\right]
\left[ \frac{b_\alpha + \psi(\widetilde\lambda_1,\ldots,\widetilde\lambda_p;\phi)}{b_\alpha + \psi(\lambda_1,\ldots,\lambda_p;\widetilde\phi)} \right ]^{a_\alpha + N_\alpha} e^{\sum_{k=1}^p  [ w_{\ast k}^2 - \widetilde w_{\ast k}^2]}
\end{align*}
where $\widetilde \lambda_k=\widetilde w_{\ast k}+2\sum_{i=1}^{N_\alpha} w_{ik}$. This acceptance ratio involves evaluating the multivariate exponent \eqref{eq:Laplaceexponent}.\medskip

In the general case, the MCMC algorithm \ref{algo:MCMC} thus requires to be able to
\begin{enumerate}
\item[(a)] evaluate pointwise the L\'evy intensity $\rho$, and potentially differentiate it,
\item[(b)] evaluate pointwise the {Laplace} exponent \eqref{eq:Laplaceexponent} and
\item[(c)] sample from the exponentially tilted distribution \eqref{eq:proposalwstar}.
\end{enumerate}
Regarding point (c), the random variable with pdf \eqref{eq:proposalwstar} has the same distribution as the random vector $\left(W_1^\prime([0,\alpha]{)},\ldots,W_p^\prime([0,\alpha]) \right)$ where $(W_1^\prime,\ldots,W_p^\prime)\sim \CRM(\rho^\prime,\lambda)$ with $\rho^\prime$ is an exponentially tilted version of $\rho$
\begin{equation}
\rho^\prime(w_1,\ldots,w_p)=e^{-\sum_k \lambda_k w_k}\rho(w_1,\ldots,w_p).\
\end{equation}
By considering an approximate tilted intensity $\rho^{\varepsilon~\prime}(w_1,\ldots,w_p)$, one can approximately sample from \eqref{eq:proposalwstar} by simulating points from a Poisson process with mean measure $\alpha{\rho^{\varepsilon}~}^\prime(w_1,\ldots,w_p)$ and summing them up.

\subsection{Markov chain Monte Carlo algorithm: Compound CRMs}

The hierarchical construction of CCRM{s} enables to derive a certain number of simplifications in the algorithm described in the previous section. Using the construction $w_{ik}=\beta_{ik}w_{i0}$ where the points $(w_{i0},\beta_{i1},\ldots,\beta_{ip})_{i=1,2,\ldots}$ have L\'evy measure \eqref{eq:Levyw0beta}, we aim at approximating the posterior
\begin{equation}
p((w_{10},\ldots,w_{N_\alpha 0}),(\beta_{1 k},\ldots,\beta_{N_\alpha k},w_{\ast  k})_{k=1,\ldots,p},\phi,\alpha|(z_{ij})_{1\leq i,j\leq N_\alpha}).\label{eq:posteriorCCRM}
\end{equation}
 Conditional on the latent count variables defined in \eqref{eq:condlatent}, we have the following conditional characterization, similar to \eqref{eq:jointposterior}
\begin{align}
&p((w_{10},\ldots,w_{N_\alpha 0}),(\beta_{1k},\ldots,\beta_{N_\alpha k},w_{\ast  k})_{k=1,\ldots,p}|(n_{ijk})_{1\leq i,j\leq N_\alpha{;}k=1,\ldots,p},\phi,\alpha)\nonumber\\
&~~~~~~~~\propto \left [\prod_{i=1}^{N_\alpha} w_{i0}^{m_{i}} \prod_{k=1}^{p} \beta_{ik}^{m_{ik}}\right ] e^{-\sum_{k=1}^p( w_{\ast k}+\sum_{i=1}^{N_\alpha} w_{ik} )^2-\sum_{i=1}^{N_\alpha} w_{i0}(\sum_{k=1}^p \gamma_k\beta_{ik}) }\nonumber\\
&~~~~~~~~~~\times\left [\prod_{i=1}^{N_\alpha} f(\beta_{i1},\ldots,\beta_{ip};\phi)\rho_0(w_{i0};\phi)\right ]\alpha^{N_\alpha}g_{\ast\alpha}(w_{\ast 1},\ldots,w_{\ast p};\phi)\label{eq:jointposteriorCCRM}
\end{align}
where $m_i=\sum_{k=1}^p m_{ik}$ and $f$ and $\rho_0$ are densities of $F$ and $\rho_0$ with respect to the Lebesgue measure.\medskip

If $f$ and $\rho_0$ are differentiable, one can use a HMC update for Step 1 of Algorithm~\ref{algo:MCMC}. In particular, when they take the form \eqref{eq:Fproductgamma} and \eqref{eq:levyGGP}, we obtain the following simple expressions for the gradient:
\begin{align*}
\frac{\partial U(q)}{d(\log w_{i0})}&=m_{i}-\sigma-w_{i0}\left[\tau+2\sum_{k=1}^{p}\beta_{ik}\left(w_{\ast k}+\sum_{{j}=1}^{N_\alpha}w_{{j}0}\beta_{{j}k}\right)\right],~~~i=1,\ldots,N_\alpha,\\
 \frac{\partial U(q)}{d(\log\beta_{ik})}&=m_{ik}+a_{k}-\beta_{ik}\left[b_{k}+2w_{i0}\left(w_{\ast k}+\sum_{{j}=1}^{N_\alpha}w_{{j}0}\beta_{{j}k}\right)\right],~~~i=1,\ldots,N_\alpha,~~~k=1,\ldots,p,
\end{align*}
where $U(q) = \log p\left(q\vert\text{rest}\right)$ with $q=\left(\log w_{i0},\log\beta_{i1},\ldots,\log\beta_{ip}\right)_{i=1,\ldots,N_{\alpha}}$.

Regarding Step 2 of Algorithm~\ref{algo:MCMC}, the Laplace exponent for CCRM takes the simple form
\begin{align}
\psi(t_{1},\ldots,t_{p})  =\int_{0}^{\infty}\left[  M(-w_{0}\gamma_{1},\ldots,-w_{0}\gamma
_{p})-M\left(-w_{0}(t_{1}+\gamma_{1}),\ldots,-w_{0}(t_{p}+\gamma_{p})\right)\right]
\rho_{0}(dw_{0})
\end{align}
which only requires evaluating a one-dimensional integral, whatever the number $p$ of communities, and this can be done numerically. For the specific model defined by \eqref{eq:Fproductgamma} and \eqref{eq:levyGGP}, we obtain
\begin{align*}
\psi(t_{1},\ldots,t_{p})  =\frac{1}{\Gamma(1-\sigma)} \int_{0}^{\infty}\left[  1-\prod_{k=1}^{p}\left(  1+\frac{w_{0}t_{k}%
}{b_{k}+w_{0}\gamma_{k}}\right)  ^{-a_{k}}\right]  \left[  \prod_{k=1}%
^{p}\left(  1+\frac{w_{0}\gamma_{k}}{b_{k}}\right)  ^{-a_{k}}\right] w_0^{-1-\sigma}e^{-w_0\tau}dw_0.
\end{align*}
Finally, we need to sample total masses $(w_{\ast 1},\ldots,w_{\ast p})$ from \eqref{eq:proposalwstar}, and this can be done by simulating points $(w_{{i0}},\beta_{i1},\ldots,\beta_{ip})_{i=1,2,\ldots}$ from a Poisson process with exponentially tilted L\'evy intensity
\begin{equation}
\alpha e^{-w_{0}\sum_{k=1}^{p}(\gamma_{k}+\lambda_k)\beta_{k}}f(\beta_{1},\ldots,\beta_{p}%
)\rho_{0}(w_{0})\label{eq:Levyw0betatilted}
\end{equation}
and summing up the weights $w_{\ast k}=\sum_{i=1,2,\ldots} w_{{i0}}\beta_{ik}$ for $k=1,\ldots,p$. For infinite-activity CRMs, this is not feasible, and we suggest to resort to the approximation of \cite{Cohen2007}. More precisely, we write $$(w_{\ast 1},\ldots,w_{\ast p}) =X_\varepsilon + X^{\varepsilon}$$ where the random vectors $X_\varepsilon\in\mathbb R_+^p$ and $X^{\varepsilon}\in\mathbb R_+^p$ are defined as  $X_\varepsilon=\sum_{i|w_{{i0}}<\varepsilon} w_{{i0}}(\beta_{i1},\ldots,\beta_{ip})$ and $X^\varepsilon=\sum_{i|w_{{i0}}>\varepsilon} w_{{i0}}(\beta_{i1},\ldots,\beta_{ip})$. We can sample a realization of the random vector $X^{\varepsilon}$ exactly by simulating the points of a Poisson process with mean intensity
\begin{equation}
\alpha e^{-w_{0}\sum_{k=1}^{p}(\gamma_{k}+\lambda_k)\beta_{k}}f(\beta_{1},\ldots,\beta_{p}%
)\rho_{0}(w_{0}){\mathds 1}_{w_0>\varepsilon}
\end{equation}
See Section~\ref{sec:simulation} and Appendix~\ref{sec:app:simujumps} for details. The positive random vector $X_\varepsilon$ is approximated by a truncated Gaussian random vector with mean $\mu_{\varepsilon}$ and variance $\Sigma_{\varepsilon}$ such that
\begin{align*}
\mu_{\varepsilon}  &  =\alpha\int_{\mathbb{R}_{+}^{p}}w_{1:p}\rho_{\varepsilon
}(dw_{1},\ldots,dw_{p})\\
\Sigma_{\varepsilon}  &  =\alpha\int_{\mathbb{R}_{+}^{p}}w_{1:p}w_{1:p}^{T}\rho
_{\varepsilon}(dw_{1},\ldots,dw_{p})
\end{align*}
where
\[
\rho_{\varepsilon}(dw_{1},\ldots,dw_{p})=e^{-\sum_{k=1}^{p}(\gamma_{k}+\lambda_k)w_{k}}%
\int_{0}^{\varepsilon}w_{0}^{-p}F\left(  \frac{dw_{1}}{w_{0}},\ldots
,\frac{dw_{p}}{w_{0}}\right)  \rho_{0}(dw_{0}).
\]
Note that $\mu_{\varepsilon}$ and $\Sigma_{\varepsilon}$ can both be expressed as one-dimensional integrals using the gradient and Hessian of the moment generating function $M$ of $F$. Theorem~\ref{th:gaussianapprox} {in Appendix~\ref{sec:app:gaussianapprox}}, which is an adaptation of the results of \cite{Cohen2007} to CCRM, gives the conditions on the parameters of CCRM under which
 \[
\Sigma_{\varepsilon}^{-1/2}(X_{\varepsilon}-\mu_{\varepsilon}%
)\overset{d}{\rightarrow}\mathcal{N}(0,I_{p}) {~\text{as}~\varepsilon\rightarrow 0}
\]
and thus the approximation is asymptotically valid. The Gaussian approximation is in particular asymptotically valid for the CCRM defined by \eqref{eq:Fproductgamma} and \eqref{eq:levyGGP} when $\sigma\in(0,1)$, hence is valid for all infinite-activity cases except $\sigma=0$.

Note that due to the Gaussian approximation in the proposal distribution for $(w_{\ast\alpha})$, {A}lgorithm~\ref{algo:MCMC} does not actually admit the posterior distribution \eqref{eq:posteriorCCRM} as invariant distribution, and is an approximation of an exact MCMC algorithm targeting this distribution. We observe in the experimental section that this approximation provides very reasonable results for the examples considered.

\section{Experiments}
\label{sec:experiments}
\subsection{Simulated data}
\label{sub:Simulated-data}
We first study the convergence of the MCMC algorithm on synthetic
data simulated from the CCRM based graph model described in Section~\ref{sec:model}  where $F$ and $\rho_0$ take the form \eqref{eq:Fproductgamma} and \eqref{eq:levyGGP}. We generate an undirected
graph with $p=2$ communities and parameters $\alpha=200$, $\sigma=0.2$,
$\tau=1$, $b_{k}=b=\frac{1}{p}$, $a_{k}=a=0.2$ and $\gamma_{k}=\gamma=0$.
The sampled graph has $1121$ nodes and $\num{6090}$ edges. For
the inference, we consider that $b$ and $\gamma$ are known and we
assume a vague prior $\mbox{Gamma}(0.01,0.01)$ on the unknown parameters
$\alpha$ and $\phi=(1-\sigma,\tau,a)$. We run $3$ parallel MCMC
chains with different initial values. Each chain starts with $\num{10000}$
iterations using our model with only one community where the scores
$\beta$ are fixed to $1$, which is equivalent to the model of \citet{Caron2014}. We then run $\num{200000}$ iterations
using our model with $p$ communities. We use $\varepsilon=10^{-3}$ as
a truncation level for simulating $w_{\ast1:p}$ and $L=10$ leapfrog
steps for the HMC. The stepsizes of both the HMC and the random walk
MH on $\left(\log(1-\sigma),\log\tau,\log a\right)$ are adapted during
the first $\num{50000}$ iterations so as to target acceptance ratios
of $0.65$ and $0.23$ respectively. The computations take around $1$h$10$ using Matlab on
a standard desktop computer.
Trace plots of the parameters $\log\alpha$, $\sigma$, $\tau$, $a$
and $\overline{w}_{\ast}=\frac{1}{p}\sum_{k=1}^{p}w_{\ast k}$ and
histograms based on the last $\num{50000}$ iterations are given in
Figure~\ref{fig:Traces-1}. Posterior samples clearly converge around the sampled value. Aiming to study further the rate of convergence of our algorithm, we explore the impact of a lower threshold value $\varepsilon $ and a higher number of iterations. Firstly, decreasing the threshold $\varepsilon$ to a value $ \ll 10^{-3}$ does not lead to any noticeable change in the MCMC histograms, suggesting that the target distribution of our approximate MCMC is very close to the posterior distribution of interest. Similarly, our approximate MCMC posterior obtained from running the algorithm for $10^6$ iterations is very close to the one obtained from $200000$ iterations.
%Hence, we report the results of our algorithm on $200000$ iterations, always being careful how to interpret the output having in mind that the posterior might be mutlimodal.

\begin{figure}[th]
\begin{centering}
\subfloat[$\log\alpha$]{\includegraphics[width=0.25\textwidth]{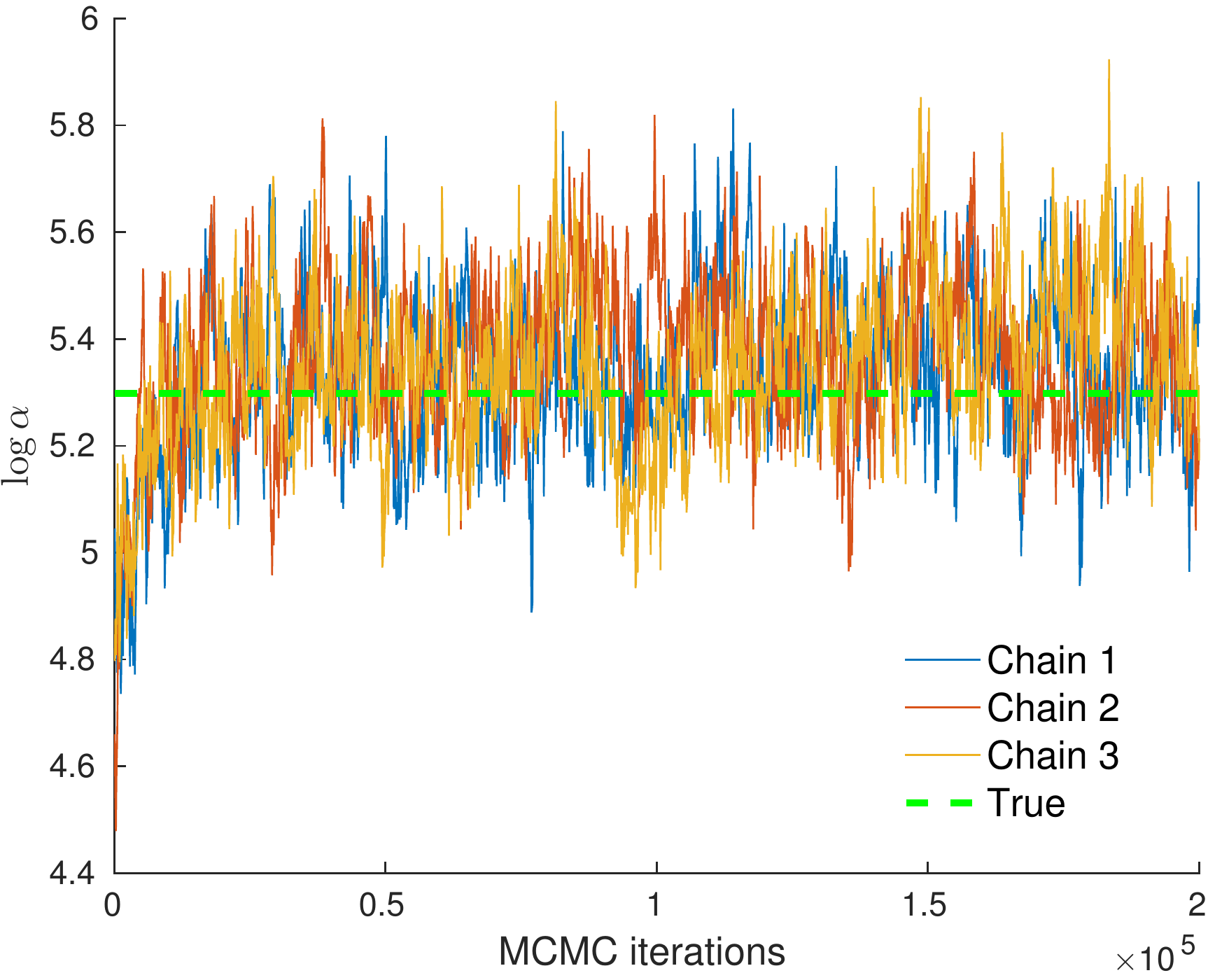}\includegraphics[width=0.25\textwidth]{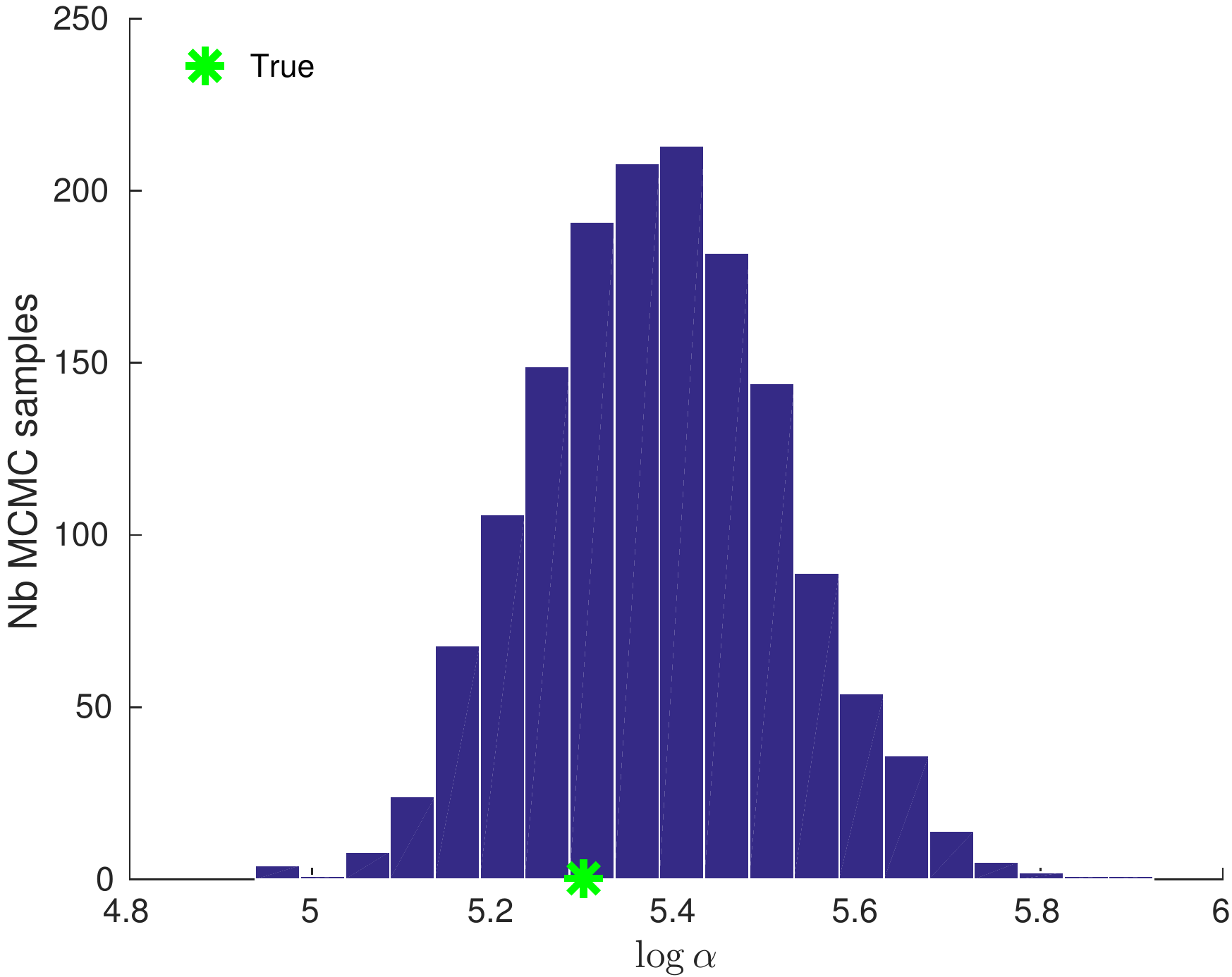}}\subfloat[$\sigma$]{\includegraphics[width=0.25\textwidth]{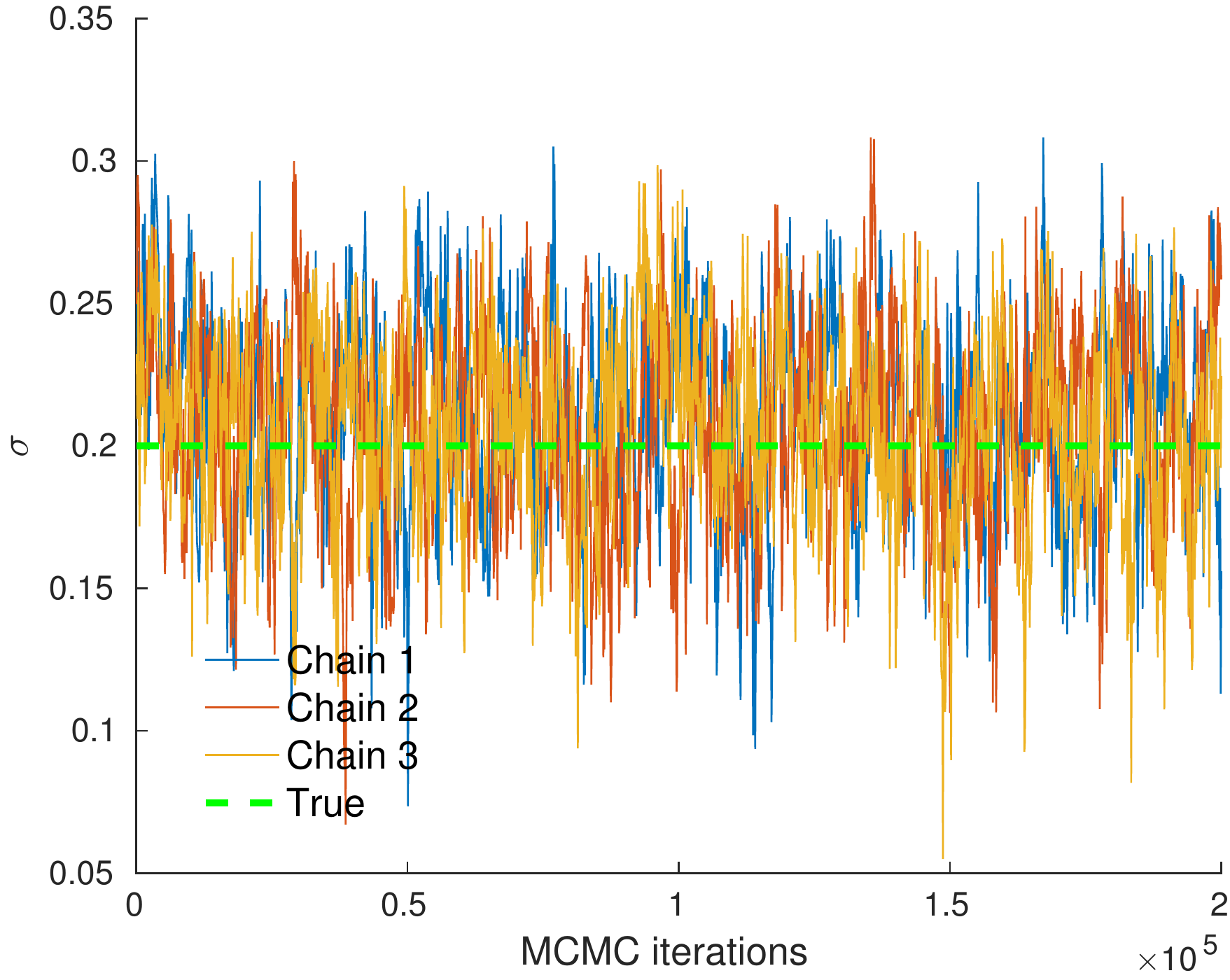}\includegraphics[width=0.25\textwidth]{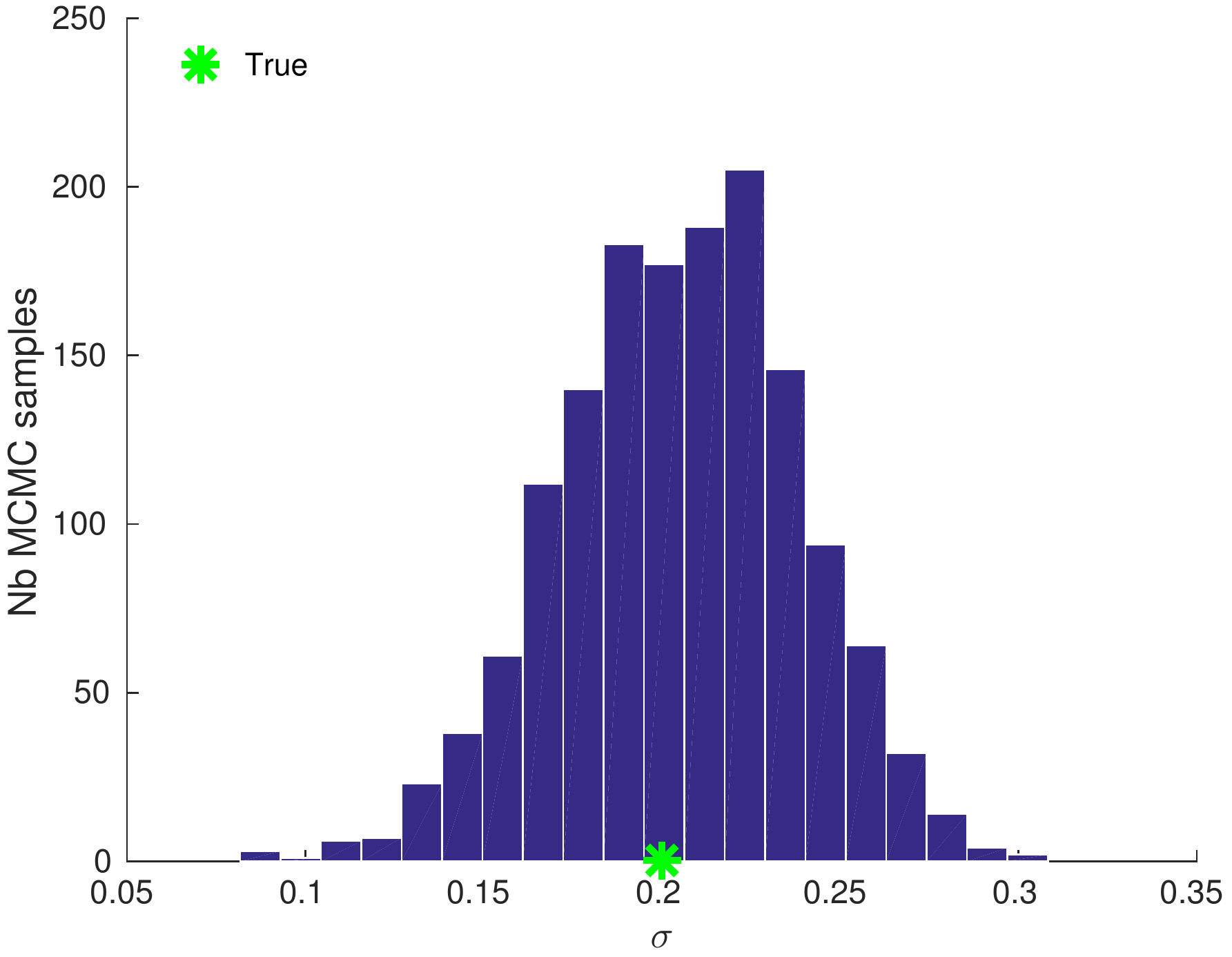}}
\par\end{centering}

\begin{centering}
\subfloat[$\tau$]{\includegraphics[width=0.25\textwidth]{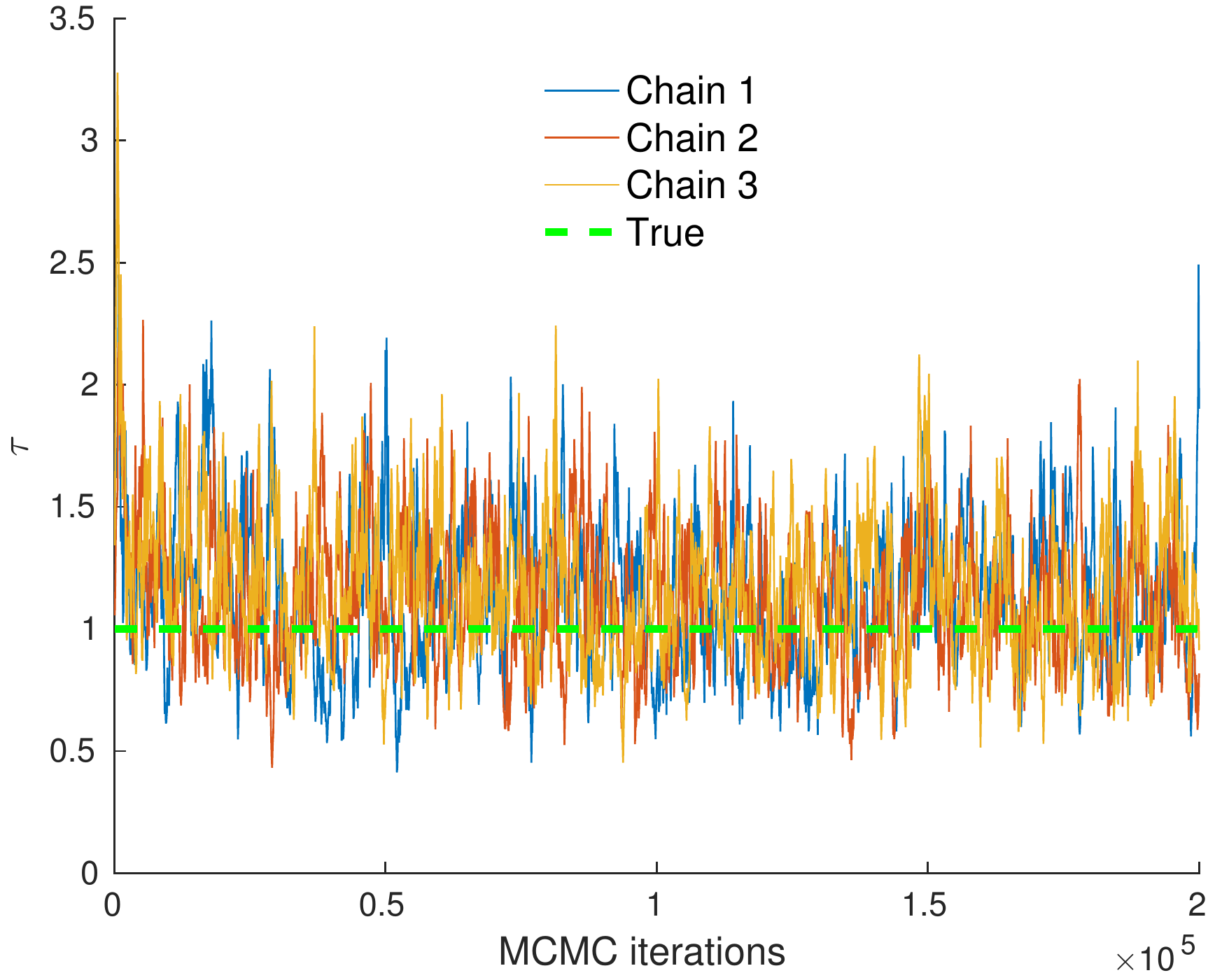}\includegraphics[width=0.25\textwidth]{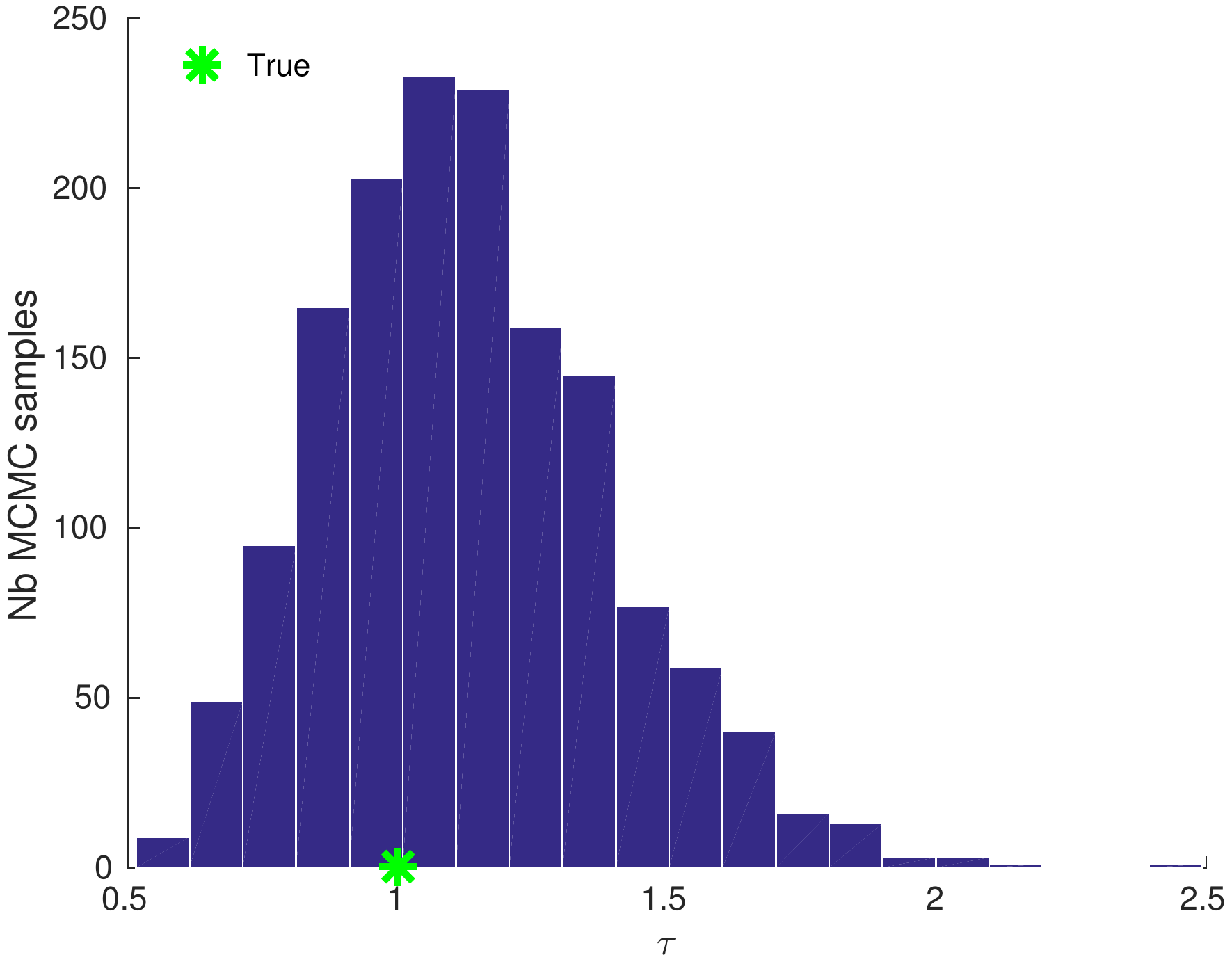}}\subfloat[$a$]{\includegraphics[width=0.25\textwidth]{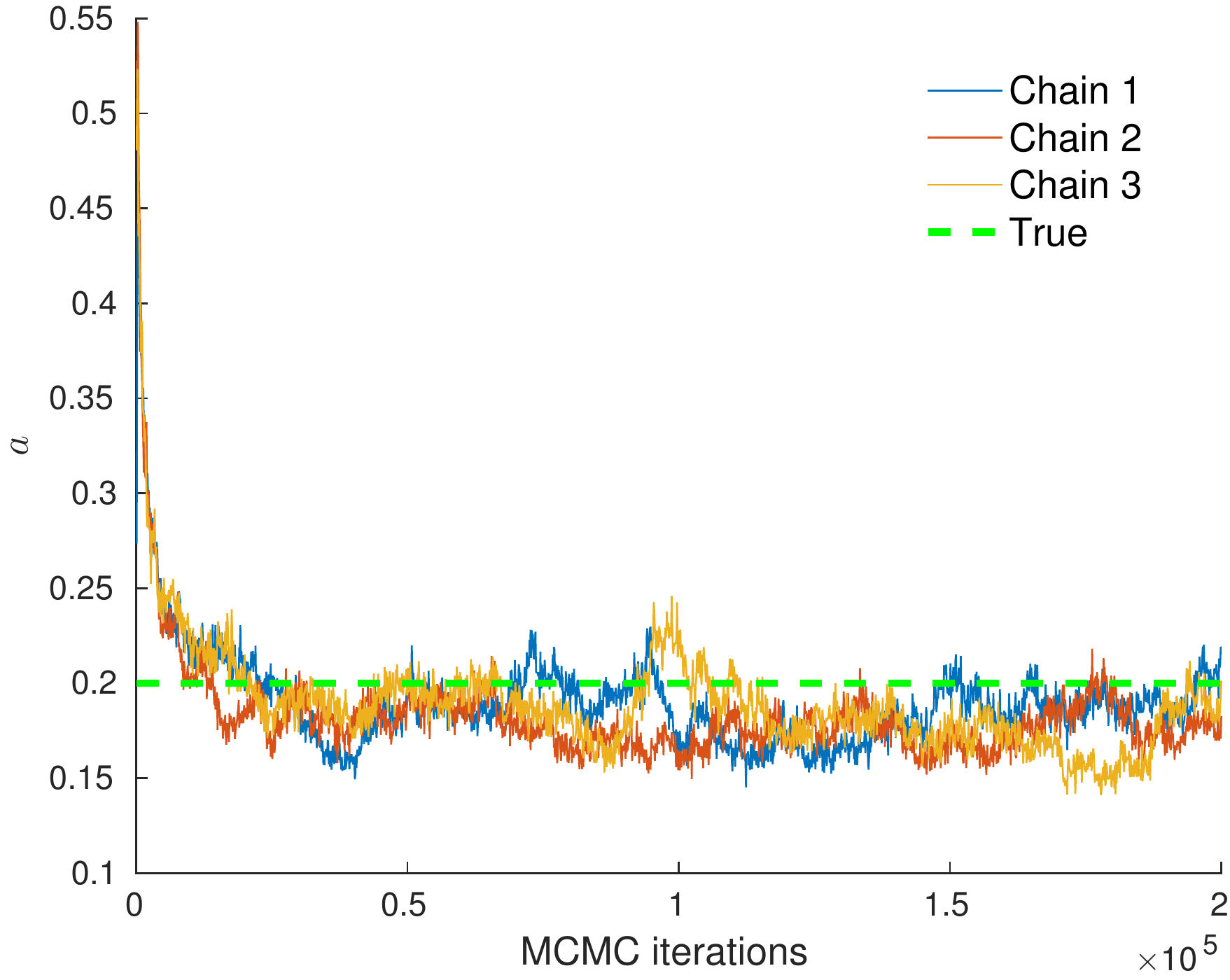}\includegraphics[width=0.25\textwidth]{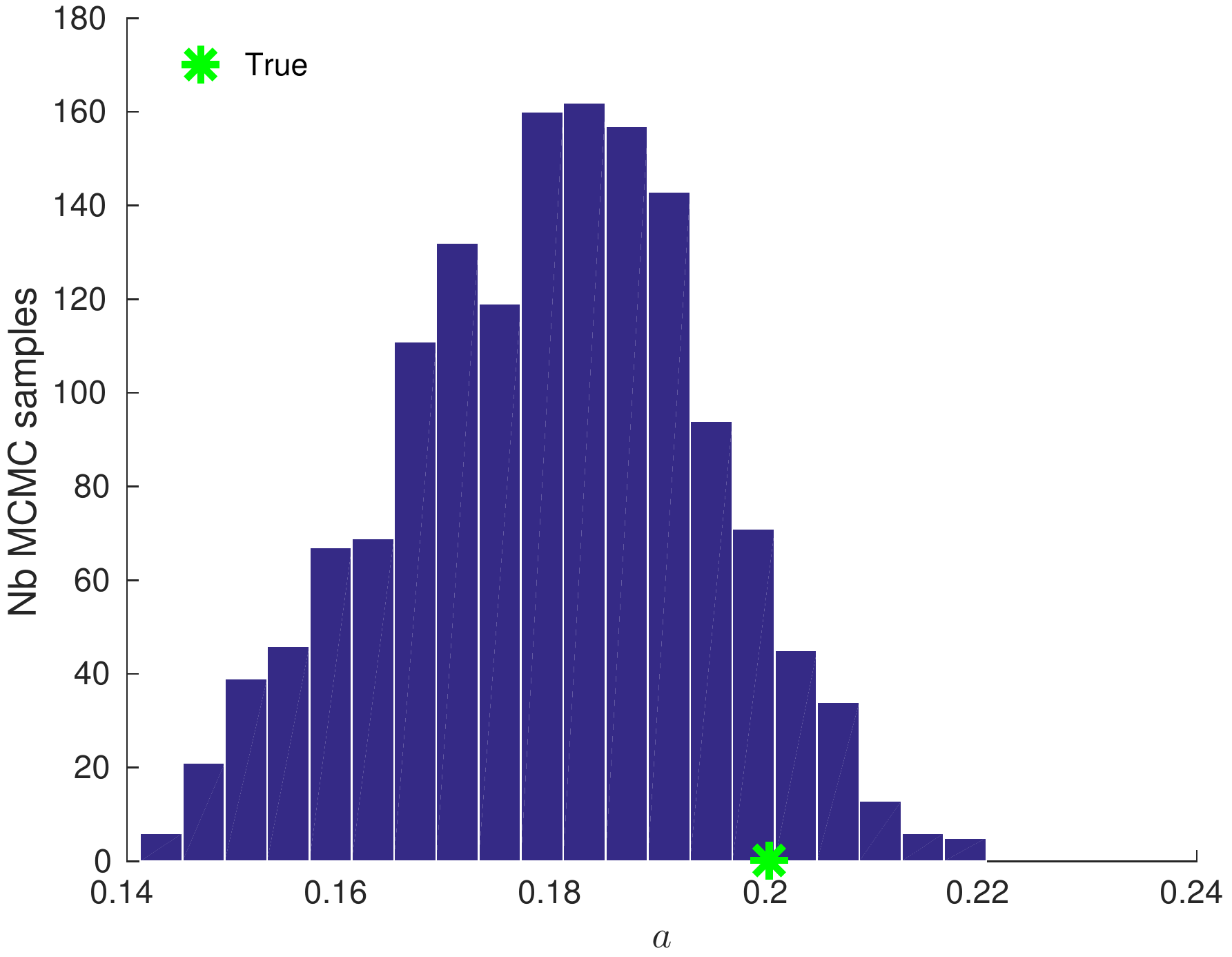}}
\par\end{centering}

\begin{centering}
\subfloat[$\overline{w}_{\ast}$]{\includegraphics[width=0.25\textwidth]{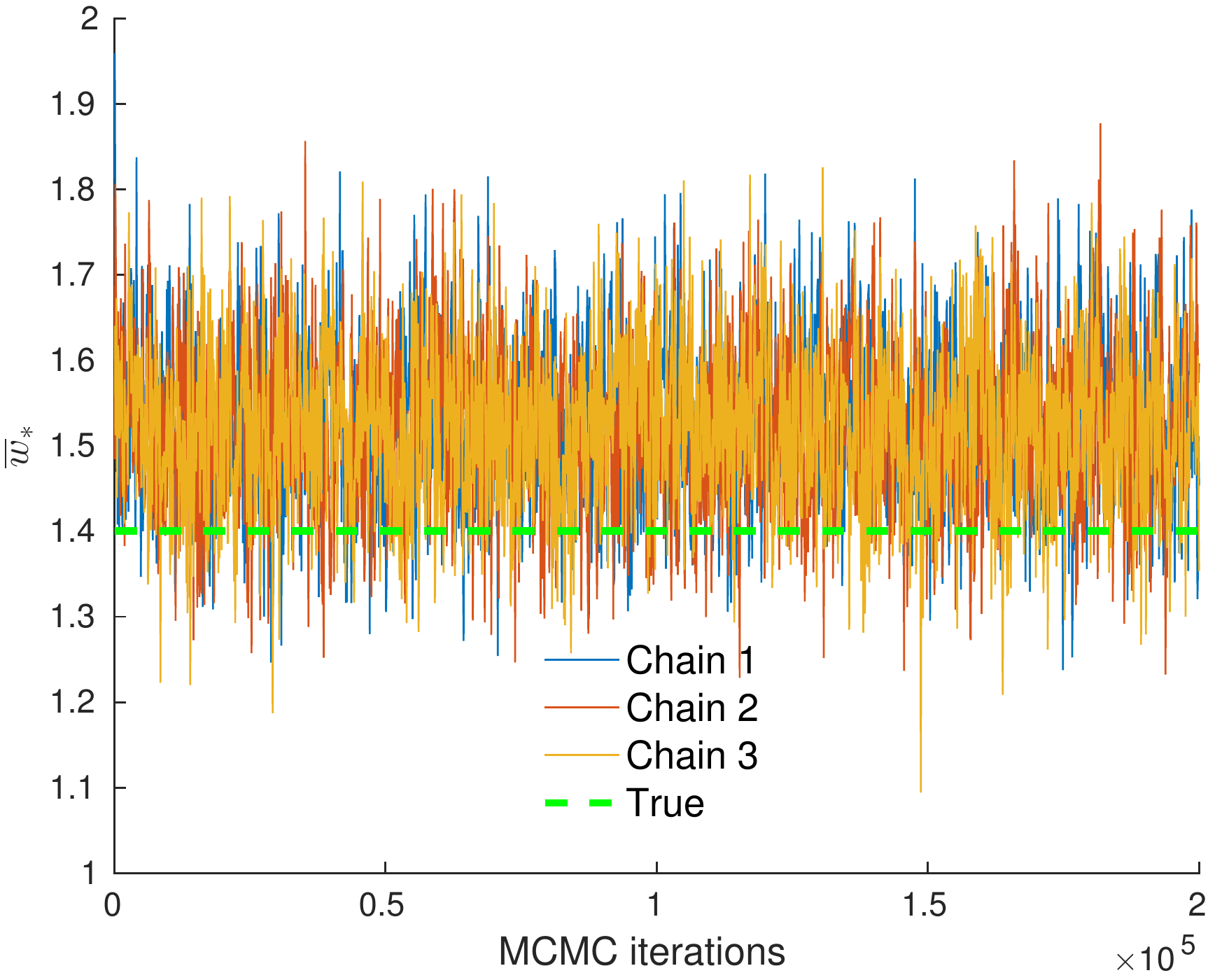}\includegraphics[width=0.25\textwidth]{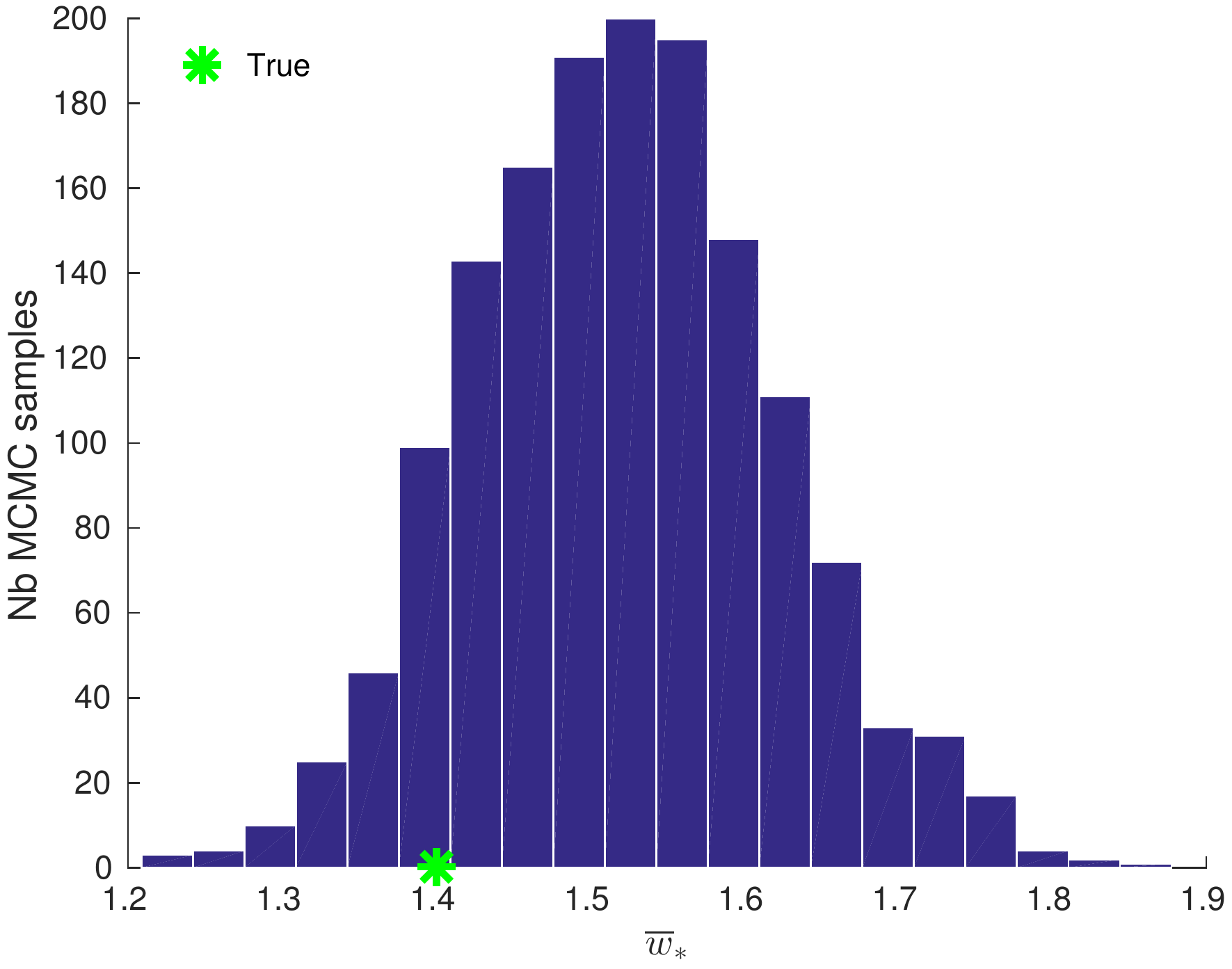}}
\par\end{centering}

\caption{\label{fig:Traces-1}MCMC trace plots (left) and histograms (right)
of parameters (a) $\log\alpha$, (b) $\sigma$, (c) $\tau$, (d) $a$
and (e) $\overline{w}_{\ast}$ for a graph generated with parameters
$p=2$, $\alpha=200$, $\sigma=0.2$, $\tau=1$, $b=\frac{1}{p}$,
$a=0.2$ and $\gamma=0$.}
\end{figure}

Our model is able to accurately recover the mean parameters
of both low and high degree nodes and to provide reasonable credible intervals, as shown in Figure~\ref{fig:Posterior-credible-intervals}(a-b). By generating $\num{5000}$ graphs from the posterior predictive
we assess that our model fits the empirical power-law degree distribution
of the sparse generated graph as shown in Figure~\ref{fig:Posterior-credible-intervals}(c). We demonstrate the interest of our nonparametric approach by
comparing these results to the ones obtained with the parametric version
of our model. To achieve this, we fix $w_{\ast k}=0$ and force
the model to lie in the finite-activity domain by assuming $\sigma\in(-\infty,0{)}$ and
using the prior distribution $-\sigma\sim\mbox{Gamma}(0.01,0.01)$. Note that in this case, the model is equivalent to that of \cite{Zhou2015}.
As shown in Figure~\ref{fig:finite-activity}(a-b), the parametric
model is able to recover the mean parameters of nodes
with high degrees, and credible intervals are similar to that obtained with the full model; however, it fails to provide reasonable credible intervals for nodes with low degree. In addition,
as shown in Figure~\ref{fig:finite-activity}(c), the posterior predictive
degree distribution does not fit the data, illustrating the unability
of this parametric model to capture power-law behaviour.

\begin{figure}[th]
\begin{centering}
\subfigure[50 nodes with highest degrees]{\includegraphics[width=0.33\textwidth]{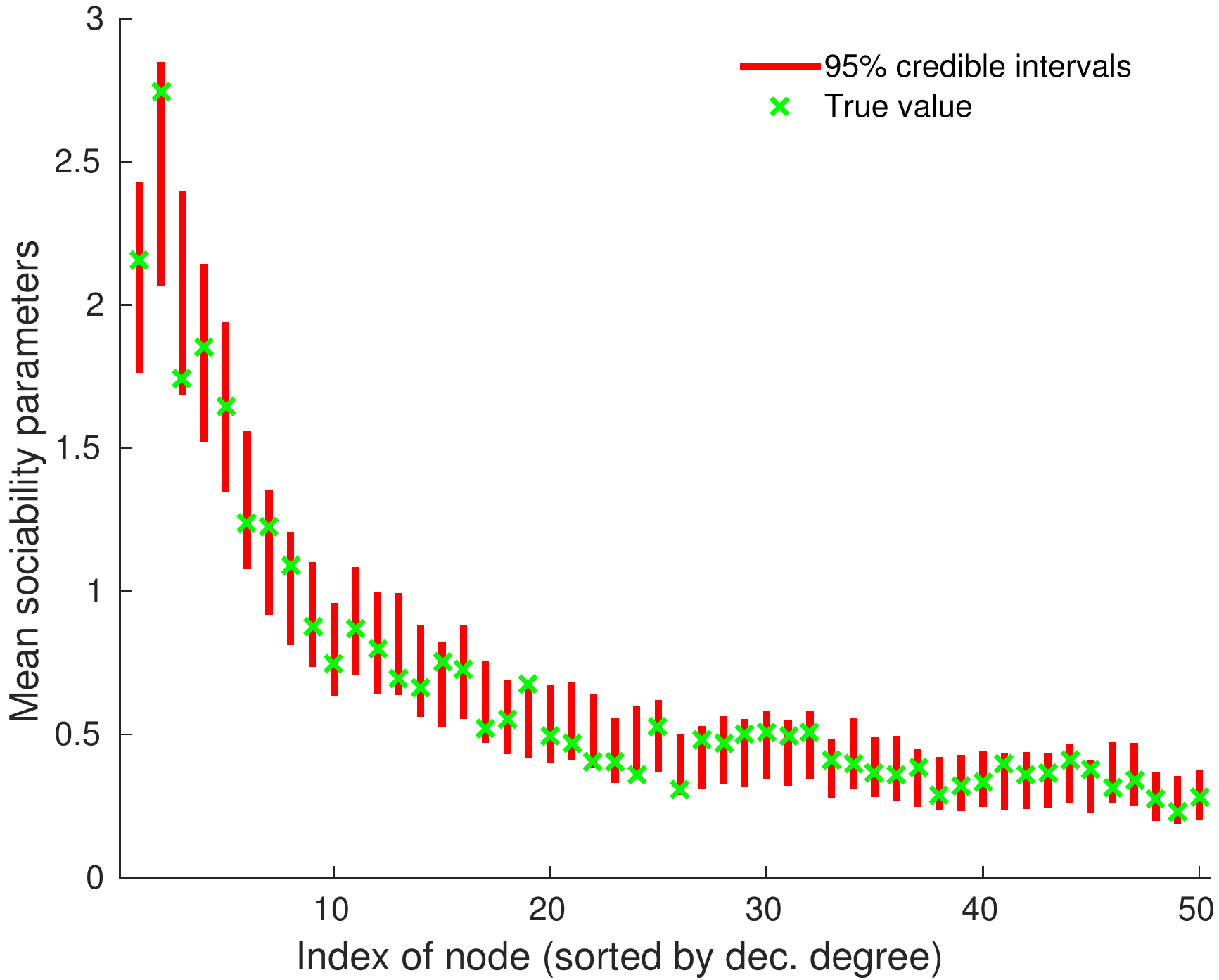}

}\subfigure[50 nodes with lowest degrees]{\includegraphics[width=0.33\textwidth]{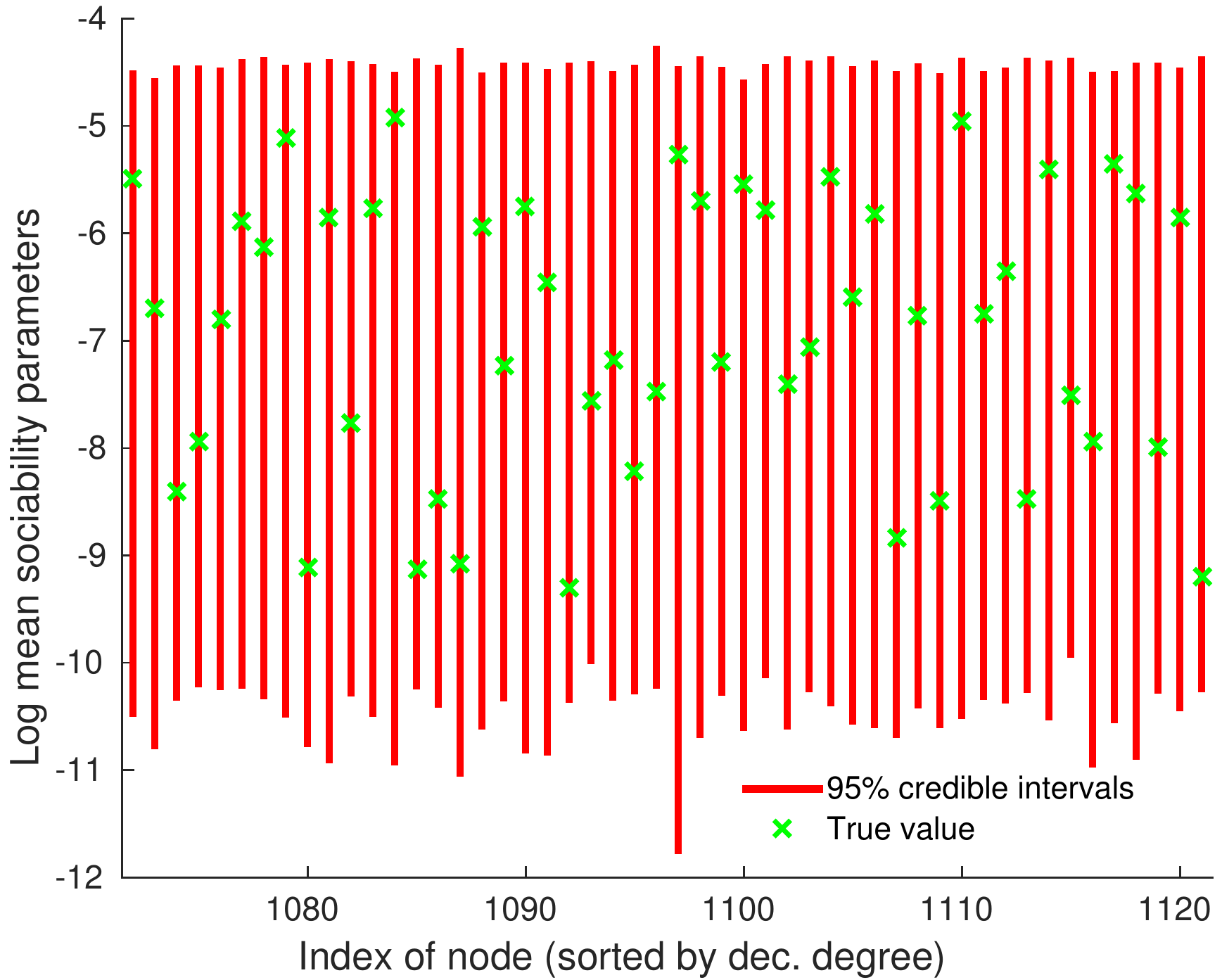}

}\subfigure[Degree distribution]{\includegraphics[width=0.33\textwidth]{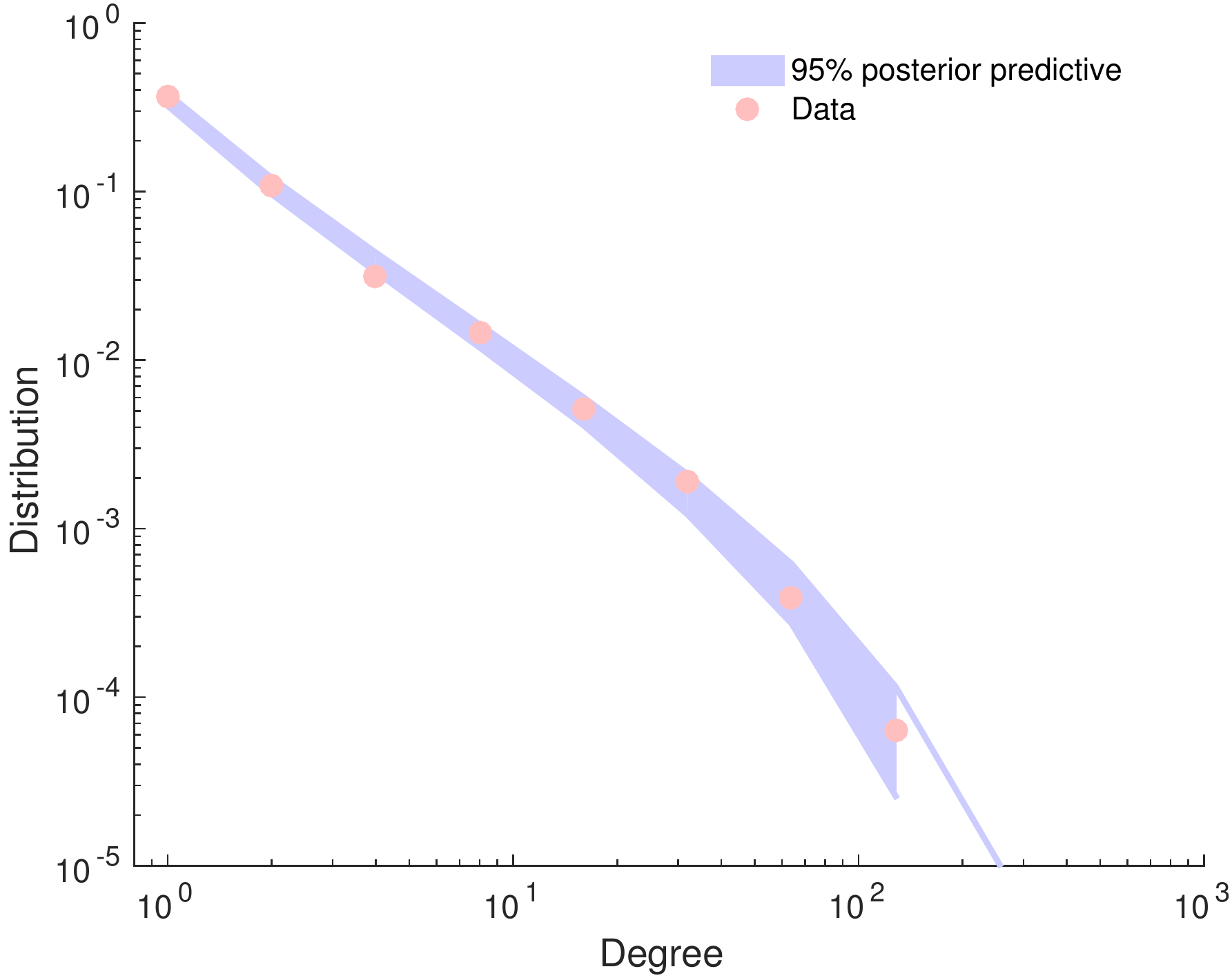}

}
\end{centering}

\caption{\label{fig:Posterior-credible-intervals}95\% posterior credible intervals
and true values of (a) the mean parameters $\overline{w}_{i}=\frac{1}{p}\sum_{k=1}^{p}w_{ik}$
of the 50 nodes with highest degrees and (b) the log mean
parameters $\log\overline{w}_{i}$ of the 50 nodes with lowest degrees.
(c) Empirical degree distribution and 95\% posterior predictive credible
interval. Results obtained for a graph generated with parameters $p=2$,
$\alpha=200$, $\sigma=0.2$, $\tau=1$, $b=\frac{1}{p}$, $a=0.2$
and $\gamma=0$.}
\end{figure}

\begin{figure}[th]
\begin{centering}
\subfigure[50 nodes with highest degrees]{\includegraphics[width=0.33\textwidth]{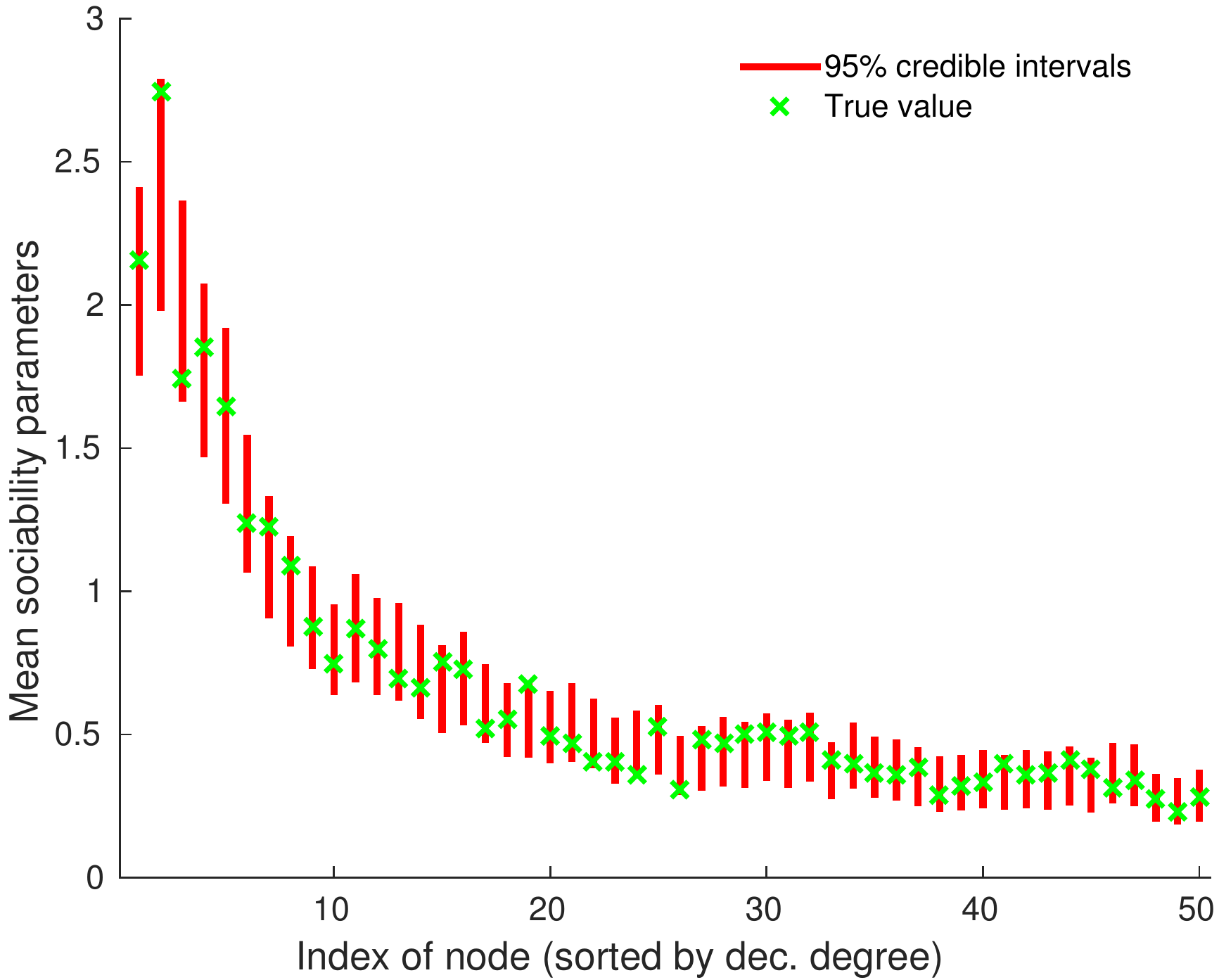}

}\subfigure[50 nodes with lowest degrees]{\includegraphics[width=0.33\textwidth]{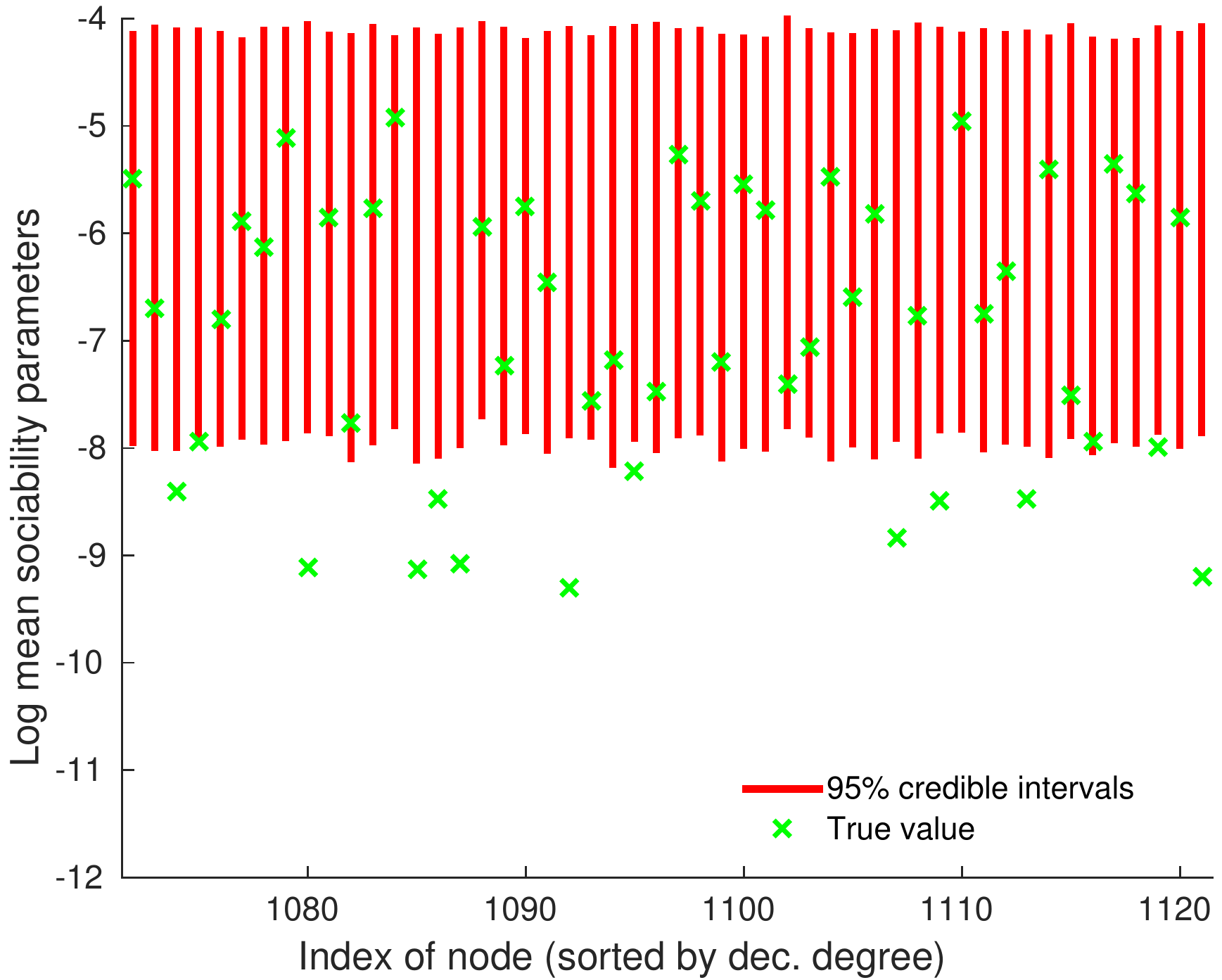}

}\subfigure[Degree distribution]{\includegraphics[width=0.33\textwidth]{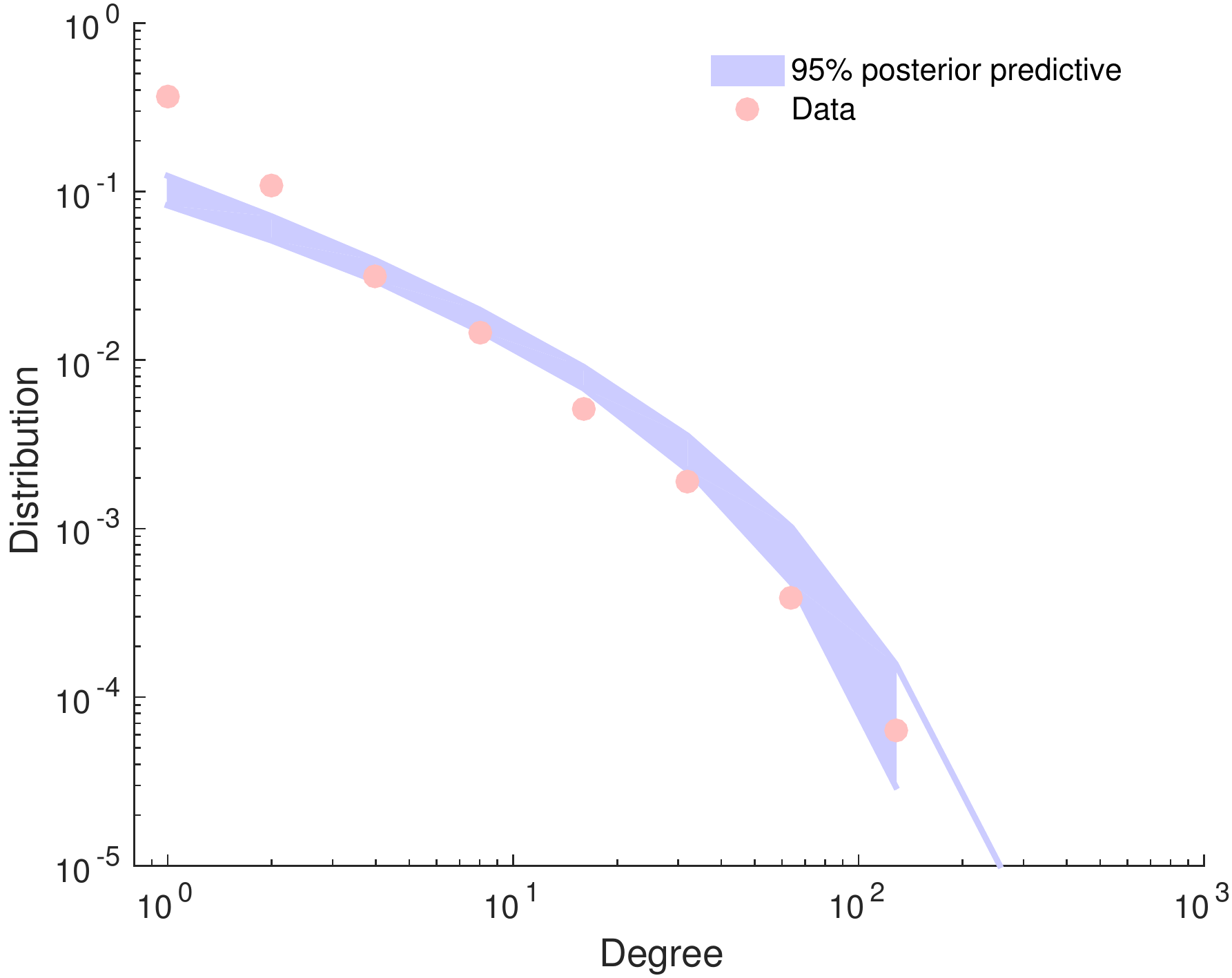}

}
\end{centering}
\caption{\label{fig:finite-activity}See Figure~\ref{fig:Posterior-credible-intervals}.
Results obtained on the same generated graph by inferring a finite-activity
model with $w_{\ast k}=0$ and $\sigma\leq0$.}
\end{figure}

\subsection{Real-world graphs}

We now apply our methods to learn the latent communities of two real-world
undirected simple graphs. The first network to be considered, the \textbf{polblogs} network~\citep{Adamic2005}, is the network of the American political blogosphere in February 2005\footnote{\url{http://www.cise.ufl.edu/research/sparse/matrices/Newman/polblogs}}. Two blogs are considered as connected if there is at least one hyperlink from one blog to the other. Additional information on the political leaning of each blog (left/right) is also available. The second network, named \textbf{USairport}, is the network of airports with at least one connection to a US airport in 2010\footnote{\url{http://www.transtats.bts.gov/DL_SelectFields.asp?Table_ID=292}}.

\begin{table}[th]
\begin{centering}
\caption{\label{tab:Size-of-real-world}Size of the networks, number of communities and computational time.}

\par\end{centering}

\centering{}%
\begin{tabular}{l|l|l|l|l}
Name & Nb nodes & Nb edges & Nb communities $p$ & Time\tabularnewline
\hline
\texttt{polblogs} & $\num{1224}$  & $\num{16715}$ & $2$ & $20$m\tabularnewline
\texttt{USairport} & $\num{1574}$ & $\num{17215}$ & $4$ & $1$h\tabularnewline
\end{tabular}
\end{table}

The sizes of the different networks are given in Table~\ref{tab:Size-of-real-world}.
We consider $\gamma_{k}=0$ is known and we assume a vague prior $\mbox{Gamma}(0.01,0.01)$
on the unknown parameters $\alpha$, $1-\sigma$, $\tau$, $a_{k}$
and $b_{k}$. We take $p=2$ communities for \texttt{polblogs} and $p=4$
communities for \texttt{USairport}. We run $3$ parallel MCMC chains,
each with $\num{10000}+\num{200000}$ iterations, using the same procedure
as used for the simulated data; see Section~\ref{sub:Simulated-data}.
Computation times are reported in Table~\ref{tab:Size-of-real-world}. The simulation of $w_{\ast1:p}$
  requires more computational time when $\sigma\geq 0$ (infinite-activity case). This explain the larger computation times for \texttt{USairport} compared to \texttt{polblogs}.

We interpret the communities based on the minimum Bayes risk point estimate
where the cost function is a permutation-invariant absolute loss on the weights $w=(w_{ik})_{i=1,\ldots,N_\alpha{;}k=1,\ldots,p}$. Let $\mathcal{S}_{p}$ be the
set of permutations of $\{1,\ldots,p\}$ and consider the cost
function
\[
C\left(w,w^{\star}\right)=\min_{\pi\in\mathcal{S}_{p}}\left[\sum_{k=1}^{p}\sum_{i=1}^{N_{\alpha}}\left|w_{i\pi(k)}-w_{ik}^{\star}\right|+\sum_{k=1}^{p}\left|w_{\ast\pi(k)}-w_{\ast k}^{\star}\right|\right]
\]
whose evaluation requires solving a combinatorial optimization problem
in $O\left(p^{3}\right)$ using the Hungarian method. We
therefore want to solve
\[
\widehat{w}=\underset{w^{\star}}{\arg\min}\,\mathbb{E}\left[C\left(w,w^{\star}\right)\vert Z\right]
\]
where $\mathbb{E}\left[C\left(w,w^{\star}\right)\vert Z\right]\simeq\frac{1}{N}\sum_{t=1}^{N}C\left(w^{(t)},w^{\star}\right)$
and $\left(w^{(t)}\right)_{t=1,\ldots,N}$ are from the MCMC output.
For simplicity, we limit the search of $\widehat{w}$ to the set of
MCMC samples giving
\[
\widehat{w}=\underset{w^{\star}\in\left\{w^{(1)},\ldots,w^{(N)}\right\}}{\arg\min}\frac{1}{N}\sum_{t=1}^{N}C\left(w^{(t)},w^{\star}\right).
\]

\begin{figure}[th]
\begin{centering}
\subfloat[\texttt{polblogs}]{\includegraphics[width=0.4\textwidth]{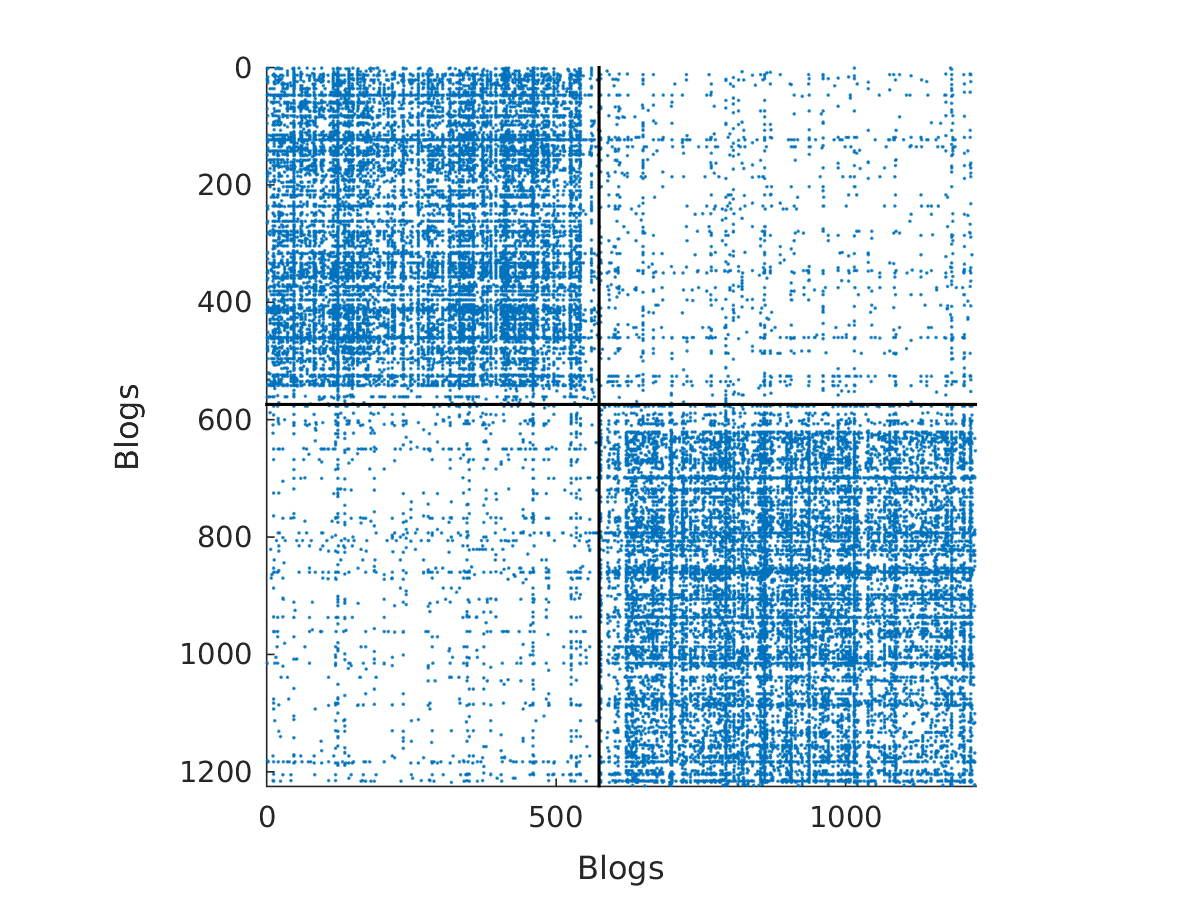}
}\subfloat[\texttt{USairport}]{\includegraphics[width=0.4\textwidth]{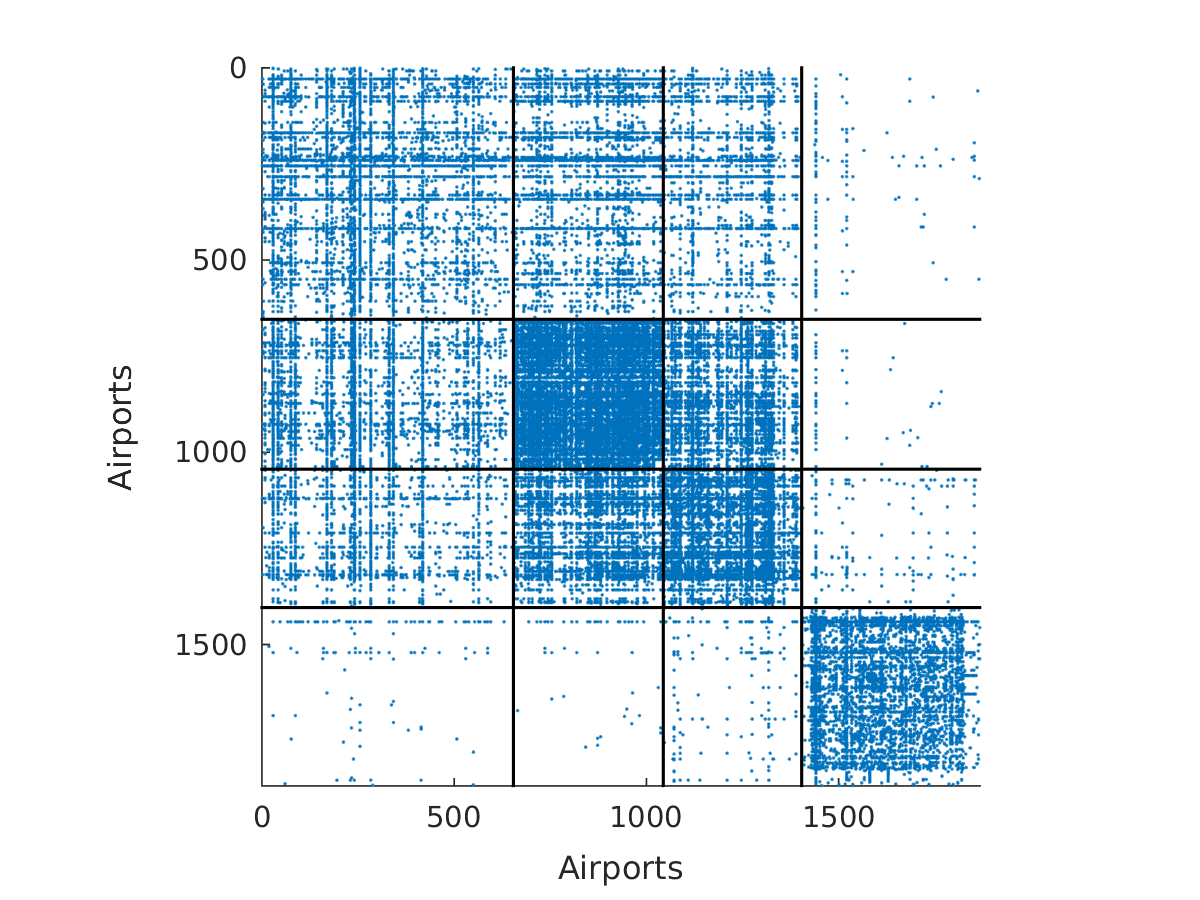}
}
\par\end{centering}
\caption{\label{fig:Block-structure}Adjacency matrices of the (a) \texttt{polblogs} and (b) \texttt{USairport} networks,
reordered by associating each node to the community where it has the highest weight.}
\end{figure}

\begin{table}[th]
\begin{centering}
\caption{Nodes with highest weight in each
community for the \texttt{polblogs} network. Blog URLs are follow{ed} by known political leaning: (L) for left-wing and (R) for right-wing.}
\label{tab:polblogs}
\begin{tabular}{>{\centering}p{0.3\textwidth}|>{\centering}p{0.3\textwidth}}
\textbf{Community 1: }``Liberal'' & \textbf{Community 2: }``Conservative''\tabularnewline
\hline
\textsf{\scriptsize{}\texttt{dailykos.com} (L)\\  \texttt{atrios.blogspot.com} (L)\\  \texttt{talkingpointsmemo.com} (L)\\  \texttt{washingtonmonthly.com} (L)\\  \texttt{liberaloasis.com} (L)\\  \texttt{talkleft.com} (L)\\  \texttt{digbysblog.blogspot.com} (L)\\  \texttt{newleftblogs.blogspot.com} (L)\\  \texttt{politicalstrategy.org} (L)\\ \texttt{juancole.com} (L)} & \textsf{\scriptsize{}\texttt{instapundit.com} (R)\\  \texttt{blogsforbush.com} (R)\\  \texttt{powerlineblog.com} (R)\\  \texttt{drudgereport.com} (R)\\  \texttt{littlegreenfootballs.com/weblog} (R)\\  \texttt{michellemalkin.com} (R)\\  \texttt{lashawnbarber.com} (R)\\  \texttt{wizbangblog.com} (R)\\  \texttt{hughhewitt.com} (R)\\ \texttt{truthlaidbear.com} (R)}%\tabularnewline
\end{tabular}
\end{centering}
\end{table}

\begin{table}
\begin{center}
\caption{Nodes with highest weights in each
community for the \texttt{USairport} network.}
\label{tab:airports}
\begin{tabular}{>{\centering}p{0.21\linewidth}|>{\centering}p{0.21\linewidth}|>{\centering}p{0.21\linewidth}|>{\centering}p{0.21\linewidth}}
\textbf{Community 1:}\textit{ }``Hub'' & \textbf{Community 2: }``East'' & \textbf{Community 3: }``West'' & \textbf{Community 4: }``Alaska''\tabularnewline
\hline
\textsf{\scriptsize{}Miami, FL\\ New York, NY\\ Newark, NJ\\ Los Angeles,
CA\\ Atlanta, GA\\ Washington, DC\\ Chicago, IL\\ Boston, MA\\ Houston,
TX\\ Orlando, FL} & \textsf{\scriptsize{}Cleveland, OH\\ Detroit, MI\\ Nashville, TN\\ Chicago,
IL\\ Knoxville, TN\\ Atlanta, GA\\ Louisville, KY\\ Indianapolis, IN\\
Memphis, TN\\ Charlotte, NC} & \textsf{\scriptsize{}Denver, CO\\ Las Vegas, NV\\ Los Angeles, CA\\ Salt
Lake City, UT\\Seattle, WA\\ Burbank, CA\\ Phoenix, AZ\\ Oakland, CA\\
Portland, OR\\ Albuquerque, NM} & \textsf{\scriptsize{}Anchorage, AK\\ Fairbanks, AK\\ Bethel, AK\\ St.
Mary's, AK\\ King Salmon, AK\\ McGrath, AK\\ Unalakleet, AK\\ Galena,
AK\\ Aniak, AK\\ Kotzebue, AK}\tabularnewline
\end{tabular}
\end{center}
%\end{adjustbox}
\end{table}

\begin{figure}[th]
\begin{centering}
\includegraphics[width=0.47\textwidth]{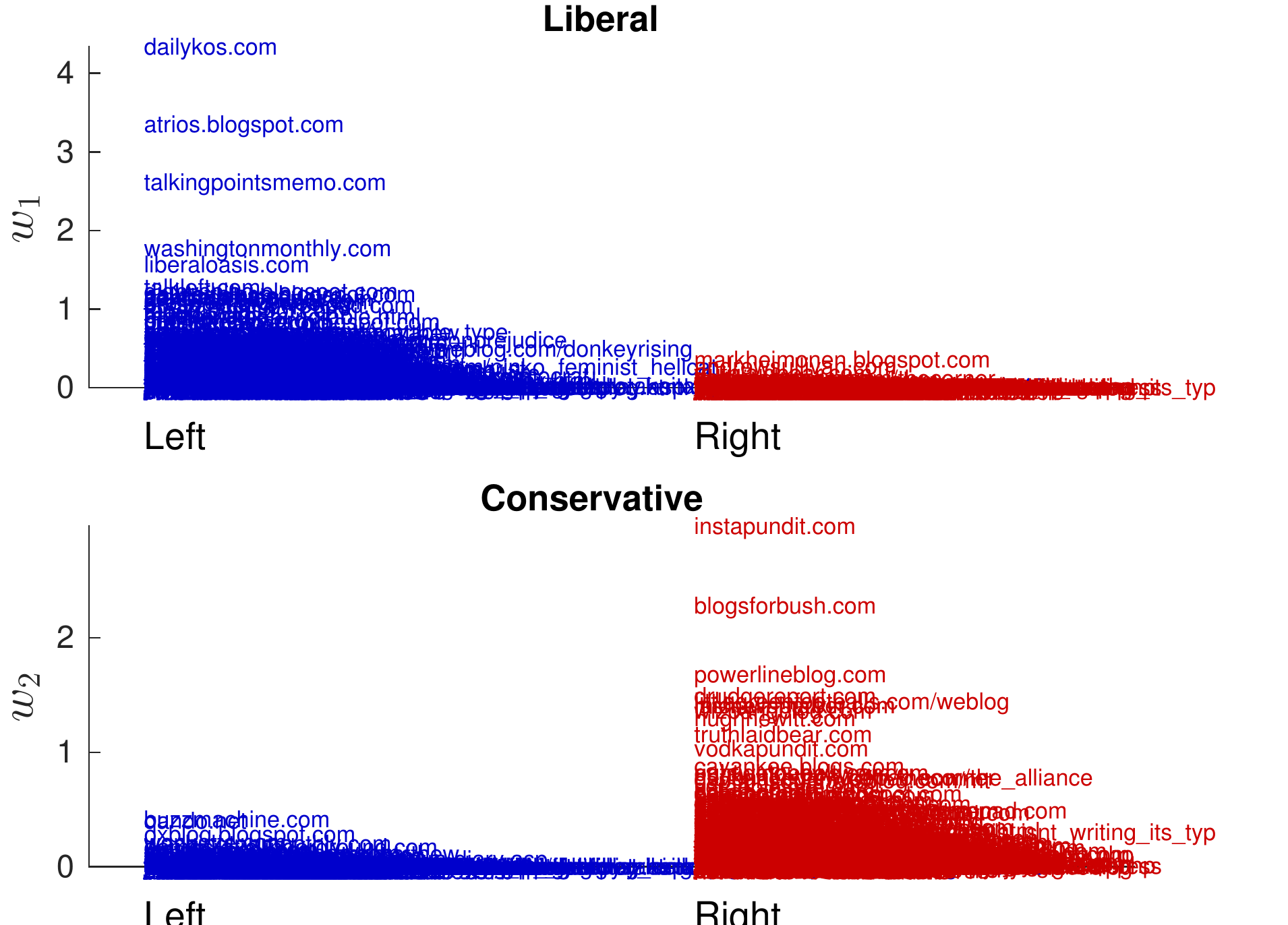}
\end{centering}

\caption{\label{fig:Feature-weights-of} Level of affiliation of each blog of the \texttt{polblogs}
network to the communities identified as ``Liberal'' (top) and ``Conservative'' (bottom). The names of the blogs are grouped according to the left-right wing ground truth. Left-wing blogs are represented in blue on the left, right-wing blogs in {red} on the right.}
\end{figure}

\begin{figure}[th]
\subfigure[\texttt{polblogs}]{\includegraphics[width=0.5\textwidth]{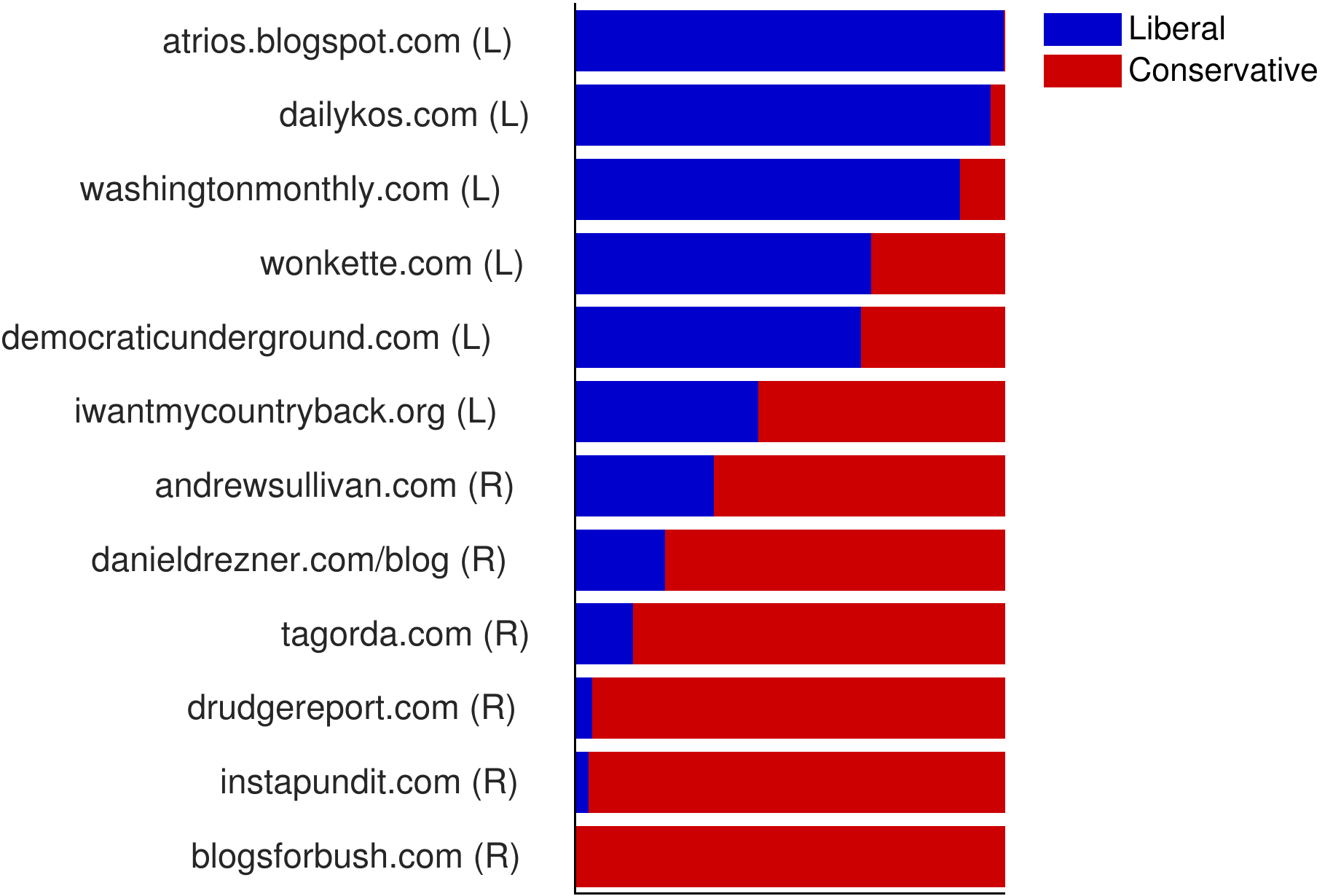}

}\subfigure[\texttt{USairport}]{\includegraphics[width=0.5\textwidth]{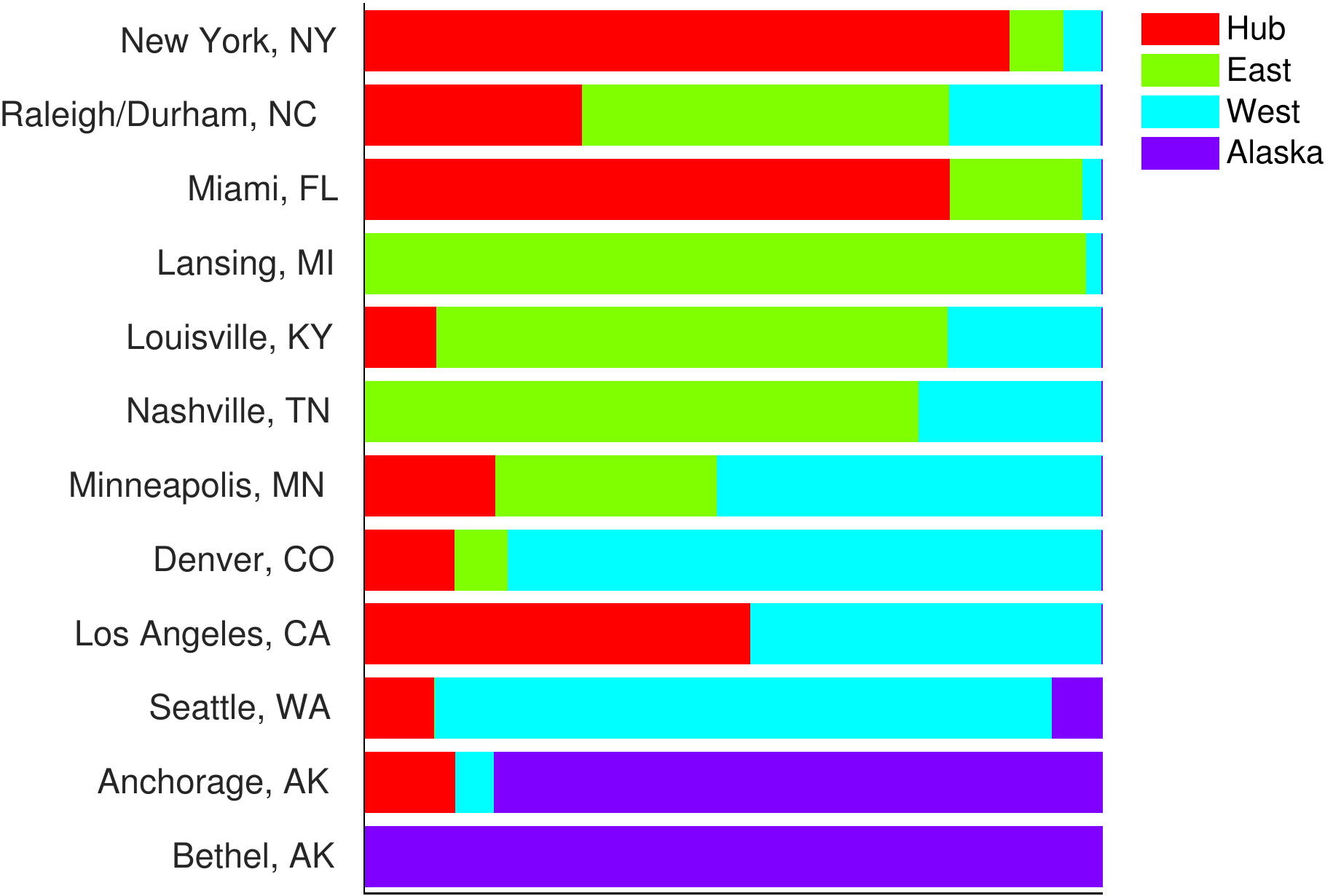}

}

\caption{\label{fig:Proportions-of-features}Relative values of the weights in each community for a subset of the nodes of the (a) \texttt{polblogs} and (b) \texttt{USairport} networks.}
\end{figure}

Table~\ref{tab:polblogs} reports the nodes with
highest weights in each community for the \texttt{polblogs} network. Figure~\ref{fig:Feature-weights-of} also shows the weight associated to each of the two community
 alongside the true left/right class for each blog.  The two learned communities, which can be interpreted as ``Liberal'' and ``Conservative'', clearly recover the political leaning of the blogs.
Figure~\ref{fig:Block-structure} shows the adjacency matrices
obtained by reordering the nodes by community membership, where each
node is assigned to the community whose weight is maximum, clearly showing the block-structure of this network. The obtained clustering yields a $93.95$\% accuracy when compared to the ground truth classification.
Figure~\ref{fig:Proportions-of-features}(a) shows the relative community proportions for a subset of the blogs. \texttt{dailykos.com} and \texttt{washingtonmonthly.com} are clearly
described as liberal while \texttt{blogsforbush.com}, \texttt{instapundit.com}
and \texttt{drudgereport.com} are clearly conservative.
Other more moderate blogs such as \texttt{danieldrezner.com/blog}
and \texttt{andrewsullivan.com} have more balanced values in both communities. Figure~\ref{fig:Empirical-degree-distribution}(a) shows that the posterior predictive degree distribution provides a good fit to the data.

For \texttt{USairport}, the four learned communities can also be easily interpreted, as seen in Table~\ref{tab:airports}. The first community, labeled ``Hub'', represents highly
connected airports with no preferred location, while the three others, labeled ``East'', ``West'' and ``Alaska'', are communities based on the location of the airport.
In Figure~\ref{fig:Proportions-of-features}(b), we can see that some airports
have a strong level of affiliation in a single community: New York and Miami for ``Hub'',
Lansing for ``East'', Seattle for ``West'' and Bethel and Anchorage
for ``Alaska''. Other airports have significant weights in different
communities: Raleigh/Durham and Los Angeles are hubs with strong regional
connections, Nashville and Minneapolis share a significant number of connections
with both East and West of the USA. Anchorage has a significant ``Hub''
weight, while most airports in Alaska are disconnected from the rest
of the world as can be seen in Figure~\ref{fig:Block-structure}(b). ``Alaska'' appears as a separate block while substantial overlaps are observed between the ``Hub'', ``East'' and ``West'' communities. Figure~\ref{fig:Empirical-degree-distribution}(b) shows that the posterior predictive degree distribution also provides a good fit to the data.

\subsection{Comparisons}

%%%degree predective deistribution
\begin{figure}[!th]
\begin{centering}
\subfloat[\texttt{polblogs}]{\includegraphics[width=0.4\textwidth]{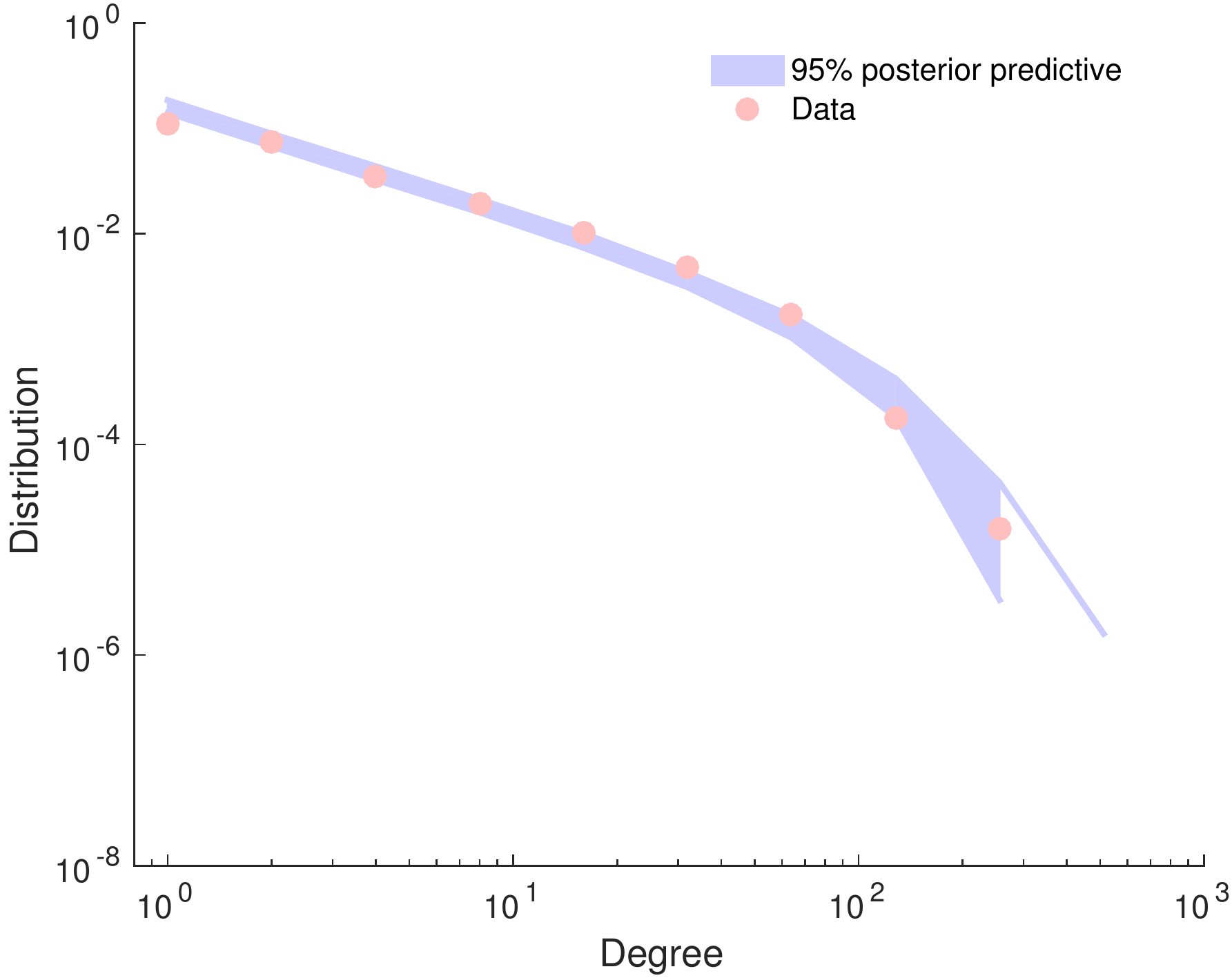}
} \hspace*{.1\textwidth}\subfloat[\texttt{USairport}]{\includegraphics[width=0.4\textwidth]{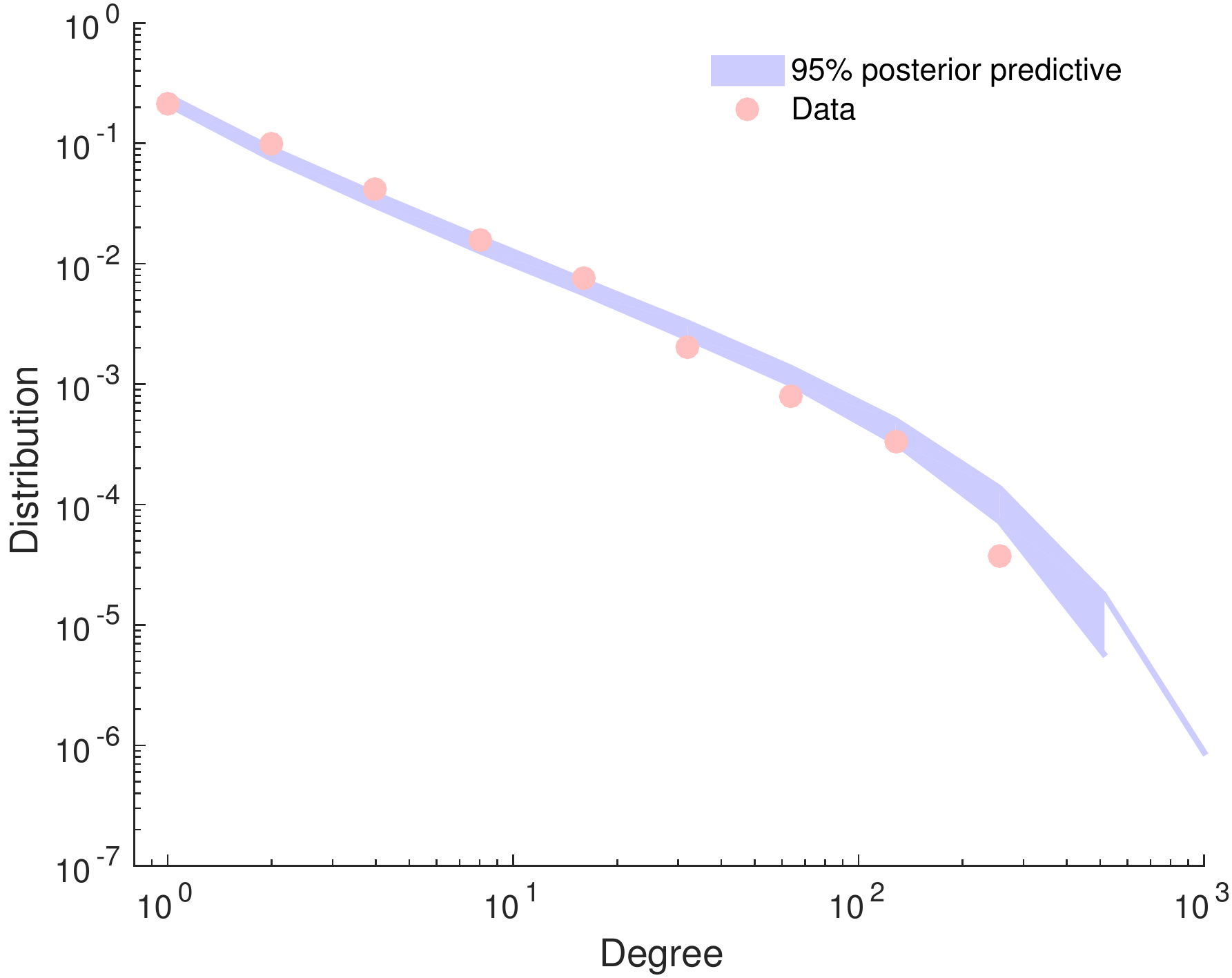}
}
\par\end{centering}
\begin{centering}
\subfloat[\texttt{polblogs}]{\includegraphics[width=0.4\textwidth]{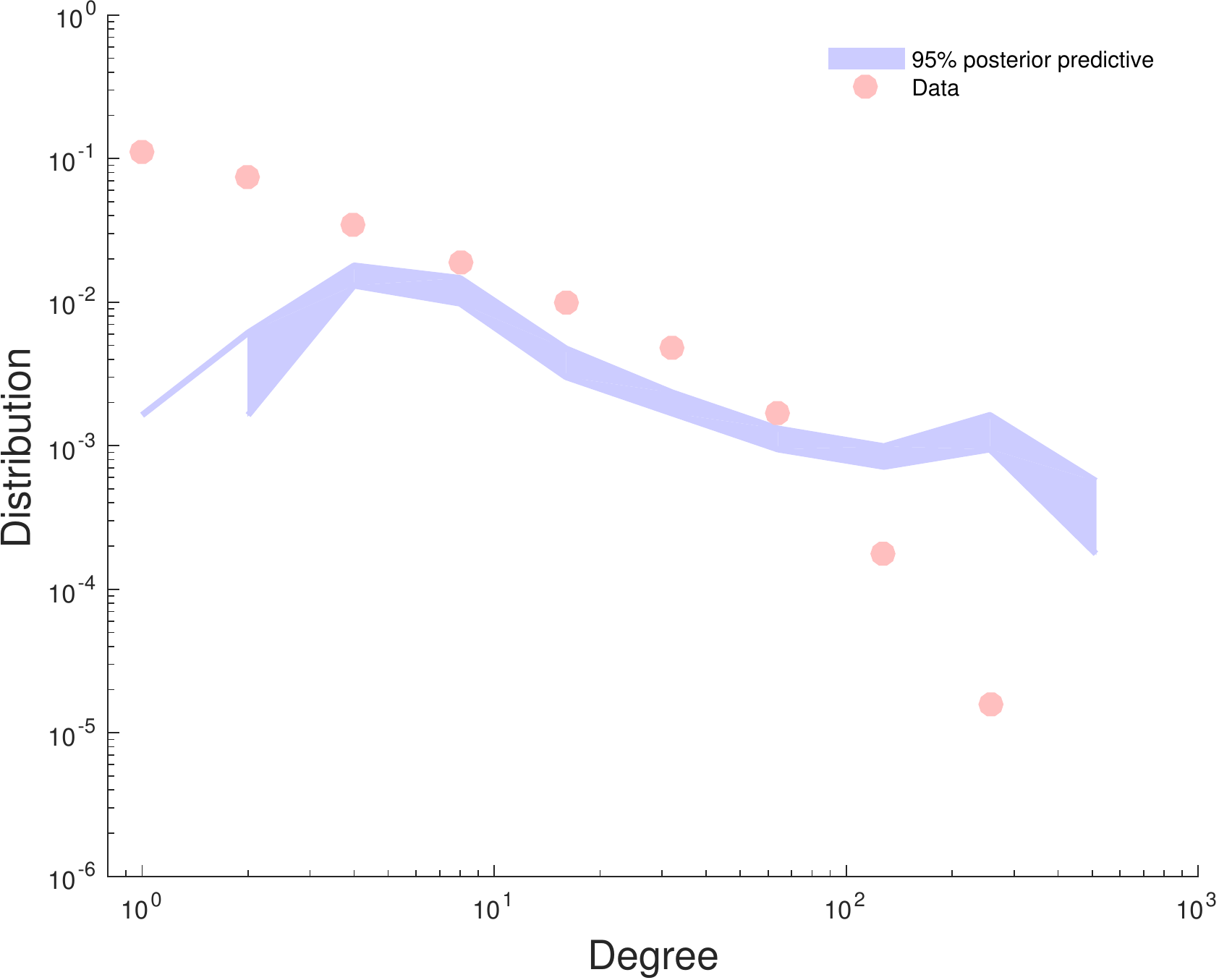}
}\hspace*{.1\textwidth}\subfloat[\texttt{USairport}]{\includegraphics[width=0.4\textwidth]{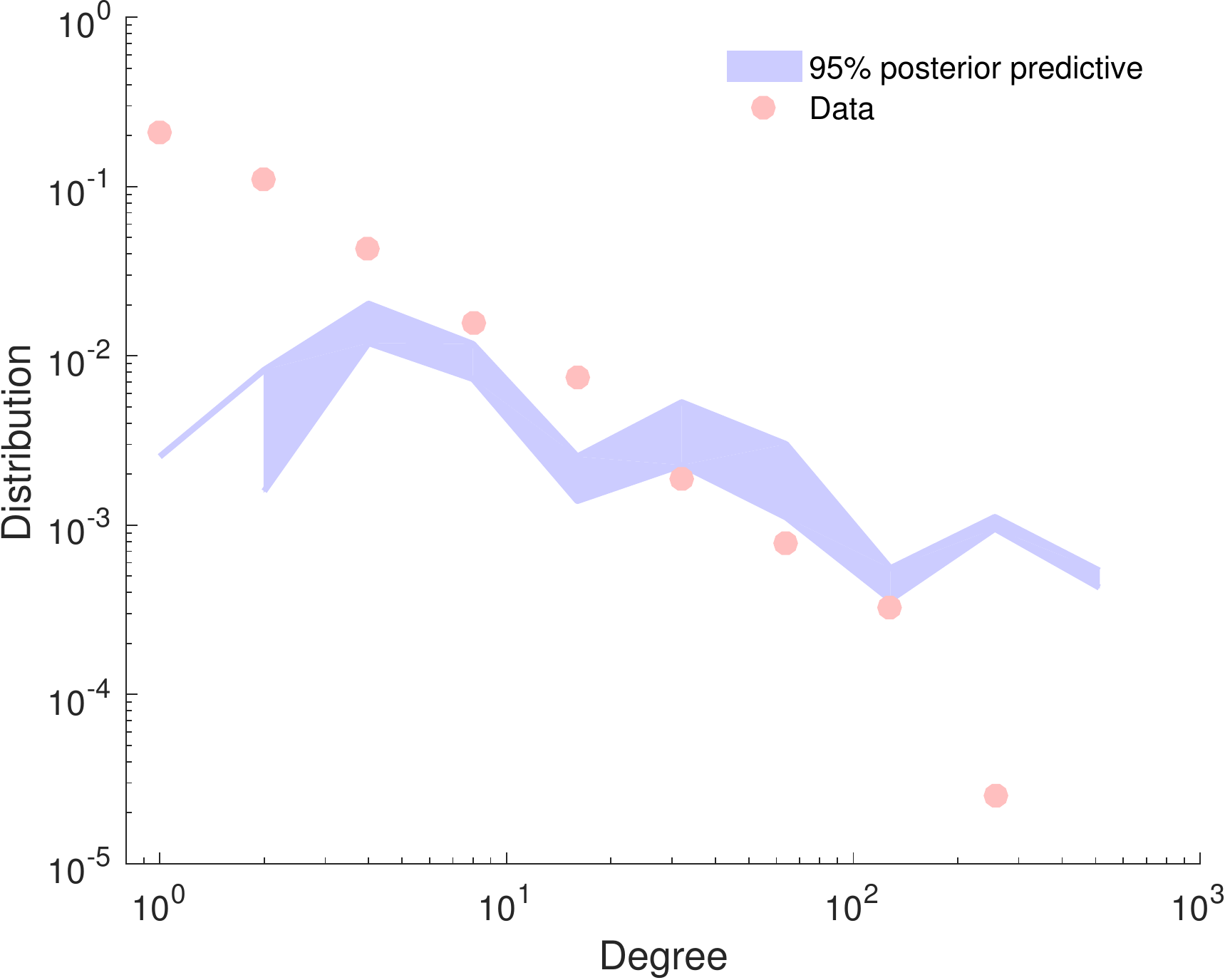}
}
\par\end{centering}
\begin{centering}
\subfloat[\texttt{polblogs}]{\includegraphics[width=0.4\textwidth]{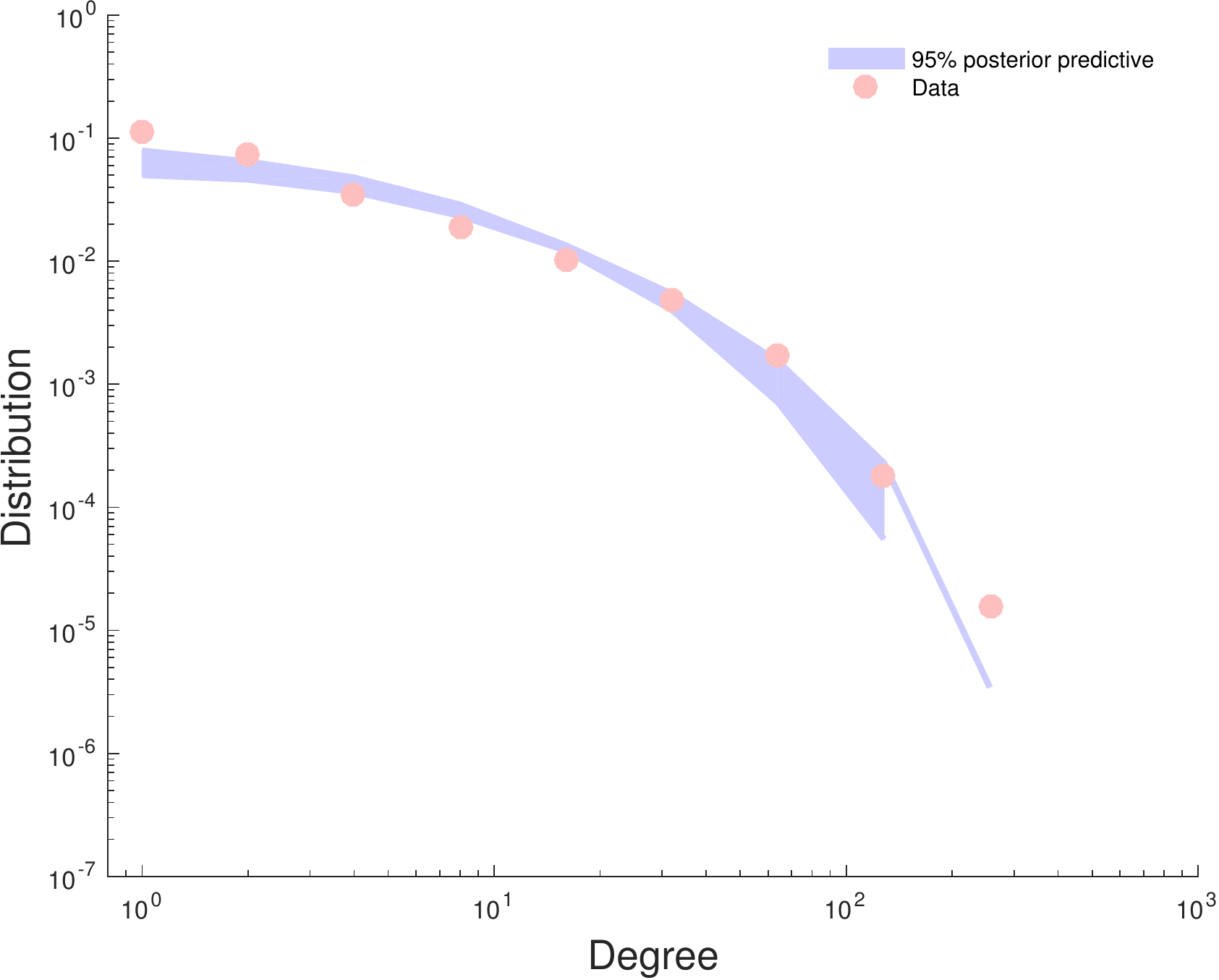}
}\hspace*{.1\textwidth}\subfloat[\texttt{USairport}]{\includegraphics[width=0.4\textwidth]{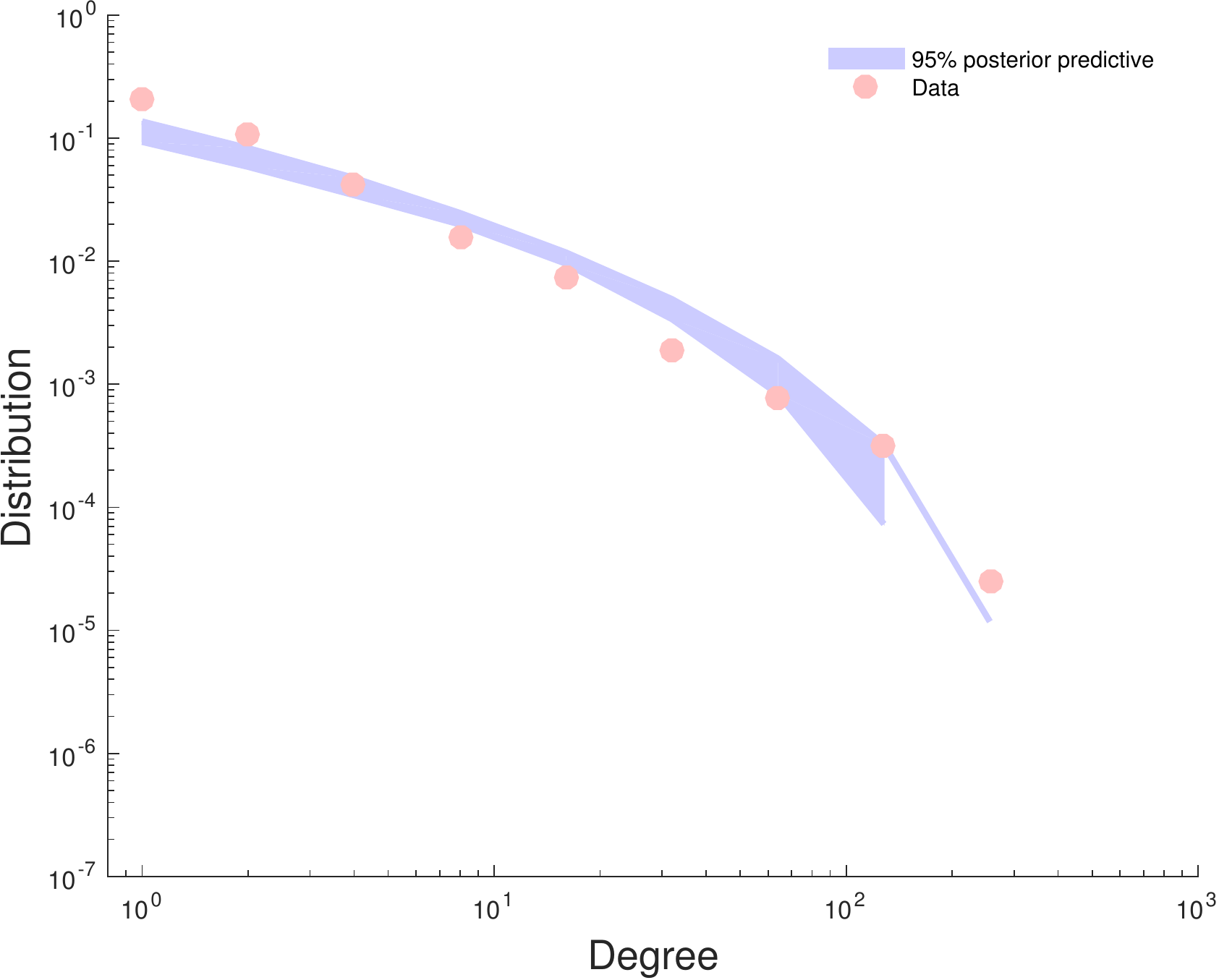}
}
\par\end{centering}
\caption{\label{fig:Empirical-degree-distribution}Empirical degree distribution
(red) and posterior predictive (blue) of the (left) \texttt{polblogs} and (right) \texttt{USairport} networks under our (top row) CCRM model, (middle row) MMSB and the (bottom row) MLFM.}
\end{figure}

%%% cluster coefficient and sd degree
\begin{figure}[!th]
%\begin{centering}
%\subfloat[\texttt{polblogs}]{\includegraphics[width=0.4\textwidth]{figures/AMEN/polblogs_2f_degree_ONE_nolegend}}\hspace*{.1\textwidth}
%\subfloat[\texttt{USairport}]{\includegraphics[width=0.4\textwidth]{figures/AMEN/USAIRPORT_4f_degree_ONE_nolegend}}
%\par\end{centering}
\begin{centering}
\subfloat[\texttt{polblogs}]{\includegraphics[width=0.4\textwidth]{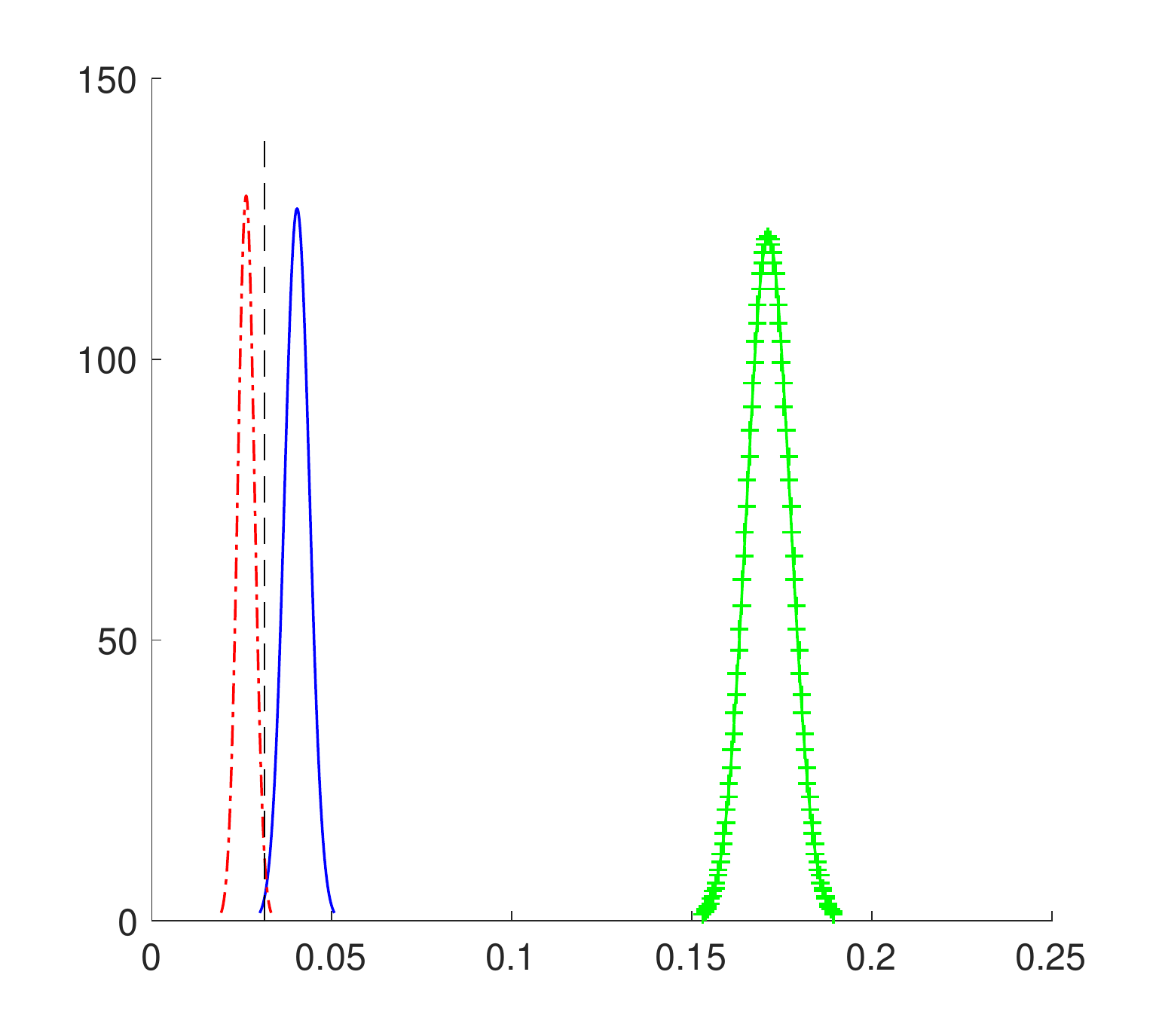}}\hspace*{.1\textwidth}
\subfloat[\texttt{USairport}]{\includegraphics[width=0.4\textwidth]{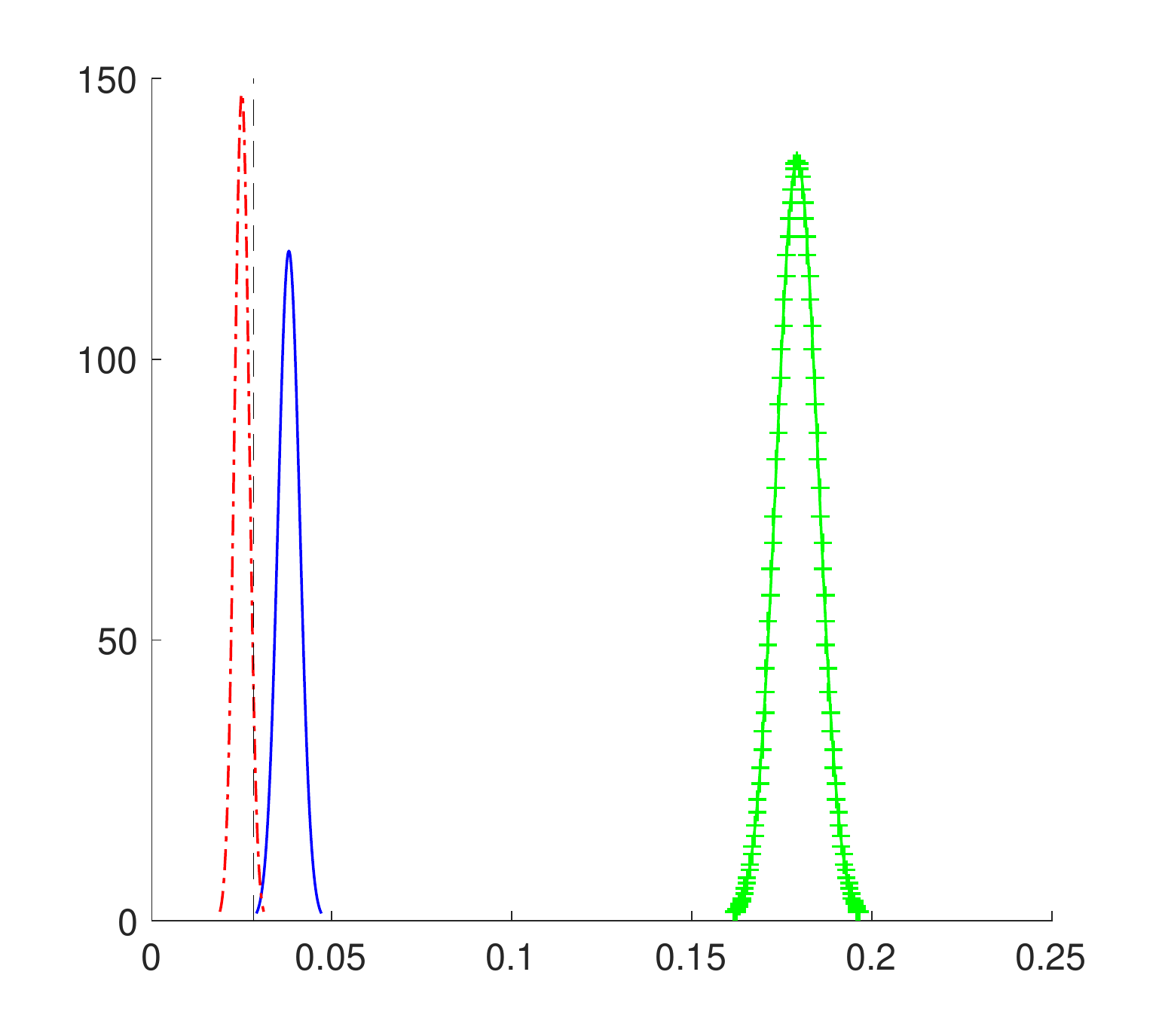}}
\par\end{centering}
\begin{centering}
\subfloat[\texttt{polblogs}]{\includegraphics[width=0.4\textwidth]{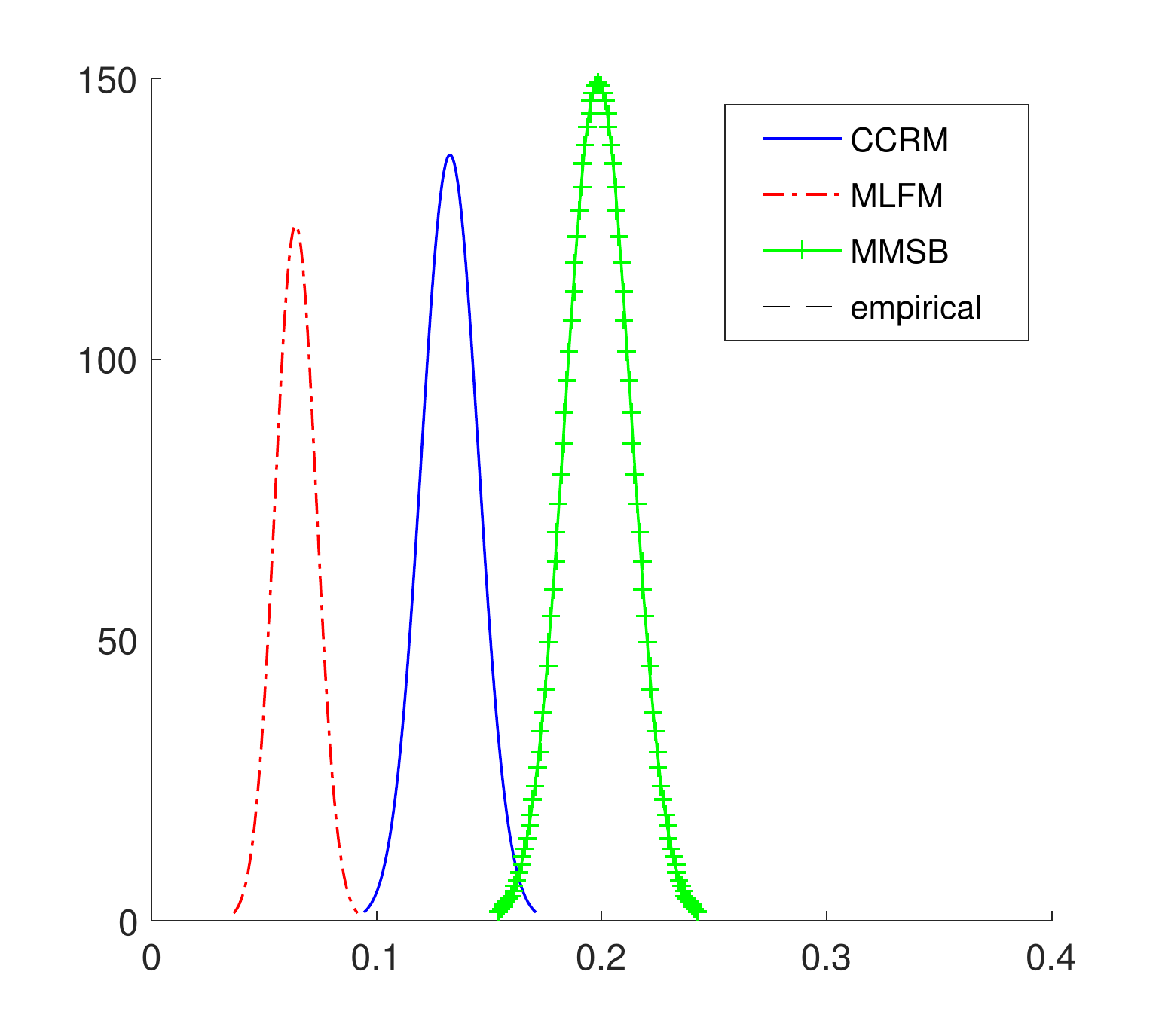}}\hspace*{.1\textwidth}
\subfloat[\texttt{USairport}]{\includegraphics[width=0.4\textwidth]{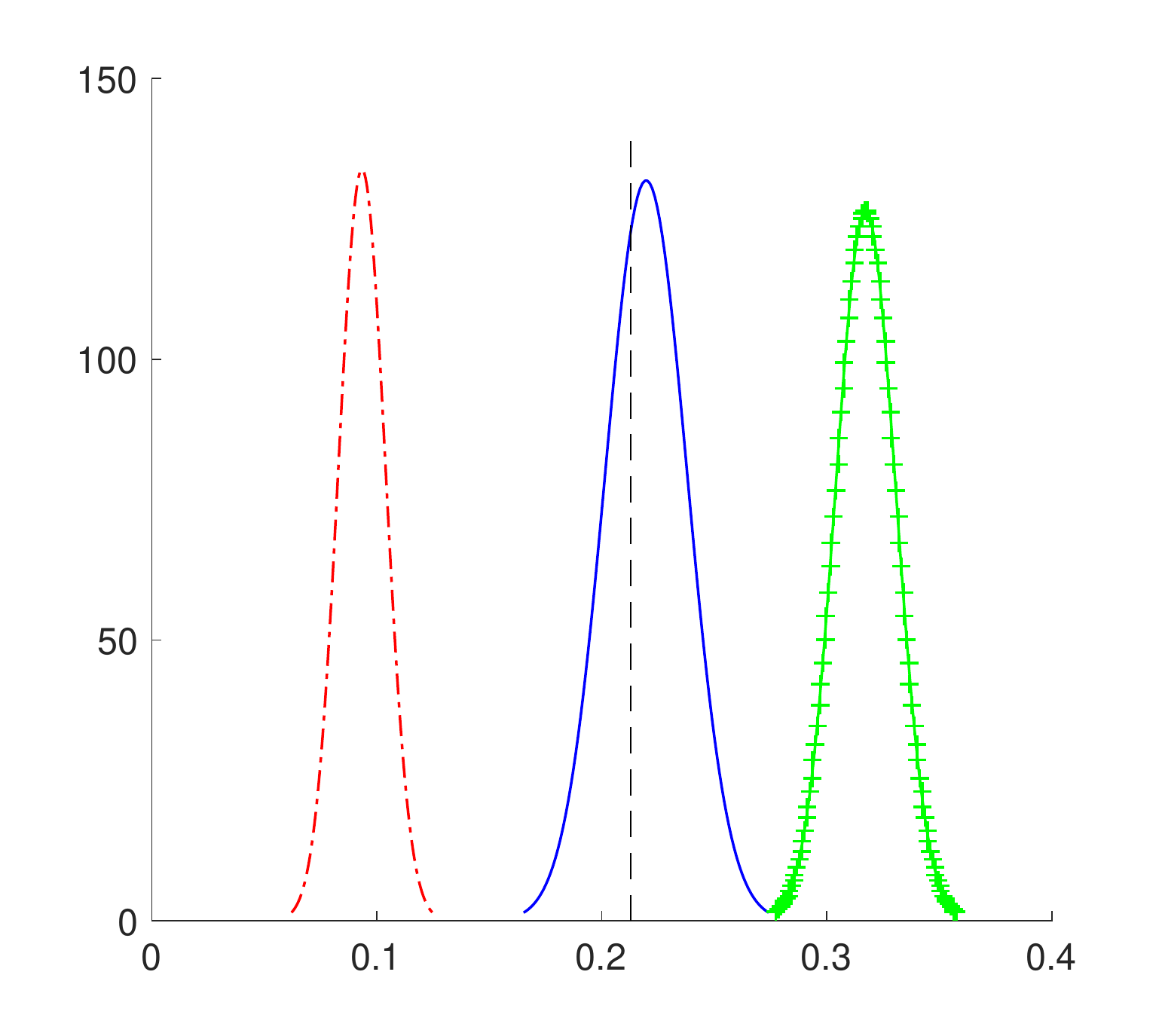}}
\par\end{centering}
\caption{\label{fig:Diagnostics}Posterior predictive plots for the \texttt{polblogs} and \texttt{USairport} networks under the CCRM, MLFM and MMSB models. Top row: standard deviation of degree; bottom row: cluster coefficient}
\end{figure}

We compare the fit of our model to the mixed membership stochastic blockmodel (MMSB) of \citet{Airoldi2008} and the multiplicative latent factor model (MLFM) of \citet{Hoff2009}. The two models are briefly explained below.
\subsection*{Multiplicative latent factor model} Let $X \in \mathbb{R}^{N \times N}$ denote a (symmetric) random matrix of effects for a set of $N$ nodes. Under this model, $X$ is explained as the sum of systematic patterns and random noise.  $Y$ denotes the adjacency matrix. For $1\leq
i<j\leq N$
\begin{align*}
X_{ij}  & = \xi_0 + \zeta_i + \zeta_j + M_{ij}+E_{ij}\\
Y_{ij} &=\1{ X_{ij}  > 0 }
\end{align*}
where $\xi_{0}\in\mathbb{R}$ is the intercept, $\zeta_i \in \mathbb{R}$ is the additive node effect, $M \in \mathbb{R}^{N \times N} $ is the matrix of multiplicative effects. The square symmetric matrix $M$ has a latent decomposition of the form $M = U\Lambda U^T,$ where  $U \in \mathbb{R}^{N\times p}$ and $\Lambda \in \mathbb{R}^{p \times p}$ is a diagonal matrix. $E \in \mathbb{R}^{N\times N}$ is the matrix of standard normal noise; $E_{ij} \overset{\text{iid}}{\sim}\mathcal{N}(0,1) $. Denoting by $u_i$ the $i^{\text{th}}$ row of $U$  we have $M_{ij} = u^T_{i}\Lambda u_{j}$.
We use the priors implemented in the package \textsf{amen}. For the additive effect $\zeta_i$ it is assumed that $  \,\zeta_{i} \overset{\text{i.i.d.}}{\sim}\mathcal{N}(0,\sigma_{\zeta}^{2})$, with $1/\sigma_{\zeta}^{2}   \sim \text{Gamma}\left(  \frac{1}{2},\frac{1}{2}\right) $. For the multiplicative effects it is assumed that $\,u_{ij} \overset{\text{ind}}{\sim} \mathcal{N}(0,\sigma_{j}^{2}) $ with $1/\sigma_{j}^{2}\sim \text{Gamma}(2,1)$.
For $\xi_{0}$ we use an improper prior $p(\xi_0) \propto 1 /\xi_0$.

\subsection*{Mixed membership stochastic blockmodel}
Let $N$ be the number of nodes in the network and $Y$ the adjacency matrix. For each node $i=1,\ldots,N$, let $$\pi_i\sim \text{Dirichlet}(\varsigma,\ldots,\varsigma)$$ be a $p$ dimensional mixed membership probability vector for node $i$, with $\varsigma>0$.
For each pair of nodes $i<j$, let
\begin{align*}
c_{i j}|\pi_i& \sim \pi_i\\
c_{j i}|\pi_j& \sim \pi_j\\
Y_{ij}|c_{ij},c_{ji},B& \sim \text{Bernoulli}\left\{ \left(1-\rho\right)B_{c_{ij}c_{ji}}\right\}
\end{align*}
where $c_{ij}\in\{1,\ldots,p\}$ is the emission indicator variable, $\rho$ is a parameter that controls the proportion of zeros that should not be explained by the blockmodel and  $B$  is the $p \times p$ matrix that contains the Bernoulli rates  of the link probabilities between different communities, i.e. $B_{k,l}$ is the probability of a connection between a member of group $k$ and one of group $l$. We assume $\varsigma \sim\text{Gamma}(1,1)$, $\rho\sim\text{Beta}(1/2,1/2)$ and $B_{k,\ell}\sim\text{Beta}(1,1)$.\bigskip

For each model, we run three MCMC chains for posterior inference, using the \textsf{amen R} package~\citep{Hoff2017} for MLFM. Under MMSB, we use $200000$ MCMC iterations,  of which $100000$ discarded as burn in. Under MLFM, which required more iterations to converge, we ran $1M$ iterations, of which $500 000$ discarded as burn-in. In both cases we thinned the output to obtain 500 samples approximately distributed from the posterior distribution.

\paragraph{Posterior predictive checks.} In order to evaluate the goodness of fit of each model, we look at some statistics of a replicated network for a \textbf{new} set of nodes, sampled under the posterior predictive. Let $Y^*$ be a $N$ by $N$ adjacency matrix corresponding to the edges of a set of $N$ new nodes, sampled from
\begin{align*}
\Pr(Y^*|Y)= \left \{
\begin{array}{ll}
  \int \Pr(Y^*|\xi_0,\Lambda)\Pr(d\xi_0,d\Lambda|Y) & \text{ for MLFM} \\
  \int \Pr(Y^*|\varsigma,\rho,B)\Pr(d\varsigma,d\rho,dB|Y) & \text{ for MMSB}
\end{array}\right .
\end{align*}
Note that the posterior predictive setting is different from that of \cite{Hoff2009}, which replicates a network for the same set of nodes and therefore conditions on the parameters $u_i,\zeta_i$.
We generate 500 samples from the posterior predictive under each model. We are interested in some standard summary statistics: the degree distribution, standard deviation of the degrees and cluster coefficient. Posterior degree distributions under each model for \texttt{polblogs} and \texttt{USairport} are presented in Figure~\ref{fig:Empirical-degree-distribution} and posterior predictive distributions of the standard deviation of the degree and clustering coefficient in Figure~\ref{fig:Diagnostics}.

Overall, the MMSB did not perform well on the two datasets considered. While being a very flexible model, applied successfully to a wide range of real-world networks, the MMSB doesn't explicitly capture degree heterogeneity; the latent communities recovered (not shown) do not correspond to those recovered by CCRM or MLFM: it tends to cluster nodes according to their degree, explaining the shape of the posterior degree predictive in Figure~\ref{fig:Empirical-degree-distribution}(c-d) . Such limitations, due to the lack of degree correction, have been acknowledged by previous authors~\citep{Karrer2011,Gopalan2013a}. Note that degree-corrected MMSB have been proposed by \cite{Gopalan2013a}, similar to the general class of models discussed in \citep{Hoff2009}.

MLFM on the other hand incorporates degree heterogeneity and thus gives better fit on the degree predictive distribution. It also gives similar results than CCRM and recovers similar latent communities. However, by construction this model cannot capture sparsity or heavy tailed degree distributions. It underestimates the proportion of nodes with degree one as shown in Figure~\ref{fig:Empirical-degree-distribution} (top row). While our CCRM model over-estimates the proportion of nodes of degree one, it tends to give a better fit to the empirical degree distribution overall. Finally, in Figure~\ref{fig:Diagnostics}, we report the goodness of fit statistics used in the analyses of \citet{Hoff2009}. The first one is the standard deviation of the degree shown in the middle row, on which MMSB performs poorly, whereas CCRM and MLFM have similar results, marginally close to the empirical value. The bottom row reports the cluster coefficient, also known as triadic dependence. In this case, CCRM gives a better fit on the sparse \texttt{USairport} dataset while MFLM gives a  better result than CCRM on the dense \texttt{polblogs} dataset.

\section*{Acknowledgments}
The authors thank George Deligiannidis for pointing out the article of \cite{Asmussen2001}.
FC acknowledges the support of the European Commission under the Marie Curie Intra-European
Fellowship Programme. Part of this work has been supported by the BNPSI ANR project no ANR-13-BS-03-0006-01.

\appendix

\section{Background on completely random measures}
\label{sec:background}

\subsection{Completely random measures}

Completely random measures (CRM) were introduced by
\cite{Kingman1967,Kingman1993} and are now standard tools for
constructing flexible Bayesian nonparametric (BNP) models; see for example the surveys of \cite{Lijoi2010} or~\citet[Section~10.1]{Daley2008a}.

A CRM $W$ on $\mathbb R_+$ is a random measure
such that, for any collection of disjoint measurable subsets $A_{1}%
,\ldots,A_{n}$ of $\mathbb R_+$, $W(A_{1}),\ldots,W(A_{n})$ are independent. A CRM
can be decomposed into a sum of three independent parts: a non-random measure,
a countable collection of atoms with {random masses at} fixed locations, and a countable
collection of atoms with random masses and random locations. Here, we will
only consider CRMs with random masses and random locations, which take the
form%
\begin{equation}
W=\sum_{i=1}^{\infty}w_{i}\delta_{\theta_{i}}%
\end{equation}
where the $w_i\in \mathbb R_+$ are the masses and $\theta_{i}\in\mathbb R_+$ are the
locations. The law of $W$ can actually be characterized by a Poisson point process $N=\{(w_{i},\theta_{i})_{i=1,2,\ldots}\}$ on $\mathbb R_+^2$  with mean measure $\nu(dw,d\theta)$. We focus here on the case where the CRM is homogeneous with independent increments. This implies that the location $\theta_i$ are independent of the weights $w_i$ and the mean measure decomposes as $\nu(dw,d\theta)=\rho(dw)\lambda(d\theta)$ where $\lambda$ is the Lebesgue measure and $\rho$ is a measure on $\mathbb R_+$ such that
\begin{equation}
\int_{0}^{\infty}(1-e^{-w})\rho(dw)<\infty.\label{eq:conditionLevy}
\end{equation}
We write $W\sim\CRM(\rho,\lambda)$. Note that $W([0,T])<\infty$ a.s. for any real $T$ while $W(\mathbb R_+)=\infty$ a.s. if $\rho$ is not degenerate at $0$. If
\begin{equation}
\int_{0}^{\infty}\rho(dw)=\infty%\label{eq:condinfiniteCRM}
\end{equation}
then there will be a.s. an  infinite number of jumps in any interval $[0,T]$ and we refer to the CRM as infinite-activity. Otherwise, it is called finite activity. Let $\overline{\rho}$ be the tail L\'evy intensity defined as
\begin{equation}
\overline{\rho}(x)=\int_x^\infty \rho(dw)\label{eq:taillevy2}
\end{equation}
for $x>0$. This function corresponds to the expected number of points $(w_i,\theta_i)$ such that $w_i>x$ and $\theta_i\in[0,1]$, and its asymptotic properties play an important role in the characterization of the graph properties.

\subsection{Vectors of CRMs}

Multivariate extensions of CRMs have been proposed recently by various authors~\citep{Epifani2010,Leisen2011,Leisen2013,Griffin2013,Lijoi2014}. These models are closely related to L\'evy copulas~\citep{Tankov2003,Cont2003,Kallsen2006} and multivariate subordinators on cones~\citep{Barndorff-Nielsen2001,Skorohod1991}.  A vector of CRMs $(W_1,\ldots,W_p)$ on $\mathbb R_+$ is a collection of random measures $W_k$, $k=1,\ldots,p$, such that, for any collection of disjoint measurable subsets $A_1,\ldots,A_n$ of $\mathbb R_+$, the vectors  $(W_1(A_1),\ldots,W_p(A_1))$, $(W_1(A_2),\ldots,W_p(A_2))$,\ldots,$(W_1(A_n),\ldots,W_p(A_n))$ are mutually independent. We only consider here vectors of CRMs with both random weights and locations. In this case, the measures
$W_{k}$, $k=1,\ldots,p$, are a.s. discrete  and take the form

\begin{equation}
W_{k}=\sum_{i=1}^{\infty}w_{ik}\delta_{\theta_{i}}.%
\end{equation}
The law of the vector of CRMs can be characterized by a Poisson point process on $\mathbb R^{p+1}_+$ with mean measure $\nu(dw_1,\ldots,dw_p,d\theta)$. We focus again on homogeneous vectors of CRMs with independent increments where the mean measure can be written as
\begin{equation}
\nu(dw_{1},\ldots,dw_{p},d\theta)=\rho(dw_{1},\ldots,dw_{p})\lambda(d\theta).%\label{eq:nu}
\end{equation}
where $\rho$ is a measure on $\mathbb R^p_+$, concentrated on $\mathbb R^p_+\backslash \{\mathbf 0 \}$,  which satisfies
\begin{equation}
\int_{\mathbb R_+^p} \min\left(1,{\sum_{k=1}^{p} w_{k}}\right)\rho(dw_1,\ldots,dw_p)<\infty.
%\label{eq:conditionLevydef}
\end{equation}
We use the same notation as for (scalar) CRMs and write simply $(W_1,\ldots,W_p)\sim \text{CRM}(\rho,\lambda)$. A key quantity is the multivariate Laplace exponent defined by
\begin{align}
\psi(t_{1},\ldots,t_{p})  &  :=-\log\mathbb{E}\left[  e^{-\sum_{k=1}^{p}%
t_{k}W_{k}([0,1])}\right] \\
&  =\int_{{\mathbb R^p_+}}\left(  1-e^{-\sum_{k=1}^{p}t_{k}w_{k}}\right)
\rho(dw_{1},\ldots,dw_{p})
\end{align}
Note that this quantity involves a $p$-dimensional integral which may not be analytically computable, and may be expensive to evaluate numerically. As for CRMs, if
\begin{equation}
\int_{\mathbb R^p_+} \rho(dw_1,\ldots,dw_p)=\infty
\end{equation}
then there will be an infinite number of $\theta_i\in[0,T]$ for which $\sum_k w_{ik}>0$ and the vector of CRMs is called infinite-activity. Otherwise, it is called finite-activity. Note that some (but not all) CRMs may still be marginally finite-activity.

\section{Proof of Propositions  \ref{th:sparsitygeneral} and  \ref{th:sparsityCCRM}}
\label{sec:proofs}

The proof of~\cite[Appendix C]{Caron2014} can be directly adapted to the multivariate generalization presented in this paper. We only provide a sketch of the proof. First, as $Z$ is a jointly exchangeable point process verifying \eqref{eq:kallenbergexchangeability} and under the moment condition~\eqref{eq:momentcondition}, it follows from the law of large numbers that \begin{equation*}
N_\alpha^{(e)}=\Theta(\alpha^2) \text{ a.s. as }\alpha\rightarrow\infty.
\end{equation*}
\paragraph{Finite-activity case.} If the vector of CRMs is finite-activity, the jump locations arise from an homogeneous Poisson process with finite rate, and $N_\alpha =\Theta(\alpha)$ a.s. It follows that
\begin{equation*}
N_\alpha^{(e)}=\Theta(N^2_\alpha) \text{ a.s. as }\alpha\rightarrow\infty.
\end{equation*}
\paragraph{Infinite-activity case.} Consider now the infinite-activity case. Following \cite{Caron2014}, one can lower bound the node counting process $N_\alpha$ by a counting process $\widetilde N_\alpha$ which is conditionally Poisson, and the same proof applies. For infinite-activity CCRM, we use the fact that $\psi(W_1([0,\alpha]),\ldots,W_p([0,\alpha]))\rightarrow \infty$ a.s., it follows that $N_\alpha =\Omega(\alpha)$ a.s., and therefore
\begin{equation*}
N_\alpha^{(e)}=o(N^2_\alpha) \text{ a.s. as }\alpha\rightarrow\infty.
\end{equation*}

Finally, for compound CRMs with regularly varying $\rho_0$ with exponent $\sigma$ and slowly varying function such that $\lim_{t\rightarrow \infty}\ell(t)>0$, Proposition \ref{th:regularvar} in Appendix~\ref{sec:app:multivariateregularvariation} implies that $N_\alpha =\omega(\alpha^{1+\sigma})$ a.s. and \[
N_\alpha^{(e)} =  O(N_\alpha^{2/(1+\sigma)}) \text{ a.s. as }\alpha\rightarrow\infty.
\]

\section{Proof of Theorem \ref{sec::expected_num}}
\label{sec:proofs2}

Let $D^*_\alpha=\sum_{k=1}^p D_{k\alpha}([0,\alpha])$ be the number of edges in the directed graph of size $\alpha$, $W_{k,\alpha}^*=W_k([0,\alpha])$ and $W_{\alpha}^*=(W_{1,\alpha}^*,\ldots,W_{p,\alpha}^*)^T$. Using Campbell's theorem,
\begin{equation}
\begin{split}
 \mathbb{E}\left[D^*_\alpha\right]
& =\mathbb{E}\left[\mathbb{E}\left[D^*_{\alpha} | W^*_{\alpha}\right]\right]= \mathbb{E}\left[ (W^*_{\alpha})^TW^*_{\alpha}\right] \\
& =  \mathbb{E}\left[W^*_{\alpha}\right]^T\mathbb{E}\left[W^*_{\alpha}\right] + \text{tr}(\text{cov}\left(W^*_{\alpha}\right))\\
& = \alpha^2 \mu^T\mu + \alpha \text{tr}(\Sigma)
\end{split}
\end{equation}
where we define
$$
\mu = \int_{\mathbb{R}^p_+} w \rho(dw_1,\ldots,dw_p),\,\,\,
 \Sigma =  \int_{\mathbb{R}^p_+}  w w^T \rho(dw_1,\ldots,dw_p).
 $$

Let $w_i=(w_{i1},\ldots,w_{ip})$. We have, using the extended Slivnyak-Mecke theorem~\citep[Theorem 3.3]{Moller2003}
\begin{equation}
\begin{split}
\mathbb{E}[N_\alpha^{(e)}]&= \mathbb{E}\left[ \mathbb{E}[N^{(e)}_{\alpha}| W_1,\ldots,W_p ]\right] \\
&=  \mathbb{E}\left[   \sum_i  \1{\theta_i \leq \alpha}\left[\left(1-e^{-w_i^Tw_i}\right) + \frac{1}{2}\sum_{j\neq i} \1{\theta_j \leq \alpha}\left(1-e^{-2w_i^Tw_j}\right)\right] \right] \\
&= \alpha \int_{\mathbb{R}^p_+}   \left (1-e^{-w^Tw}\right )  \rho(dw_{1},\ldots,dw_p) +  \frac{\alpha^2}{2} \int_{\mathbb{R}^p_+} \psi(2w_1,\ldots,2w_p)  \rho(dw_{1},\ldots,dw_p).
\end{split}
\end{equation}

Using the extended Slivnyak-Mecke theorem then Campbell's theorem,
\begin{equation}
 \begin{split}
\mathbb{E}[N_\alpha]&=  \mathbb{E} \left[ \mathbb{E} \left[ N_\alpha |W_1,\dots,W_p \right] \right] \\
&=  \mathbb{E} \left[\sum_i \left( 1-   e^{-2w_i^T (\sum_{j\neq i} w_j 1_{\theta_j \leq \alpha}) -w_i^T w_i } \right) \1{\theta_i \leq \alpha} \right]  \\
%&\overset{\text{Campbell}}{=} \mathbb{E}  \left[ \alpha \int   1-   e^{-2w^T \sum_j 1_{\theta_j \leq \alpha} -w^Tw }  \rho(dw) \right] \\
&= \alpha \int_{\mathbb{R}_+^p}  \mathbb{E} \left( 1-   e^{-2w^T (\sum_j w_j \1{\theta_j \leq \alpha}) -w^Tw } \right) \rho(dw_{1},\ldots,dw_p) \\
%&= \alpha \int   \left( 1-   \mathbb{E}[ e^{-2w^T \sum_j 1_{\theta_j \leq \alpha} -w^Tw }] \right) \rho(dw) \\
&= \alpha
\int_{\mathbb{R}_+^p} \left(1-e^{-w^Tw - \alpha \psi(2w_1,\ldots,2w_p)}\right) \rho(dw_1,\ldots,dw_p )
\end{split}
\label{mean_nodes_1}
\end{equation}

By monotone convergence, we have, as $\alpha$ tends to infinity,
$$
\mathbb{E}[N_\alpha]\sim \alpha \int_{\mathbb{R}_+^p} \rho(dw_{1},\ldots,w_p)
$$
if the CRM is finite-activity and $
\mathbb{E}[N_\alpha]=\omega (\alpha)$
otherwise.

\section{Simulation from a tilted truncated generalized gamma process}
\label{sec:app:simujumps}

We want to sample points from a Poisson process with {truncated} mean measure%
\begin{equation}
\rho^{\varepsilon}(dw)=h(w) w^{-1-\sigma}e^{-\tau w} {\mathds 1}_{w>\varepsilon}dw
\end{equation}
where $h$ is a monotone decreasing and bounded function, and $(\tau,\sigma)$ verify either $\tau\geq 0$ and $\sigma\in(0,1)$, or $\tau>0$ and $\sigma\in(-1,0]$. We will resort to adaptive thinning~\citep{Lewis1979,Ogata1981,Favaro2013}.

For $\tau>0$, consider the family of adaptive bounds%
\[
g_{t}(s)=h(t)t^{-1-\sigma}%
\exp(-\tau s)
\]
with $g_{t}(s)>\rho(s)$ for $s>t$. We have,
\begin{align*}
G_{t}(s)  &  =\int_{t}^{s}g_{t}(s^{\prime})ds^{\prime}\\
&  =\frac{h(t)}{\tau}t^{-1-\sigma}%
(\exp(-\tau t)-\exp(-\tau s))
\end{align*}
and
\begin{align*}
G_{t}^{-1}(r)  &  =-\frac{1}{\tau}\log\left(  \exp(-\tau t)-\frac{r\tau}{ t^{-1-\sigma}h(t)}\right).
\end{align*}
For $\tau=0$, we consider bounds
\[
g_{t}(s)=h(t)s^{-1-\sigma}%
\]
and we obtain
\begin{align*}
G_{t}(s)   &=\frac{h(t)}{\sigma}(t^{-\sigma
}-s^{-\sigma})\\
G_{t}^{-1}(r)&=\left[  t^{-\sigma}-\frac{r\sigma}{h(t)}\right]  ^{-1/\sigma}.%
\end{align*}%

The adaptive thinning sampling scheme is as follows:

\begin{enumerate}
\item Set $N=\emptyset,t=\varepsilon$

\item iterate until termination

\begin{enumerate}
\item Draw $r\sim \text{Exp}(1)$

\item If $r>G_{t}(\infty)$, terminate; else set $t^{\prime}=G_{t}^{-1}(r);$

\item with probability $\rho^{\varepsilon}(t^{\prime})/g_{t}(t^{\prime})$
accept sample {$t^{\prime}$} and set $N=N\cup\{t^{\prime}\}$

\item set $t=t^{\prime}$ and continue
\end{enumerate}

\item Return $N$ a draw from the Poisson random measure with intensity
$\rho^{\varepsilon}$ on $[\varepsilon,+\infty)$
\end{enumerate}

The efficiency of this approach depends on the acceptance probability, which is given, for $\tau>0$, by
\[
\frac{\rho^{\varepsilon}(s)}{g_{t}(s)}=\frac{h(s)s^{-1-\sigma}%
}{h(t)t^{-1-\sigma}}<1%
\]
for $s>t$.

\section{Bipartite networks}

\label{sec:app:bipartite}

It is possible to use a construction similar to that of Section~\ref{sec:model} to model bipartite graphs, and extend the model of \cite{Caron2012}. A bipartite graph is a graph with two types of nodes, where only connections between nodes of different types are allowed. Nodes of the first type are embedded at locations $\theta_i\in\mathbb R_+$, and nodes of the second type at location $\theta_j^\prime\in\mathbb R_+$. The bipartite graph will be represented by a (non-symmetric) point process
\begin{equation}
Z=\sum_{i,j} z_{ij}\delta_{(\theta_i,\theta_j^\prime)}\label{eq:pointprocessZbipartite}
\end{equation}
where $z_{ij}=1$ if there is an edge between node $i$ of type 1 and node $j$ of type 2.

\paragraph{Statistical Model.} We consider the model
\begin{align*}
W_{1},\ldots,W_{p}  &  \sim\CRM(\rho,\lambda)\\
W_{1}^{\prime},\ldots,W_{p}^{\prime}  &  \sim\CRM(\rho^\prime,\lambda)\\
\intertext{and for $k=1,\ldots,p$,}
D_k|W_{k},W_{k}^{\prime}  &  \sim\text{Poisson}\left( W_{k}\times W_{k}^{\prime}\right) \\
D_k  &  =\sum_{i,j}n_{ijk}\delta_{(\theta_{i}%
,\theta_{j}^\prime)}
\end{align*}
and $z_{ij}=\min(1,\sum_{k=1}^p n_{ijk})$.

\paragraph{{Posterior i}nference.} We derive here the inference algorithm when $(W_1,\ldots,W_p)$ and $(W_{1}^{\prime},\ldots,W_{p}^{\prime})$ are compound CRMs with {$F$} and $\rho_0$ taking the form \eqref{eq:Fproductgamma} and \eqref{eq:levyGGP}.

Assume that we observe a set of connections $z=(z_{ij})_{i=1,\ldots,N_{\alpha};j=1,\ldots N_{\alpha}^{\prime}}$. We introduce latent variables $n_{ijk}$, for ${1}\leq i\leq N_\alpha$, ${1}\leq {j}\leq N^{\prime}_\alpha$, $k=1,\ldots,p$,
\[
(n_{ij1},\ldots,n_{ijp})|w,w^{\prime},z\sim\left\{
\begin{array}
[c]{ll}%
\delta_{(0,\ldots,0)} & \text{if }z_{ij}=0\\
\tPoi(w_{i1}w_{j1}^{\prime},\ldots,w_{ip}w_{jp}^{\prime}) & \text{if
}z_{ij}=1
\end{array}
\right. .
\]

We want to approximate
\[
p((w_{10},\ldots w_{N_{\alpha}0}),(\beta_{1k},\ldots,\beta_{N_{\alpha}%
k},w_{\ast k})_{k=1,\ldots,p},(w_{10}^{\prime},\ldots,w_{N_{\alpha}^{{\prime}}0}^{\prime
}),(\beta_{1k}^{\prime},\ldots,\beta_{N_{\alpha}^{{\prime}}k}^{\prime},w_{\ast k}%
^{\prime})_{k=1,\ldots,p},\phi,\alpha,\phi^{\prime},\alpha^\prime|z)
\]

{Denote $m_{ik}=\sum_{j=1}^{N_{\alpha}^{\prime}}n_{ijk}$ and $m_{i}=\sum_{k=1}^{p}m_{ik}$.} The MCMC algorithm iterates as follows:

\begin{enumerate}
\item Update $(\alpha,\phi){\vert\text{rest}}$ using a Metropolis-Hastings step{.}

\item Update
\[
w_{i0}|\text{rest}\sim\Gam\left(  m_{i}-\sigma,\tau+\sum_{k=1}%
^{p}\beta_{ik}\left[  \gamma_{k} + \left(  \sum_{j=1}^{N_{\alpha}^{\prime}%
}w_{jk}^{\prime}\right)  +w_{\ast k}^{\prime}\right]  \right){.}
\]

\item Update
\[
\beta_{ik}|\text{rest}\sim\Gam\left(  a_{k}+m_{ik},b_{k}+w_{i0}\left[
\gamma_{k} + \left(  \sum_{j=1}^{N_{\alpha}^{\prime}}w_{jk}^{\prime}\right)
+w_{\ast k}^{\prime}\right]  \right){.}
\]

\item Update $(w_{\ast1},\ldots,w_{\ast p})$%
%TCIMACRO{\TEXTsymbol{\vert}}%
%BeginExpansion
$\vert$%
%EndExpansion
rest{.}

\item Update the latent variables $n_{ijk}{\vert\text{rest}}${.}
\item Repeat steps 1-4 to update $(\alpha^{\prime}%
,\phi^{\prime})$, $(w_{10}^{\prime},\ldots,w_{N_{\alpha}^{\prime}0}^{\prime}%
)$, $(\beta_{1k}^{\prime}%
,\ldots,\beta_{N_{\alpha}^{\prime}k}^{\prime})_{k=1,\ldots,p}$ and $(w_{\ast
1}^{\prime},\ldots,w_{\ast p}^{\prime})$.
\end{enumerate}

\section{Gaussian approximation of the sum of small jumps}
\label{sec:app:gaussianapprox}

\begin{theorem}
\label{th:gaussianapprox}
Consider the multivariate random variable $X_{\varepsilon}\in\mathbb{R}%
_{+}^{p}$ with moment generating function%
\[
\mathbb{E}[e^{-t^{T}X_{\varepsilon}}]=\exp\left[  -\alpha\int_{\mathbb{R}%
_{+}^{p}}\left(  1-e^{-\sum_{k=1}^{p}t_{k}w_{k}}\right)  \rho_{\varepsilon
}(dw_{1},\ldots,dw_{p})\right]
\]
where $\alpha>0$ and%
\[
\rho_{\varepsilon}(dw_{1},\ldots,dw_{p})=e^{-\sum_{k=1}^{p}\gamma_{k}w_{k}}%
\int_{0}^{\varepsilon}w_{0}^{-p}F\left(  \frac{dw_{1}}{w_{0}},\ldots
,\frac{dw_{p}}{w_{0}}\right)  \rho_{0}(dw_{0})
\]

\end{theorem}
with $\varepsilon>0$, $\rho_0$ is a L\'evy measure on $\mathbb R_+$ and $F$ is a probability distribution on $\mathbb{R}%
_{+}^{p}$ with {density} $f$ verifying%
\begin{align*}
\int_{0}^{\infty}f(zu_{1},\ldots,zu_{p})dz  &  >0\text{ }U\text{-almost
everywhere}\\
\int_{\mathbb{R}_{+}^{p}}\left\Vert \beta_{1:p}\right\Vert ^{2}f(\beta
_{1},\ldots,\beta_{p})d\beta_{1:p}  &  <\infty
\end{align*}
where $U$ is the uniform distribution on the unit sphere $S^{p-1}$. Then if
$\rho_{0}$ is a regularly varying L\'{e}vy measure with exponent $\sigma
\in(0,1)$, i.e.%
\[
\int_{x}^{\infty}\rho_{0}(dw_{0})\overset{x\downarrow0}{\sim}x^{-\sigma}%
\ell(1/x)
\]
where $\ell:(0,\infty)\rightarrow(0,\infty)$ is a slowly varying function then%
\[
\Sigma_{\varepsilon}^{-1/2}(X_{\varepsilon}-\mu_{\varepsilon}%
)\overset{d}{\rightarrow}\mathcal{N}(0,I_{p})
\]
as $\varepsilon\rightarrow0$, where%
\begin{align*}
\mu_{\varepsilon}  &  =\alpha\int_{\mathbb{R}_{+}^{p}}w\rho_{\varepsilon
}(dw_{1},\ldots,dw_{p})\\
\Sigma_{\varepsilon}  &  =\alpha\int_{\mathbb{R}_{+}^{p}}ww^{T}\rho
_{\varepsilon}(dw_{1},\ldots,dw_{p})
\end{align*}
with%
\begin{align*}
\mu_{\varepsilon}  &  \sim\alpha\mathbb{E}[\beta]\frac{\sigma}{1-\sigma
}\varepsilon^{1-\sigma}\ell(1/\varepsilon)\text{ as }\varepsilon\rightarrow0\\
\Sigma_{\varepsilon}  &  \sim\alpha\mathbb{E}[\beta\beta^{T}]\frac{\sigma
}{2-\sigma}\varepsilon^{2-\sigma}\ell(1/\varepsilon)\text{ as }\varepsilon
\rightarrow0
\end{align*}
where $\beta$ {is distributed from} $F$.

\begin{proof}
We write the model in spherical form. Let $r=\sqrt{\sum w_{k}^{2}}$ and
$u_{k}=\frac{w_{k}}{r}$ for $k=1,\ldots,p-1$. The determinant of the Jacobian
is $\frac{r^{p-1}}{\sqrt{1-\sum_{k=1}^{p-1}u_{k}^{2}}}$ and so%
\begin{align*}
\widetilde{\rho}_{\varepsilon}(r,u_{1},\ldots,u_{p-1})  &  =\frac{r^{p-1}%
}{u_{p}}e^{-r\sum_{k=1}^{p}\gamma_{k}u_{k}}\int_{0}^{\varepsilon}{w_{0}^{-p}}f\left(  \frac{ru_{1}}{w_{0}},\ldots,\frac{ru_{p}}{w_{0}%
}\right)  \rho_{0}(dw_{0})drdu_{1:p-1}\\
&  :=\mu_{\varepsilon}(dr|u_{1:p-1})U(du_{1:p-1})
\end{align*}
where $u_{p}=\sqrt{1-\sum_{k=1}^{p-1}u_{k}^{2}}$, $\mu_{\varepsilon
}(dr|u)=r^{p-1}e^{-r\sum_{k=1}^{p}\gamma_{k}u_{k}}\int_{0}^{\varepsilon}%
{w_{0}^{-p}}f\left(  \frac{ru_{1}}{w_{0}},\ldots,\frac{ru_{p}}{w_{0}%
}\right)  \rho_{0}(dw_{0})dr$ and $U(du)=\frac{1}{u_{p}}du_{1:p}$ is the
uniform distribution on the unit sphere $S^{p-1}$.

In order to apply Theorem 2.4 in \cite{Cohen2007} \citep[see also][]{Asmussen2001}, we need to show that there exists a function
$b_{\varepsilon}:(0,1]\rightarrow(0,+\infty)$ such that%
\begin{equation}
\lim_{\varepsilon\rightarrow0}\frac{\sigma_{\varepsilon}(u)}{b_{\varepsilon}%
}>0\text{, }U\text{-almost everywhere} \label{eq:cohenprop1}%
\end{equation}
where
\[
\sigma_{\varepsilon}^{2}(u)=\int_{0}^{\infty}r^{2}\mu_{\varepsilon}(dr|u)
\]
and for every $\kappa>\varepsilon$%
\begin{equation}
\lim_{\varepsilon\rightarrow0}\frac{1}{b_{\varepsilon}^{2}}\int_{\left\Vert
w_{1:p}\right\Vert >\kappa b_{\varepsilon}}\left\Vert w_{1:p}\right\Vert
^{2}\rho_{\varepsilon}(dw_{1},\ldots,dw_{p})=0 \label{eq:cohenprop2}%
\end{equation}

Assume that $\int_{0}^{\infty}f\left(  zu_{1},\ldots,zu_{p}\right)  dz>0$
$U$-almost everywhere. With the change of variable $z=\frac{r}{w_{0}}$, and
the dominated convergence theorem we obtain%
\begin{align*}
\sigma_{\varepsilon}^{2}(u)  &  =\int_{0}^{\infty}z^{p+1}f\left(
zu_{1},\ldots,zu_{p}\right)  \left[  \int_{0}^{\varepsilon}e^{-zw_{0}%
\sum_{k=1}^{p}\gamma_{k}u_{k}}w_{0}^{2}\rho_{0}(dw_{0})\right]  dz\\
&  \sim\left(  \int_{0}^{\infty}z^{p+1}f\left(  zu_{1},\ldots,zu_{p}\right)
dz\right)  \left(  \int_{0}^{\varepsilon}w_{0}^{2}\rho_{0}(dw_{0})\right)
\text{ as }\varepsilon\rightarrow0\\
&  \sim\left(  \int_{0}^{\infty}z^{p+1}f\left(  zu_{1},\ldots,zu_{p}\right)
dz\right)  \frac{\sigma}{2-\sigma}\varepsilon^{2-\sigma}\ell(1/\varepsilon
)\text{ as }\varepsilon\rightarrow0
\end{align*}

Let $b_{\varepsilon}=\varepsilon^{1-\sigma/2}\sqrt{\ell(1/\varepsilon)}$, we
have
\begin{equation}
\lim_{\varepsilon\rightarrow0}\frac{\sigma_{\varepsilon}^{2}(u)}%
{b_{\varepsilon}^{2}}=\left(  \int_{0}^{\infty}z^{p+1}f\left(  zu_{1}%
,\ldots,zu_{p}\right)  dz\right)  \frac{\sigma}{2-\sigma}>0\text{,
}U\text{-almost everywhere} \label{eq:cohenprop1b}%
\end{equation}

Now consider, for any $\kappa>0$,%
\begin{align*}
I_{\varepsilon}  &  =\int_{\left\Vert w_{1:p}\right\Vert >\kappa
b_{\varepsilon}}\left\Vert w_{1:p}\right\Vert ^{2}\nu_{\varepsilon}%
(dw_{1},\ldots,dw_{p})\\
&  =\int_{0}^{\varepsilon}\int_{\left\Vert \beta_{1:p}\right\Vert
>\frac{\kappa b_{\varepsilon}}{w_{0}}}w_{0}^{2}\left\Vert \beta_{1:p}%
\right\Vert ^{2}e^{-w_{0}\sum_{k=1}^{p}\gamma_{k}\beta_{k}}f\left(  \beta
_{1},\ldots,\beta_{k}\right)  \rho_{0}(dw_{0})d\beta_{1:p}%
\end{align*}

For $w_{0}\in(0,\varepsilon)$, we have $\frac{\kappa b_{\varepsilon}}{w_{0}%
}\geq\frac{\kappa b_{\varepsilon}}{\varepsilon}=\varepsilon^{-\sigma/2}%
\ell(1/\varepsilon)>\kappa_{2}\varepsilon^{-\sigma/4}$ for $\varepsilon$ small
enough as $t^{\delta}\ell(t)\rightarrow0$ for any $\delta>0$ as $t\rightarrow
\infty$. So for $\varepsilon$ small enough%
\begin{align*}
I_{\varepsilon}  &  >\int_{0}^{\varepsilon}\int_{\left\Vert \beta
_{1:p}\right\Vert >\kappa_{2}\varepsilon^{-\sigma/4}}w_{0}^{2}\left\Vert
\beta_{1:p}\right\Vert ^{2}e^{-w_{0}\sum_{k=1}^{p}\gamma_{k}\beta_{k}}f\left(
\beta_{1},\ldots,\beta_{k}\right)  \rho_{0}(dw_{0})d\beta_{1:p}\\
&  >\left[  \int_{\left\Vert \beta_{1:p}\right\Vert >\kappa_{2}\varepsilon
^{-\sigma/4}}\left\Vert \beta_{1:p}\right\Vert ^{2}f\left(  \beta_{1}%
,\ldots,\beta_{k}\right)  d\beta_{1:p}\right]  \left[  \int_{0}^{\varepsilon
}w_{0}^{2}\rho_{0}(dw_{0})\right]
\end{align*}
As $\left[  \int_{0}^{\varepsilon}w_{0}^{2}\rho_{0}(dw_{0})\right]  \sim
\frac{\sigma}{2-\sigma}b_{\varepsilon}^{2}$ when $\varepsilon\rightarrow0$, we
conclude that
\begin{equation}
\lim_{\varepsilon\rightarrow0}I_{\varepsilon}=\lim_{\varepsilon\rightarrow
0}\frac{\sigma}{2-\sigma}\int_{\left\Vert \beta_{1:p}\right\Vert >\kappa
_{2}\varepsilon^{-\sigma/4}}\left\Vert \beta_{1:p}\right\Vert ^{2}f\left(
\beta_{1},\ldots,\beta_{k}\right)  d\beta_{1:p}=0 \label{eq:cohenprop2b}%
\end{equation}
Equations~\eqref{eq:cohenprop1b} and \eqref{eq:cohenprop2b} with Theorem 2.4
of \cite{Cohen2007} yield
\[
\Sigma_{\varepsilon}^{-1/2}(X_{\varepsilon}-\mu_{\varepsilon}%
)\overset{d}{\rightarrow}\mathcal{N}(0,I_{p})
\]
as $\varepsilon\rightarrow0$, where%
\begin{align*}
\mu_{\varepsilon}  &  =\alpha\int_{\mathbb{R}_{+}^{p}}w_{1:p}\ \rho_{\varepsilon
}(dw_{1},\ldots,dw_{p})\\
&  =\alpha\int_{_{\mathbb{R}_{+}^{p}}}\int_{0}^{\varepsilon}w_{0}\beta_{1:p}\text{
}e^{-w_{0}\sum_{k=1}^{p}\gamma_{k}\beta_{k}}\rho_{0}(dw_{0})f(\beta_{1},\ldots,\beta_{p}%
)d\beta_{1:p}\\
&  \sim\alpha\mathbb{E}[\beta_{1:p}]\frac{\sigma}{1-\sigma}\varepsilon^{1-\sigma
}\ell(1/\varepsilon)\text{ as }\varepsilon\rightarrow0
\end{align*}
and%
\begin{align*}
\Sigma_{\varepsilon}  &  =\alpha\int_{\mathbb{R}_{+}^{p}}w_{1:p}w_{1:p}^{T}\ \rho
_{\varepsilon}(dw_{1},\ldots,dw_{p})\\
&  \sim\alpha\mathbb{E}[\beta_{1:p}\beta_{1:p}^{T}]\frac{\sigma}{2-\sigma}\varepsilon
^{2-\sigma}\ell(1/\varepsilon)\text{ as }\varepsilon\rightarrow0
\end{align*}

using the dominated convergence theorem and lemmas \ref{prop:asymptrho2} and \ref{th:lemmagnedin}.
\end{proof}

\section{Technical lemmas}
\label{sec:app:multivariateregularvariation}

\begin{proposition}\label{th:regularvar}
Let $\nu$ be a L\'{e}vy measure defined by Eq.~\eqref{eq:nu} and \eqref{eq:rho} and $\psi$ be its
multivariate Laplace exponent. Assume that $\overline{\rho}_{0}$ is a
regularly varying function with exponent $\sigma\in(0,1)$:
\begin{equation}
\overline{\rho}_{0}\overset{x\downarrow 0}{\sim} x^{-\sigma}\ell(1/x)
\end{equation}
Then $\psi$ is (multivariate) regularly varying~\citep{Resnick2013}, with exponent $\sigma.$ More
precisely, for any $(x_{1},\ldots x_{p})\in(0,\infty)^{p}$, we have%
\begin{align*}
\psi(tx_{1},\ldots,tx_{p})  &  =\int_{{\mathbb R^p_+}}\left(  1-e^{-t\sum
_{k=1}^{p}x_{k}w_{k}}\right)  \nu(dw_{1},\ldots,dw_{p})\\
&  \overset{t\uparrow
\infty}{\sim} t^{\sigma}\Gamma(1-\sigma)\ell(t)\mathbb{E}\left[  \left(  \sum
_{k=1}^{p}x_{k}\beta_{k}\right)  ^{\sigma}\right].
\end{align*}

\end{proposition}

\begin{proof}%
\begin{align*}
\psi(tx_{1},\ldots,tx_{p})  &  =\int_{{\mathbb R^p_+}}\left(  1-e^{-t\sum
_{k=1}^{p}x_{k}w_{k}}\right)  \nu(dw_{1},\ldots,dw_{p})\\
&  =\int_{{\mathbb R^p_+}}\left(  1-e^{-t\sum_{k=1}^{p}x_{k}w_{k}}\right)
\nu(dw_{1},\ldots,dw_{p})\\
&  =\int_{{\mathbb R^p_+}}f(\beta_{1},\ldots,\beta_{p})\left[  \int_{0}%
^{\infty}\left(  1-e^{-w_{0}t\sum_{k=1}^{p}x_{k}\beta_{k}}\right)
e^{-w_{0}\sum_{k=1}^{p}\gamma_{k}\beta_{k}}\rho_{0}(dw_{0})\right]
d\beta_{1:p}%
\end{align*}
which gives, using Lemmas \ref{prop:asymptrho2}, \ref{th:lemmagnedin}, and the
dominated convergence theorem%
\[
\psi(tx_{1},\ldots,tx_{p})\overset{t\uparrow\infty}{\sim} t^{\sigma}\Gamma(1-\sigma)\ell(t)\int%
_{(0,\infty)^{p}}\left(  \sum_{k=1}^{p}x_{k}\beta_{k}\right)  ^{\sigma}%
f(\beta_{1},\ldots,\beta_{p})d\beta_{1:p}.
\]

\end{proof}

\begin{lemma}
\label{prop:asymptrho2}If%
\[
\int_{x}^{\infty}\rho(dw)\overset{x\downarrow0}{\sim}x^{-\sigma}\ell(1/x)
\]
Then
\[
\int_{x}^{\infty}e^{-cw}\rho(dw)\overset{x\downarrow0}{\sim}x^{-\sigma}%
\ell(1/x)
\]

\end{lemma}

\begin{proof}%
\begin{align*}
\int_{x}^{\infty}e^{-cw}\rho(dw)  &  =\int_{x}^{\infty}\rho(dw)-\int%
_{x}^{\infty}(1-e^{-cw})\rho(dw)\\
&  \overset{x\downarrow0}{\sim}x^{-\sigma}\ell(1/x)
\end{align*}
as $\int_{0}^{\infty}(1-e^{-cw})\rho(dw)<\infty$ for any $c>0$.
\end{proof}

\begin{lemma}
\label{th:lemmagnedin}\citep{Gnedin2007,Bingham1989}. Let $\rho$ be a L\'{e}vy measure with regularly varying tail L\'evy intensity
\begin{equation}
\int_{x}^{\infty}\rho(dw)\overset{x\downarrow0}{\sim
}x^{-\sigma}\ell(1/x) \label{eq:levytailregular}%
\end{equation}
where $\sigma\in(0,1)$ and $\ell$ is a slowly varying function (at infinity). Then
(\ref{eq:levytailregular}) is equivalent to%
\begin{align}
 \int_{0}^{x}w^{k}\rho(dw)&\overset{x\downarrow0}{\sim}\frac{\sigma}%
{k-\sigma}x^{k-\sigma}\ell(1/x)\\
\int_{0}^{\infty}(1-e^{-tw})\rho(dw)  &  \overset{t\uparrow\infty}{\sim}\Gamma(1-\sigma)t^{\sigma}%
\ell(t)
\end{align}
for any $k\geq1.$
\end{lemma}

\bibliographystyle{plainnat}
\bibliography{nnfbnpgraph}

\end{document}

%% file: figures/tikz1b.tex
\definecolor {processblue}{cmyk}{0.96,0,0,0}

\begin {tikzpicture}[-latex ,auto ,node distance =1.5 cm and 1.5cm ,on grid ,
semithick ,
state/.style ={ circle ,top color =white , bottom color = processblue!20 ,
draw,processblue , text=blue , minimum width =1 cm}]
\node[state] (C) {$\theta_1$};
\node[state] (A) [above left=of C] {$\theta_2$};
\node[state] (B) [above right =of C] {$\theta_3$};
\node[state] (D) [above right =of A] {$\theta_4$};
\path (A) edge [loop left] node[left] {$3$} (A);
\path (C) edge [loop below] node[below] {$1$} (C);
\path (D) edge [bend right =25] node[above =0.15 cm] {$1$} (A);
\path (A) edge [bend right = -15] node[below =0.15 cm] {$1$} (B);
%\path (C) edge [bend left =15] node[above =0.15 cm] {$1$} (B);
\path (B) edge [bend right = -25] node[below =0.15 cm] {$2$} (C);
\end{tikzpicture}

%% file: figures/tikz1.tex
\definecolor {processblue}{cmyk}{0.96,0,0,0}

\begin {tikzpicture}[-latex ,auto ,node distance =1.5 cm and 1.5cm ,on grid ,
semithick ,
state/.style ={ circle ,top color =white , bottom color = processblue!20 ,
draw,processblue , text=blue , minimum width =1 cm}]
\node[state] (C) {$\theta_1$};
\node[state] (A) [above left=of C] {$\theta_2$};
\node[state] (B) [above right =of C] {$\theta_3$};
\path (A) edge [loop left] node[left] {$4$} (A);
\path (C) edge [bend left =25] node[below =0.15 cm] {$2$} (A);
%\path (A) edge [bend right = -15] node[below =0.15 cm] {$1/2$} (C);
\path (C) edge [bend left =15] node[above =0.15 cm] {$1$} (B);
\path (B) edge [bend right = -25] node[below =0.15 cm] {$3$} (C);
\end{tikzpicture}

%% file: figures/tikz2.tex
\definecolor {processblue}{cmyk}{0.96,0,0,0}

\begin {tikzpicture}[-latex ,auto ,node distance =1.5 cm and 1.5cm ,on grid ,
semithick ,
state/.style ={ circle ,top color =white , bottom color = processblue!20 ,
draw,processblue , text=blue , minimum width =1 cm},every loop/.style={}]
\node[state] (C) {$\theta_1$};
\node[state] (A) [above left=of C] {$\theta_2$};
\node[state] (B) [above right =of C] {$\theta_3$};
\node[state] (D) [above right =of A] {$\theta_4$};
\node (blank) [below =of C] {};
\draw[-] (A) edge [loop left] (A);
\draw[-] (C) edge [loop below] (C);
\draw[-] (C) -- (A);
\draw[-] (C) -- (B);
\draw[-] (A) -- (D);
\draw[-] (A) -- (B);
\end{tikzpicture}